\documentclass[draftcls, onecolumn, journal]{IEEEtran}
\ifCLASSINFOpdf
\else
   \usepackage{graphicx}
\fi
\usepackage{fontenc}
\usepackage{amsfonts}
\usepackage[latin1]{inputenc}
\usepackage{amssymb}
\usepackage{amsmath}
\usepackage{bbm}
\usepackage{bm}
\usepackage{dsfont}
\usepackage{graphicx}

\usepackage{enumitem}

\usepackage[bookmarks=true]{hyperref}
\newcommand{\R}{\mathbb{R}}

\newcommand{\bra}[1]{\langle #1|}
\newcommand{\ket}[1]{|#1\rangle}
\newcommand{\braket}[2]{\langle #1|#2\rangle}
\newcommand{\myIn}{{\mathcal{X}}}
\newcommand{\myin}{{x}}
\newcommand{\mysin}{{x'}}
\newcommand{\myOut}{{\mathcal{Y}}}
\newcommand{\myout}{y}

\newcommand{\myChv}[1]{W_{#1}}
\newcommand{\myCh}[2]{W_{#1}(#2)}
\newcommand{\myChn}[2]{\bm{W}_{#1}(#2)}
\newcommand{\myChnv}[1]{\bm{W}_{#1}}
\newcommand{\Esp}{E_{\text{sp}}}
\newcommand{\Eex}{E_{\text{ex}}}
\newcommand{\Ex}{E_{\text{x}}}
\newcommand{\Pe}{\mathsf{P}_{\text{e}}}
\newcommand{\Pem}{\mathsf{P}_{\text{e}|m}}
\newcommand{\Pemax}{\mathsf{P}_{\text{e,max}}}

\DeclareMathOperator{\Tr}{Tr}
\DeclareMathOperator*{\argmin}{arg\,min}

\newtheorem{theorem}{Theorem}

\newtheorem{remark}{Remark}
\newtheorem{corollary}{Corollary}

\newcounter{forex}[section]

\begin{document}

\title
{Lower Bounds on the Probability of Error for Classical and Classical-Quantum Channels}

\author{Marco~Dalai,~\IEEEmembership{Member,~IEEE}
\thanks{M. Dalai is with the Department
of Information Engineering, University of Brescia, Italy, e-mail: marco.dalai@ing.unibs.it

This paper was presented in part at ISIT 2012 \cite{dalai-ISIT-2012}  and at ISIT 2013 \cite{dalai-ISIT-2013a}, \cite{dalai-ISIT-2013b}.

}
}

\maketitle

\begin{abstract}

In this paper, lower bounds on error probability in coding for discrete classical and classical-quantum channels are studied. The contribution of the paper goes in two main directions: i) extending classical bounds of Shannon, Gallager and Berlekamp to classical-quantum channels, and ii) proposing a new framework for lower bounding the probability of error of channels with a zero-error capacity in the low rate region. 
The relation between these two problems is revealed by showing that Lov\'asz' bound on zero-error capacity emerges as a natural consequence of the sphere packing bound once we move to the more general context of classical-quantum channels.
A variation of Lov\'asz' bound is then derived to lower bound the probability of error in the low rate region by means of auxiliary channels. As a result of this study, connections between the Lov\'asz theta function, the expurgated bound of Gallager, the cutoff rate of a classical channel and  the sphere packing bound for classical-quantum channels are established.
\end{abstract}


\section{Introduction}
\label{sec:Intro}
This paper touches some topics in sub-fields of information theory that are usually of interest to different communities. For this reason, we introduce this work with an overview of the different contexts. 

\subsection{Classical Context}
\label{sec:Intro-classical}
One of the central topics in coding theory is the problem of bounding the probability of error of optimal codes for communication over a given channel. In his 1948 landmark paper \cite{shannon-1948}, Shannon introduced the notion of channel capacity, which  represents the largest rate at which information can be sent through the channel with vanishing probability of error. This means that, at rates strictly smaller than the capacity $C$, communication is possible with a probability of error that vanishes with increasing block-length. In the following years, an important refinement of this fundamental result was obtained in \cite{feinstein-1955}, \cite{shannon-1957}, \cite{elias-1955}. In particular, it was proved that the probability of error $\Pe$ for the optimal encoding strategy at rates below the capacity vanishes exponentially fast in the block-length, a fact that we can express as
\begin{equation*}
\Pe\approx e^{-n E(R)},
\end{equation*}
where $\Pe$  is the probability of error, $n$  is the block-length and $E(R)$  is a function of the rate $R$, called \emph{channel reliability}, which is positive (possibly infinite) for all rates  smaller than the capacity. While Shannon's theorem made the evaluation of the capacity relatively simple, determining the function $E(R)$ soon turned out to be a very difficult problem. 

As Shannon himself first observed and studied \cite{shannon-1956}, for a whole class of channels, communication is possible at sufficiently low but positive rates $R$ with probability of error precisely equal to zero, a fact that is usually described by saying that function $E(R)$ is infinite at those rates. Shannon defined the zero-error capacity $C_0$ of a channel as the supremum of all rates at which communication is possible with probability of error exactly equal to zero. This problem soon appeared as one of a radically different nature from that of determining the traditional capacity. The zero-error capacity only depends on the confusability graph of the channel, and determining its value is usually considered to be a problem of a combinatorial nature rather than of a probabilistic one \cite{korner-orlitsky-1998}. As a consequence, since $C_0$  is precisely the smallest value of $R$ for which $E(R)$ is finite, it is clear that determining the precise value of $E(R)$ for general channels is expected to be a problem of exceptional difficulty. The first bounds to $C_0$ were obtained by Shannon himself \cite{shannon-1956}. In particular, he gave the first non-trivial upper bound in the form $C_0\leq C_{FB}$, where $C_{\text{FB}}$ is the zero-error capacity with feedback\footnote{We avoid the subscript `0' in the feedback case since it is known that the ordinary capacity is not improved by feedback \cite{shannon-1956}, \cite{cover-thomas-book}.}, which he was able to determine exactly by means of a clever combinatorial approach.

In the following years, works by Fano \cite{fano-book}, Shannon, Gallager and Berlekamp \cite{shannon-gallager-berlekamp-1967-1}, \cite{shannon-gallager-berlekamp-1967-2} were devoted to the problem of bounding the  function $E(R)$ for general discrete memoryless channels. The function could be determined exactly for all rates  larger than some critical rate $R_{\text{crit}}$, but no general solution could be found for lower rates, something that was however surely expected in light of the known hardness of even determining the value $C_0$ at which $E(R)$ must diverge. An important result, based on large deviation techniques in probability theory, was the so called sphere packing upper bound, that is, the determination of a function $\Esp(R)$ such that $E(R)\leq \Esp(R)$ for all rates $R$. The smallest rate $R_\infty$ for which $\Esp(R)$ is finite is clearly also an upper bound to $C_0$, and it turned out, quite nicely, that $R_\infty=C_{\text{FB}}$  (whenever $C_{\text{FB}}>0$). So, the same bound obtained by Shannon with a combinatorial approach based on the use of feedback could also be obtained indirectly from a bound based on a probabilistic argument. We may say that upper bounds to $E(R)$ and upper bounds to $C_0$ were somehow  ``coherent''. There was however only a partial coherence. In fact, channels with the same confusability graph can have different $R_\infty$. This means that the tightest bound to $C_0$ for a given channel, obtained by minimizing $R_\infty$ over all channels with the same confusability graph, could be smaller than the rate $R_\infty$ of that channel. In that case, in a range of values of the rate $R$, no upper bound to $E(R)$ was available even though this quantity was already known to be finite.

The situation remained as such until 1979, when Lov\'asz published his ground-breaking paper \cite{lovasz-1979}. Lov\'asz obtained a new upper bound to $C_0$ based on his theta function $\vartheta$ which, among other things, allowed him to precisely determine the capacity of the pentagon, the simplest graph for which Shannon was not able to determine the capacity. Lov\'asz' interest for this problem, however, came from a purely graph theoretic context, and his approach was combinatorial in nature, \emph{apparently} very different from the probabilistic techniques previously used in channel coding theory. 

Lov\'asz' contribution is usually considered a clear indication that bounding $C_0$ is a problem that must be attacked with techniques developed under the context of combinatorics rather than under the probabilistic one.
Links between Lov\'asz' contribution and classical coding theory results (see for example \cite{schrijver-1979}) have probably not been strong enough to avoid a progressive independent development of a new research branch without further advances in the study of error probability of optimal codes at low rates. Lov\'asz' theta function was recognized as a fundamental quantity in combinatorial optimization due to its relevant algebraic properties; to date, it is usually interpreted in the context of semidefinite relaxation/programming and it is probably more used in mathematics and computer science than in information theory (see for example \cite{lovasz-2003} for more details, or \cite{bachoc-et-al-2010} for more recent developments). As an effect of this trend, no advances in bounding $E(R)$ for general channels with a zero-error capacity were made after the appearance of Lov\'asz' work. Perhaps Lov\'asz' method was so combinatorially oriented that it appeared difficult to exploit it in the probabilistic context within which bounds to $E(R)$ were usually developed.   Since 1979, thus, contrarily to what happened in the '60s, a ``gap'' exists between bounds to $E(R)$ and bounds to $C_0$.

One of the main objectives of this paper is to show that Lov\'asz' work and the sphere packing bound of Shannon, Gallager and Berlekamp rely on a similar idea, which can be described in a unified way in probabilistic terms if one moves to the more general setting of quantum probability. The right context is that of classical-quantum channels; for these channels, equivalent definitions of capacity, zero-error capacity and reliability function can be given with exact analogy with the classical case.

In this paper, we prove the sphere packing bound for classical-quantum channels and  we show that  Lov\'asz' result emerges naturally as a consequence  of this bound. In particular, we show that when the rate $R_\infty$ is minimized not just over classical channels but over classical-quantum channels with a given confusability graph, then the achieved minimum is precisely the Lov\'asz theta function. Figure \ref{fig:quant_prob} gives a pictorial representation of the resulting scenario.
\begin{figure}
\centering
\includegraphics[width=0.7\linewidth]{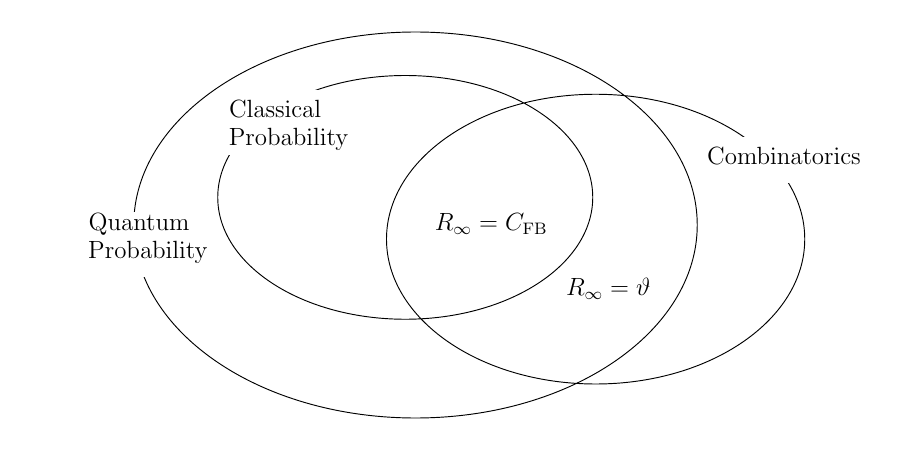}
\caption{The role of quantum probability in the study of the channel reliability in the low rate region. In the classical setting, the rate $R_\infty$ at which the sphere packing bound $E_{\text{sp}}(R)$ diverges equals the zero-error capacity with feedback (if $C_0>0$), a quantity originally studied with a combinatorial approach. In the quantum context, channels exist for which $R_\infty$ equals the Lovasz' theta function, a more powerful upper bound on $C_0$, also usually introduced in a combinatorial setting, for which no probabilistic interpretation has been given before.}
\label{fig:quant_prob}
\end{figure}
This shows that classical-quantum channels provide the right context for making bounds to $E(R)$ and bounds to $C_0$ coherent again at least to the same extent as they had been in the '60s. 

In this paper, however, we also attempt to make a first step toward a real unification of bounds to $E(R)$ and bounds to $C_0$, which means that for any channel one has a finite upper bound to $E(R)$ for each $R$ that is known to be larger than $C_0$. There are different ways of attempting such a unification. This paper focuses on an approach inspired by a common idea in Lov\'asz' construction and in the expurgated \emph{lower} bound to $E(R)$ of Gallager \cite{gallager-1965}. The resulting \emph{upper} bound to $E(R)$ is in many cases not tight at even moderately high rates. It has however the nice property of being finite over the same range of rate values for all channels with the same confusability graph, and of giving a powerful bound to $C_0$ as a consequence of bounds to $E(R)$. Furthermore, this approach reveals interesting connections between the Lov\'asz theta function, the cutoff rate of classical channels, the expurgated bound of Gallager and the rate $R_\infty$ of classical-quantum channels. The resulting situation in the general case is qualitatively depicted in Figure \ref{fig:umbrella}. The bounds obtained in this paper are simply sketched in order to clarify that we do not claim any tightness. We believe however that the presented ideas shed some light on an unexplored path that deserves further study. 

The final objective of an investigation on this topic should be a bound, that we also symbolically show in Figure \ref{fig:umbrella}, that smoothly departs from the sphere packing bound to diverge at $R=\vartheta$. We believe any result in this direction would be fundamental to coding theory and combinatorics.
\begin{figure}
\centering
\includegraphics[width=0.7\linewidth]{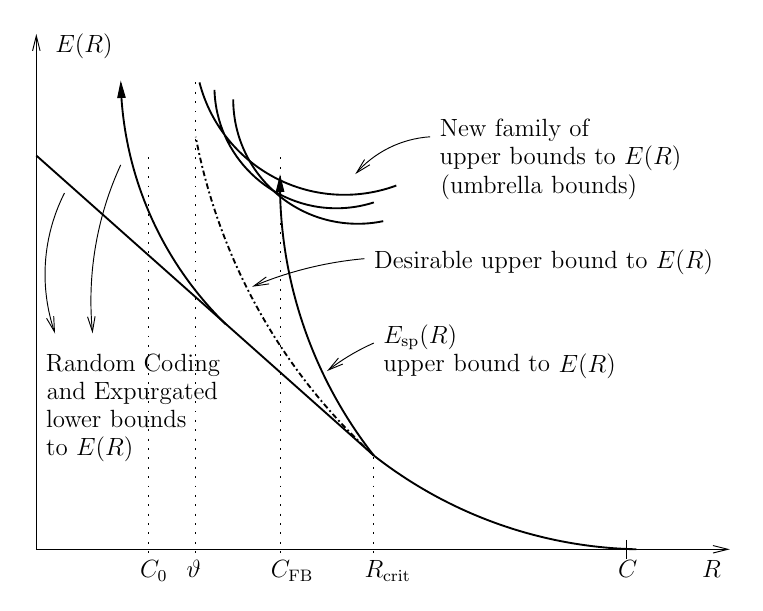}
\caption{The relation between the umbrella bounds to $E(R)$ for classical channels derived in this paper and other known quantities. A ``desirable'' upper bound to $E(R)$ is also shown in the figure that smoothly corrects the sphere packing bound to meet $\vartheta$; this should be target of future works in this direction.}
\label{fig:umbrella}
\end{figure}
When bounding the function $E(R)$ near the zero-error capacity, in fact, we are faced with the problem of finding a trade-off between probabilistic and combinatorial methods. Bounds to $E(R)$ for general non-symmetric channels usually require probabilistic methods, while effective bounds to $C_0$, like Lov\'asz' theta function, seem to require a combinatorial approach. In this paper, we suggest the use of quantum probability and quantum information theory, even for the classical problem, as a mean to expand the probabilistic setting toward the domain of combinatorics. 
 In this work, we show that by expanding the probabilistic approach of \cite{shannon-gallager-berlekamp-1967-1} for bounding $E(R)$ to the quantum case, we already recover such a powerful combinatorial result as Lov\'asz' bound to $C_0$. Thus, quantum probability represents a promising approach for finding a unified derivation of effective bounds to $E(R)$ and to $C_0$ even in the classical setting. 
 
This point of view may be considered also in light of a trend, which has emerged in recent years, which sees quantum probability methods used for the derivation of classical results (see for example the survey \cite{drucker-de-wolf-2011}). 
From the classical point of view, one may regard quantum probability as a purely mathematical tool and may want to investigate this tool without any reference to any \emph{quantum} theory. In this perspective, one may ask where these benefits of quantum probability really come from, in mathematical terms. To give a very concise partial answer to this question, one may say that while 
classical probability finds its roots in the use of distributions as vectors in the simplex, which are unit norm vectors under the $L_1$ norm, quantum probability hinges around the use of wave functions as vectors on the unit sphere, that is unit norm vectors under the $L_2$ norm (see \cite{gleason-1957} and, for example, \cite{hardy-2001} or \cite{aaronson-2004} for the reasons behind the need of using the $L_1$ or the $L_2$ norms). In coding theory, when working in the low rate region, what moves the problem from the probability domain into combinatorics is the increasing importance of very small error probability events. Optimal codewords are associated to conditional output distributions represented by almost orthogonal vectors which all lie very close to the boundary of the simplex (actually on the boundary in the case of zero-error codes).
The use of quantum probability may be seen as a form of relaxation of the classical problem where the smooth unit sphere is used in place of the probability simplex, so as to overcome the difficulties encountered when working too close to the contour. 

Lov\'asz' own result, on the other hand, can be interpreted in this way. As will be seen in Section \ref{sec:Umbrella}, the vectors used in Lov\'asz' representations essentially play the same role of the square roots of channel transition probabilities. The main difference is that Lov\'asz allows for negative components in these vectors, which have no intuitive meaning in the classical probabilistic description of the problem. The true benefit of these negative components is that, in the study of the $n$-fold channel extension, they allow to ``generate'' orthogonal codewords by exploiting what could be called, inspired by physics, ``interference'' between positive and negative components, something which cannot be done using classical probability theory.

The idea of extending the domain in order to solve a problem has in the end proved unavoidable in many different cases in mathematics. Integer valued sequences, like Fibonacci's, can usually only be represented in closed form using irrational numbers; the solutions to general cubic equations, even when all of them are real, can only  be expressed by means of radicals by resorting to the complex numbers; many integrals and series in the real domain are much easily solved in the complex plane, etc. We believe that this extension of the working domain will be highly beneficial in coding theory as well.

\subsection{Classical-Quantum Context}
\label{sec:Intro-quantum}
As already mentioned, classical-quantum channels play a central role in this paper. 
A number of results in the theory of classical communication through classical-quantum channels have been obtained in the past years that parallel many of the results obtained in the period 1948-1965 for classical channels (see \cite{holevo-1998} for a very comprehensive overview).
As in the classical case, we are here primarily concerned with the study of error exponents for optimal transmission at rates below the channel capacity. 
Upper bounds to the probability of error of optimal codes for pure-state channels were obtained by Burnashev and Holevo \cite{burnashev-holevo-1998} that are the equivalent of the so called random coding bound obtained by Fano \cite{fano-book} and  Gallager \cite{gallager-1965} and of the expurgated bound of Gallager \cite{gallager-1965} for classical channels. The expurgated bound was then extended to general quantum channels by Holevo \cite{holevo-2000}. The formal extension of the random coding bound expression to mixed states is conjectured to represent an upper bound for the general case but no proof has been obtained yet (see \cite{burnashev-holevo-1998, holevo-2000}).

A missing step in these quantum versions of the classical results is an equivalent of the sphere packing bound. This is probably due to the fact that a complete solution to the problem of determining the asymptotic error exponents in quantum hypothesis testing has been obtained only recently. In particular, the so called \emph{quantum Chernoff bound} was obtained in \cite{audenaert-2007}, for the direct part, and in \cite{nussbaum-szkola-2009}, for the converse part (both results were obtained in 2006, see \cite{audenaert-et-al-2008} for an extensive discussion). 
Those two works also essentially provided the basic tools that enabled the solution of the so called asymmetric problem in \cite{nagaoka-2006, audenaert-et-al-2008}, where the set of achievable pairs of  error exponents for the two hypotheses are determined. 
This result is often called \emph{Hoeffding  bound} in quantum statistics. 
The authors in \cite{audenaert-et-al-2008} attribute the result for the classical case also to Blahut \cite{blahut-1974} and Csisz\'ar and Longo \cite{csiszar-longo-1971}. It is the author's impression, however, that the result was already known much earlier, at least among information theorists at the MIT, since it is essentially used in Fano's 1961 book \cite{fano-book} (even if not explicitly stated in terms of hypothesis testing) and partially attributed to some 1957 unpublished seminar notes by Shannon (see also \cite{shannon-1957} for an example of Shannon's early familiarity with the Chernoff bound). A more explicit formulation in terms of binary hypothesis testing is contained in \cite{shannon-gallager-berlekamp-1967-1, shannon-gallager-berlekamp-1967-2} in a very general form, which already considers the case of distributions with different supports (compare with \cite{blahut-1974} and see for example \cite[Sec. 5.2]{audenaert-et-al-2008}) and also provides results for finite observation lengths and varying statistical distributions (check for example \cite[Th. 1, pag. 524]{shannon-gallager-berlekamp-1967-2}). This is in fact what is needed in studying error exponents for hypothesis testing between different codewords of a code for a general discrete memoryless channel.

With respect to \cite{shannon-gallager-berlekamp-1967-1, shannon-gallager-berlekamp-1967-2}, the main difference in the study of error exponents in binary hypothesis testing contained in \cite{blahut-1974} and \cite{csiszar-longo-1971} is that these papers focus more on the role of the Kullback-Leibler discrimination (or relative entropy), which can be used as a building block for the study of the whole problem in the classical case. In \cite{shannon-gallager-berlekamp-1967-1, shannon-gallager-berlekamp-1967-2}, instead, a quantity equivalent to what is now known as the \emph{R\'enyi divergence} was used as the building block. The two approaches are equivalent in the classical case, and it was historically the presentation in terms of the ubiquitous Kullback-Leibler divergence which emerged as the preferred one as opposed to the R\'enyi divergence. 
Along this same line, a simpler and elegant proof of the sphere packing bound, again in terms of the Kullback-Leibler divergence, was derived by Haroutunian \cite{haroutunian-1968} by comparing the channel under study with dummy channels with smaller capacity.
This proof, which is substantially simpler than the one presented in \cite{shannon-gallager-berlekamp-1967-1}, was then popularized in \cite{csiszar-korner-book}, and became the preferred proof for this result.

As a matter of fact, however, as pointed out in \cite[Sec. 4, Remark 1]{nagaoka-2006} and \cite[Sec. 4.8]{audenaert-et-al-2008}, it turns out that the solution to the study of error exponents in quantum hypothesis testing can be expressed in terms of the R\'enyi divergence and not in terms of the Kullback-Leibler divergence. Thus, since the sphere packing bound is essentially based on the theory of binary hypothesis testing, it is reasonable to expect that Haroutunian's approach to the sphere packing bound may fail in the quantum case. This could be in our opinion the reason why a quantum sphere packing bound has not been established in the literature.

In this paper, we propose a derivation of the sphere packing bound for classical-quantum channels by following closely the approach used in  \cite{fano-book, shannon-gallager-berlekamp-1967-1}. 
The quantum case is related to the classical one by means of the \emph{Nussbaum-Szko{\l}a mapping} \cite{nussbaum-szkola-2009}, that represented the key point in proving the converse part of the quantum Chernoff bound (see \cite{audenaert-et-al-2008} for more details). This allows us to formulate a quantum version of the Shannon-Gallager-Berlekamp generalization of the Chernoff bound \cite[Th. 5]{shannon-gallager-berlekamp-1967-1} on binary hypothesis testing (in the converse part).
The proof of the sphere packing bound used in \cite{shannon-gallager-berlekamp-1967-1} will then be adapted to obtain the equivalent bound for classical-quantum channels.  This proves the power of the methods employed in \cite{shannon-gallager-berlekamp-1967-1}. Furthermore, the mentioned generalization of the Chernoff bound allows us to adapt the technique used in \cite{shannon-gallager-berlekamp-1967-2} to find an upper bound to the reliability at $R=0$, which leads to an exact expression when combined with the expurgated bound proved by Holevo \cite{holevo-2000}.

\subsection{Paper overview}
\label{sec:Intro-overview}
This paper is structured as follows. In Section \ref{sec:Basic} we introduce the notation and the basic notions on classical and classical-quantum channels, and on the main statistical tools used in this paper. In Section \ref{sec:Umbrella} we introduce what we call the ``umbrella bound'' for classical channels in its simplest and self-contained form. This section is entirely classical, only the \emph{bra-ket} notation for scalar products is used for convenience. The scope of this section is to show how Lov\'asz' idea can be extended to bound the reliability function $E(R)$ at all rates larger than the Lov\'asz theta function. Interesting connections between the Lov\'asz theta function, the cutoff rate and the expurgated bound emerge by means of this analysis. This section represents a preview of the more general results that will be derived in Section \ref{sec:SP-umbrella} and prepares the reader for the interpretation of Lov\'asz' representations as auxiliary channels. These results were first presented in \cite{dalai-ISIT-2013b}.

We then start considering bounds for classical-quantum channels. In Section \ref{sec:Quantum-H-test} we develop a fundamental bound to the probability of error in a binary decision test between quantum states. This bound is a quantum version of the converse part of the Shannon-Gallager-Berlekamp generalization of the Chernoff bound \cite{shannon-gallager-berlekamp-1967-1}. In Section \ref{sec:SPB}, this tool is used to prove the sphere packing bound for classical-quantum channels. The proof of the bound follows the approach in \cite{shannon-gallager-berlekamp-1967-1}, which contains the key idea that is also found in Lov\'asz' bound. Part of the results presented in Sections \ref{sec:Quantum-H-test} and \ref{sec:SPB} were first presented in \cite{dalai-ISIT-2012}.
In Section \ref{sec:Inform-radii} we provide a detailed analysis of the analogy between the sphere packing and Lov\'asz' bound. In doing so, we generalize a result of Csisz\'ar \cite{csiszar-1995}, showing that the quantum sphere packing bound can be written in terms of an information radius \cite{sibson-1969}, which appears to be the leading theme that puts the Lov\'asz theta function, the cutoff rate, the rate $R_\infty$ and the ordinary capacity $C$ under the same light. In this section, it is also proved that the minimum $R_\infty$ rate of all channels with a given confusability graph is precisely the Lov\'asz theta function of that graph. Part of these results were first presented in \cite{dalai-ISIT-2013a}.
In Section \ref{sec:classical-pure}, we provide a more detailed analysis of classical and pure-state channels, establishing a connection between the cutoff rate of a classical channel and the rate $R_\infty$ of a pure-state channel that could underlie the classical one. This leads as a side result to seemingly new expressions for the cutoff rate.

In Section \ref{sec:SP-umbrella}, then, we reconsider the umbrella bound anticipated in Section \ref{sec:Umbrella}, giving it a more general and generally more powerful form that allows one to bound the reliability function $E(R)$ of a classical-quantum channel in terms of the sphere packing bound of an auxiliary classical-quantum channel. 
Finally, in Section \ref{sec:Zero-rate} we consider the special case of channels with no zero-error capacity, for which we present the quantum extension of some classical bounds. In particular, we show that the zero-rate bound of \cite{shannon-gallager-berlekamp-1967-2} can be extended to classical-quantum channels by means of the results of Section \ref{sec:Quantum-H-test}, thus obtaining the precise value of the reliability at $R=0$. 
The quantum extension of some other known classical bounds is then briefly discussed.

\section{Basic Notions and Notations}
\label{sec:Basic}
In this section, we present the choice of notation in detail while discussing the basic results on classical and classical-quantum channels, on the used divergences and on the zero-error capacity (see \cite{gallager-book}, \cite{viterbi-omura-book}, \cite{cover-thomas-book},  \cite{hayashi-book-2006}, \cite{wilde-2013}, \cite{holevo-1998}, \cite{holevo-book} for more details).

As a general rule, we use lower case letters for unit norm vectors ($u,\psi$ etc.) and capital letters for distributions ($P,Q$ etc. ) or density operators ($A,B,S$ etc.). Bold letters refer to quantities associated to $n$-fold tensor powers or products. For classical discrete memoryless channels, we use the notation $W_{x}(y)$ for the probability of output $y$ when the input is $x$. So, $W_x$ is the output distribution induced by input $x$. In a similar way, for a classical-quantum channel, $S_x$ denotes the density operator associated to input $x$. Finally, we use Dirac's \emph{bra-ket} notation for inner and outer products in Hilbert spaces but we avoid the \emph{ket} notation for vectors.

\subsection{Classical Channels}
\label{sec:Basic-classical}

Let $\myIn=\{1,2,\ldots,|\myIn |\}$ and $\myOut=\{1,2,\ldots,|\myOut|\}$ be the input and output alphabets of a discrete memoryless channel with transition probabilities 
 $\myCh{\myin}{\myout}$, $\myin\in\myIn,\myout\in\myOut$.
If $\bm{\myin}=(\myin_1,\myin_2,\ldots,\myin_n)$ is a sequence of $n$ input symbols and correspondingly  $\bm{\myout}=(\myout_1,\myout_2,\ldots,\myout_n)$ is a sequence of output symbols, then the probability of observing $\bm{\myout}$ at the output of the channel given input $\bm{\myin}$ is 
\begin{equation*}
\myChn{\bm{\myin}}{\bm{\myout}}=\prod_{i=1}^n \myCh{\myin_i}{\myout_i}.
\end{equation*}
A block code with parameters $M$ and $n$ is a mapping from a set $\{1,2,\ldots,M\}$ of $M$ messages onto a set $\{\bm{\myin}_1, \bm{\myin}_2, \ldots, \bm{\myin}_M\}$ of $M$ sequences each composed of $n$ symbols from the input alphabet. The rate $R$ of the code is defined as $R=(\log M)/n$.
A decoder is a mapping from the set of length-$n$ sequences of symbols from the output alphabet into the set of possible messages $\{1,2,\ldots,M\}$. If message $m$ is to be sent, the encoder transmits  the codeword $\bm{\myin}_m$ through the channel.
An output sequence $\bm{\myout}$ is received by the decoder, which maps it to a message $\hat{m}$. An error occurs if $\hat{m}\neq m$.

Let $\bm{\myOut}_m$ be the set of output sequences that are mapped to the message $m$. When message $m$ is sent, the probability of error is
\begin{equation*}
\Pem=1-\sum_{\bm{\myout}\in \bm{\myOut}_m} \myChn{\bm{\myin}_m}{\bm{\myout}}.
\end{equation*}
The maximum error probability of the code is defined as the largest $\Pem$, that is,
\begin{equation*}
\Pemax=\max_{m}\Pem.
\end{equation*}

Let $\Pemax^{(n)}(R)$ be the smallest maximum error probability among all codes of length $n$ and rate \emph{at least} $R$.
Shannon's theorem states that sequences of codes exists such that $\Pemax^{(n)}(R)\to 0$ as $n\to\infty$ for all rates smaller than a constant $C$, called \emph{channel capacity}, which is given by the expression
\begin{equation*}
C=\max_{P}\sum_{\myin,\myout} P(\myin) \myCh{\myin}{\myout}\log\frac{\myCh{\myin}{\myout}}{\sum_{\mysin} P(\mysin)\myCh{\mysin}{\myout}},
\end{equation*} 
where the maximum is over all probability distributions on the input alphabet (see \cite{cover-thomas-book}, \cite{csiszar-korner-book}, \cite{gallager-book} for more details on the capacity).

For $R<C$, Shannon's theorem only asserts that $P_{e,\max}^{(n)}(R)\to 0$ as $n\to\infty$.  
In the most general case\footnote{Some ``pathological'' channels, like the noiseless binary symmetric channel, can have $C=C_0$ and exhibit no exponential decrease of $\Pemax^{n}(R)$ for any $R$. These are not very interesting cases however.}, for a range of rates  $C_0<R < C$, the optimal probability of error $P_{e,\max}^{(n)}(R)$ is known to have an exponential decrease in $n$, and it is thus reasonable to define the \emph{reliability function} of the channel as
\begin{equation}
E(R)=\limsup_{n\to\infty} -\frac{1}{n}\log P_{e,\max}^{(n)}(R).
\label{eq:E(R)_def_class}
\end{equation} 
The value $C_0$ is the so called \emph{zero-error capacity}, also introduced by Shannon \cite{shannon-1956}, which is defined as the highest rate at which communication is possible with probability of error precisely equal to zero. More formally,
\begin{equation}
C_0=\sup\{R \, :\,   P_{e,\max}^{(n)}(R)=0 \mbox{ for some } n\}.
\label{eqdef:C0}
\end{equation}
For $R<C_0$, we may define the reliability function $E(R)$ as being infinite\footnote{According to the definition of $E(R)$ in equation \eqref{eq:E(R)_def_class}, which is equivalent to the definition of $E(R)$ given in \cite{shannon-gallager-berlekamp-1967-1}, the value at $R=C_0$ can be both finite or infinite depending on the channel. This is however only a technical detail that does not change the behavior of $E(R)$ at $R\neq C_0$.\label{fn:E(C0)}}.

It is known that the same function $E(R)$ results 
if in \eqref{eq:E(R)_def_class} one substitutes $\Pemax^{(n)}(R)$ with the smallest \emph{average} probability of error $\Pe^{(n)}(R)$ defined as the minimum of $\Pe=\frac{1}{M}\sum_m \Pem$ over all codes of block length $n$ and rate at least $R$
(see for example \cite{shannon-gallager-berlekamp-1967-1}, \cite{gallager-1965}).
For almost all channels, the function $E(R)$ is known only in the region of high rates. The random coding lower bound states that $E(R)\geq E_\text{r}(R)$, where
\begin{align}
E_{\text{r}}(R) & = \max_{0\leq \rho \leq 1 }\left[E_0(\rho)-\rho R \right]\label{eq:rand_Class_1}\\
E_0(\rho) & = \max_{P} E_0(\rho, P)\label{eq:rand_Class_2}\\
E_0(\rho,P) & = -\log \sum_\myout \left(\sum_\myin P(\myin) \myCh{\myin}{\myout}^{1/{(1+\rho)}}\right)^{1+\rho}.
\end{align}
This bound is tight in the high rate region. This is proved by means of the sphere packing upper bound, which states that $E(R)\leq \Esp(R)$, where
\begin{align*}
\Esp(R) & = \sup_{\rho\geq 0}\left[E_0(\rho)-\rho R \right].
\end{align*}
For those rates $R$ for which $\Esp(R)$ is achieved by a $\rho\leq 1$, we see that $\Esp(R)=E_{\text{r}}(R)$. It can be shown that there is a constant $R_{\text{crit}}$, called \emph{critical rate}, which is in most interesting cases smaller than $C$ (see \cite[Appendix]{shannon-gallager-berlekamp-1967-1}), such that $\Esp(R)=E_{\text{r}}(R)$ for all rates $R\geq R_{\text{crit}}$. So, the reliability function is known exactly in the range $R_{\text{crit}}\leq R\leq C$.

	In the low rate region, two important different cases are to be distinguished. If there is at least one pair of inputs $\myin$ and $\mysin$ such that $\myCh{\myin}{\myout}\myCh{\mysin}{\myout}=0$ $\forall \myout$, then communication is possible at sufficiently low rates with probability of error exactly equal to zero, which means that the zero-error capacity defined in \eqref{eqdef:C0} is positive. Otherwise, $C_0=0$ and the probability of error, though small, is always positive.
An improvement over the random coding bound is given by Gallager's expurgated bound \cite{gallager-1965}, which states that $E(R)\geq E_{\text{ex}}(R)$ where
\begin{align}
E_{\text{ex}}(R) & = \sup_{\rho\geq 1 }\left[E_{\text{x}}(\rho)-\rho R\right]\label{eq:class_expurgated1}\\
E_{\text{x}}(\rho)& = \max_{P} E_{\text{x}}(\rho,P)\label{eq:class_expurgated2}\\
E_{\text{x}}(\rho,P) & = 
\!-\rho\log \sum_{\myin, \mysin} \! P(\myin)P(\mysin)\!{\left(\sum_\myout \sqrt{\myCh{\myin}{\myout}\myCh{\mysin}{\myout}}\right)}^{1/\rho}.
\label{eq:class_expurgated3}
\end{align}
In the low rate region, the known upper bounds differ substantially depending on whether the channel has a zero-error capacity or not.
The function $\Esp(R)$ goes to infinity for rates $R$ smaller than the quantity
\begin{equation}
R_\infty = \max_{P} \left[-\log \max_\myout \sum_{\myin: \myCh{\myin}{\myout}>0} P(\myin)\right],
\label{eq:Rinfty_class}
\end{equation}
which is in the general case larger than $C_0$, even in cases where $C_0=0$.
No general improvement has been obtained in this low rate region over the sphere packing bound in the general case of channels with a zero-error capacity. 
For channels with no zero-error capacity, instead, a major improvement was obtained in \cite{shannon-gallager-berlekamp-1967-1, shannon-gallager-berlekamp-1967-2}, where it is proved that the expurgated bound is tight at $R=0$ and that it is possible to upper bound $E(R)$ by the so called \emph{straight line} bound, which is a segment connecting the plot of $\Eex(R)$ at $R=0$ to the function $\Esp(R)$ tangentially to the latter. Furthermore, for the specific case of the binary symmetric channel, much more powerful bounds are available \cite{mceliece-et-al-1977}, \cite{barg-mcgregor-2005}.

\subsection{Classical-Quantum Channels}
\label{sec:Basic-quantum}
Consider a classical-quantum channel with input alphabet $\myIn=\{1,\ldots,|\myIn|\}$ and associated density operators $S_\myin$, $\myin\in\myIn$, in a finite dimensional Hilbert space\footnote{The $S_\myin$ can thus be represented as positive semi-definite Hermitian matrices with unit trace.} $\mathcal{H}$. The $n$-fold product channel acts in the tensor product space $\bm{\mathcal{H}}=\mathcal{H}^{\otimes n}$ of $n$ copies of $\mathcal{H}$. To a sequence $\bm{\myin}=(\myin_1,\myin_2,\ldots,\myin_n)$ is associated the signal state $\bm{S}_{\bm{\myin}}=S_{\myin_1}\otimes S_{\myin_2}\cdots\otimes S_{\myin_n}$.
A block code with $M$ codewords is a mapping from a set of $M$ messages $\{1,\ldots,M\}$ into a set of $M$ codewords  $\bm{\myin}_1,\ldots, \bm{\myin}_M$, as in the classical case. 
The rate of the code is again $R=(\log M)/n$.

A quantum decision scheme for such a code is a so-called POVM (see for example \cite{wilde-2013}), that is, a  collection of $M$ positive operators\footnote{The operators $\Pi_m$ can thus be represented as positive semi-definite matrices. The notation $\sum \Pi_m \leq \mathds{1}$ simply means that $\mathds{1}-\sum \Pi_m $ is positive semidefinite. Note that, by construction, all the eigenvalues of each operator $\Pi_m$ must be in the interval $[0,1]$.} $\{\Pi_1,\Pi_2,\ldots,\Pi_M\}$ such that $\sum \Pi_m \leq \mathds{1}$, where $\mathds{1}$ is the identity operator.
The probability that message $m'$ is decoded when message $m$ is transmitted is $\mathsf{P}_{m'|m}=\Tr \Pi_{m'} \bm{S}_{\bm{\myin}_m}$. 
The probability of error after sending message $m$ is
\begin{equation*}
\Pem=1-\Tr\left(\Pi_m \bm{S}_{\bm{\myin}_m}\right).
\end{equation*}
We then define $\Pemax$, $\Pemax^{(n)}(R)$ and $E(R)$ precisely as in the classical case. 

A pure-state channel is one for which all density operators $S_x$ have rank one, in which case we write $S_x=\ket{\psi_x}\bra{\psi_x}$. If, on the other hand, all density operators $S_x$ commute, then they are all simultaneously diagonal in some basis. In this case it is easily proved that the optimal measurements are also diagonal in the same basis and, thus, the classical-quantum channel reduces to a classical one. Each classical channel can then be thought of as a classical-quantum one with diagonal states $S_x$, the distribution $W_x$ being the diagonal of $S_x$. 

Bounds to the reliability function of classical-quantum channels were first investigated by Burnashev and Holevo \cite{burnashev-holevo-1998} and by Holevo \cite{holevo-2000}. The \emph{formal} quantum analog of the random coding and of the expurgated exponent expressions can be written as in \eqref{eq:rand_Class_1}-\eqref{eq:rand_Class_2} and \eqref{eq:class_expurgated1}-\eqref{eq:class_expurgated2} respectively where, in this case,
\begin{align*}
E_0(\rho,P) & = -\log \Tr \left(\sum_\myin P(\myin) S_\myin^{1/(1+\rho)}\right)^{1+\rho}
\end{align*}
and
\begin{align*}
\Ex(\rho,P) & = -\rho\log \sum_{\myin,\mysin}P(\myin)P(\mysin)\left(\Tr\sqrt{S_\myin}\sqrt{S_\mysin}\right)^{1/\rho}.
\end{align*}
Operational meaning to these quantities has been given in \cite{burnashev-holevo-1998} and \cite{holevo-2000}. In particular, the random coding bound and the expurgated bound were proved for pure-state channels in \cite{burnashev-holevo-1998}, and the proof of the expurgated bound was then extended to mixed-state channels in \cite{holevo-2000}. However, no extension of the random coding bound to mixed-state channels has been obtained yet. To the best of this author's knowledge, the best currently available lower bound to $E(R)$ was obtained in \cite{hayashi-1997}.

\subsection{Zero-Error Capacity}
\label{sec:Basic-zero-error}
Both for the classical and for the classical-quantum cases, we can define the zero-error capacity $C_0$ of the channel according to equation \eqref{eqdef:C0}. 
In the classical case, if a code satisfies $\Pemax=0$, then for each pair of different codewords $\bm{\myin}_m,\bm{\myin}_{m'}$ the output distributions $\myChnv{\bm{\myin}_m}$ and $\myChnv{\bm{\myin}_{m'}}$ must have disjoint supports. This implies that, for at least one index $i$, the two codewords contain in the $i$-th position two symbols $\myin_{m,i}$ and $\myin_{m',i}$ that are not confusable, which means that $\myChv{\myin_{m,i}}$ and $\myChv{\myin_{m',i}}$ have disjoint supports. For a given channel, then, it is useful to define a \emph{confusability graph} $G$ whose vertices are the elements of $\mathcal{X}$ and whose edges are the elements $(x,x')\in\mathcal{X}^2$ such that $x$ and $x'$ are confusable. It is then easily seen that $C_0$ only depends on the confusability graph $G$. Furthermore, for any graph $G$, we can always find a channel with confusability graph $G$. Thus, we may equivalently speak of the zero-error capacity of a channel or of the \emph{capacity $C(G)$ of the graph} $G$, and we will use those two notions interchangeably through the paper.

An identical discussion holds for classical-quantum channels (see \cite{medeiros2006} for a more detailed discussion. For recent results on zero-error communication via general quantum channels, see \cite{duan-severini-winter-2013} and references therein). In fact, if a code satisfies $\Pemax=0$, then for each $m\neq m'$ we must have $\Tr(\Pi_m \bm{S}_{\bm{\myin}_m})=1$ and $\Tr(\Pi_m  \bm{S}_{\bm{\myin}_{m'}})=0$. This is possible if and only if the signals $ \bm{S}_{\bm{\myin}_m}$ and $ \bm{S}_{\bm{\myin}_{m'}}$ are orthogonal, that is $\Tr(\bm{S}_{\bm{x}_m}\bm{S}_{\bm{x}_{m'}})=0$. Using the property that $\Tr((A\otimes B)(C\otimes D))=\Tr(AC)\Tr(BD)$, we then have
\begin{align*}
\Tr(\bm{S}_{\bm{x}_m}\bm{S}_{\bm{x}_{m'}}) & =\prod_{i=1}^n\Tr(S_{x_{m,i}}S_{x_{m',i}}).
\end{align*}
This implies that $\Tr(S_{x_{m,i}}S_{x_{m',i}})=0$ for at least one value of $i$, which means that that $x_{m,i}$ and $x_{m',i}$ are not confusable. 
We see then that there is no difference with respect to the classical case: the zero error capacity only depends on the confusability graph. Given a graph $G$, we can interpret the capacity of the graph $C(G)$ as either the zero error capacity $C_0$ of a classical or of a classical-quantum channel with that confusability graph.

Finding the zero-error capacity remains an unsolved problem (see \cite{korner-orlitsky-1998} for a detailed discussion). As mentioned before, a first upper bound to $C_0$ was obtained by Shannon by means of an argument based on feedback. He could prove that the zero-error capacity, when feedback is available, is given by the expression
\begin{equation*}
C_{\text{FB}} = \max_{P}\left[ -\log \max_\myout \sum_{\myin: \myCh{\myin}{\myout}>0} P(\myin)\right]
\end{equation*}
whenever $C_0>0$. The expression above is precisely the value $R_\infty$ at which the sphere packing bound diverges (whether $C_0>0$ or not). The best bound is then obtained by using the channel (with the given confusability graph) which minimizes the above value. 

A major breakthrough was obtained by Lov\'asz in terms of his theta function \cite{lovasz-1979}. He could prove that $C_0$ is upper bounded by the quantity $\vartheta$ defined as\footnote{We use a logarithmic version of the original theta function as defined by Lov\'asz, so as to make comparisons with rates simpler.}
\begin{align*}
\vartheta& = \min_{\{u_\myin\}}\min_c \max_\myin \log \frac{1}{|\braket{u_\myin}{c}|^2},
\end{align*}
where the outer minimum is over all sets of unit-norm vectors in any Hilbert space such that $u_\myin$ and $u_\mysin$ are orthogonal if symbols $\myin$ and $\mysin$ cannot be confused, the inner minimum is over all  unit norm vectors $c$, and the maximum is over all input symbols. See also \cite{schrijver-1979} and \cite{mceliece-et-al-1978} for interesting comments on Lov\'asz' result.

Shortly afterwards, Haemers \cite{haemers-1978}, \cite{haemers-1979} obtained another interesting upper bound to $C_0$, which he proved to be strictly better than $\vartheta$ in some cases. The Haemers bound asserts that $C_0$ is upper bounded by the rank of any $|\myIn|\times |\myIn|$ matrix $A$ with elements $A(\myin,\myin)\neq 0$ and $A(\myin,\mysin)=0$ if $\myin$ and $\mysin$ are non-confusable inputs. This bound is in many cases looser than Lov{\'a}sz' one, but Haemers proved it to be tighter for the graph which is the complement of the so called \emph{Schl\"afli graph}.

Recently, an extension of the Lov\'asz theta function for the quantum communication problem has been derived in \cite{duan-severini-winter-2013} which is based on algebraic properties satisfied by the Lov\'asz theta function. There, the authors propose an extension for general quantum channels based on what they call \emph{non-commutative graphs}. In this paper, however, we always consider classical confusability graphs, and the interest is in the connection between bounds to $E(R)$ and bounds to $C_0$. Further work would be needed to understand if there is a connection between the results derived here and the results of \cite{duan-severini-winter-2013}.

\subsection{Distances and Divergences}
\label{sec:Basic-divergences}
In this paper, a fundamental role is played by statistical measures of dissimilarity between probability distributions and between density operators. This section defines the used notation and recalls the properties of those measures that will be needed in the rest of the paper.

In classical binary hypothesis testing between two probability distributions $U$ and $V$ on an alphabet $\mathcal{Z}$, a fundamental role is played by the function $\mu_{U,V}(s)$ defined by
\begin{equation*}
\mu_{U,V}(s)=\log \sum_{z\in\mathcal{Z}} U(z)^{1-s}V(z)^s,\qquad 0<s<1
\end{equation*}
and extended to $s=0,1$ by defining
\begin{equation*}
\mu_{U,V}(0)=\lim_{s\to 0}\mu_{U,V}(s)\quad \mbox{and} \quad \mu_{U,V}(1)=\lim_{s\to 1}\mu_{U,V}(s).
\end{equation*}
The minimum value of $\mu_{U,V}(s)$ in the interval $[0,1]$ is of importance for the study of symmetric binary hypothesis testing, and it is convenient to introduce the \emph{Chernoff distance} $d_{\text{C}}(U,V)$ between the two probability distributions $U$ and $V$, which is defined by
\begin{equation*}
d_{\text{C}}(U,V)=-\min_{0\leq s \leq 1}\mu_{U,V}(s).
\end{equation*}
It will also be useful later to discuss the relation between the Chernoff distance and other distance measures. Of particular importance is the \emph{Batthacharyya distance}, defined as
\begin{eqnarray*}
d_{\text{B}}(U,V) & = & -\mu_{U,V}(1/2)\\
& = & - \log \sum_z \sqrt{U(z)V(z)}.
\end{eqnarray*}
It is known (see for example \cite{shannon-gallager-berlekamp-1967-1}) that the function  $\mu_{U,V}(s)$ is a non-positive convex function of $s$ and, from this, the following inequalities are deduced
\begin{equation}
d_{\text{B}}(U,V)\leq d_{\text{C}}(U,V) \leq 2 d_{\text{B}}(U,V).
\label{eq:d_Bd_C}
\end{equation}
Examples are easily found showing that both equalities above are possible.

Another quantity\footnote{Different authors seem to adopt different definitions and notations for the R\'enyi divergence, see \cite{hayashi-book-2006}, \cite{csiszar-1995}. We adopt a notation similar to that used in \cite{shannon-gallager-berlekamp-1967-1} and in \cite{csiszar-1995}, which, although it is not particularly coherent, it is useful for the purpose of this paper and for comparison with the literature.} which is related (actually equivalent) to  $\mu_{U,V}(s)$ and which will be useful in Section \ref{sec:Inform-radii} is the R\'enyi divergence of order $\alpha$ defined as 
\begin{align}
D_\alpha(U||V) & = \frac{1}{\alpha-1}\mu_{U,V}(1-\alpha)\\
& = \frac{1}{\alpha-1}\log \sum_{z\in \mathcal{Z}} U(z)^\alpha V(z)^{1-\alpha}.
\label{eqdef:D_alpha}
\end{align}
When $\alpha\to 1$ the divergence $D_\alpha(U||V)$ tends to the Kullback-Leibler divergence defined as \cite{cover-thomas-book} 
\begin{equation*}
D(U||V)=\sum_{z\in\mathcal{Z}} U(z)\log\frac{U(z)}{V(z)}.
\end{equation*}

We now introduce the corresponding quantities for the quantum case, that is, when two density operators $A$ and $B$ are to be distinguished in place of the two distributions $U$ and $V$. The function $\mu_{A,B}(s)$ is defined by
\begin{equation}
\mu_{A,B}(s)=\log \Tr A^{1-s}B^s, \qquad 0<s<1
\label{eq:mu_quantum_def_1}
\end{equation}
and
\begin{equation}
\mu_{A,B}(0)=\lim_{s\to 0}\mu_{A,B}(s)\quad \mbox{and} \quad \mu_{A,B}(1)=\lim_{s\to 1}\mu_{A,B}(s).
\label{eq:mu_quantum_def_2}
\end{equation}
The Chernoff and the Bhattacharrya distances $d_{\text{C}}(A,B)$ and $d_{\text{B}}(A,B)$  between the two density operators $A$ and $B$ are then defined as in the classical case by
\begin{equation*}
d_{\text{C}}(A,B)=-\min_{0\leq s \leq 1}\mu_{A,B}(s),
\end{equation*}
and 
\begin{eqnarray*}
d_{\text{B}}(A,B) & = & -\mu_{A,B}(1/2)\\
& = & - \log \Tr \sqrt{A}\sqrt{B},
\end{eqnarray*}
and they are again related by the inequalities 
\begin{equation*}
d_{\text{B}}(A,B)\leq d_{\text{C}}(A,B) \leq 2 d_{\text{B}}(A,B).
\end{equation*}
In the quantum case, however, another important measure of the difference between two quantum states is the so called \emph{fidelity} between the two states, which is given by the expression\footnote{We use here the usual notation $|A|=\sqrt{A^*A}$.} $\Tr |\sqrt{A}\sqrt{B}|$. Here, it will be useful for us to adopt a logarithmic measure and define\footnote{Usually the quantity $2(1-\Tr|\sqrt{A}\sqrt{B}|)$  is called \emph{Bures distance}. We use the notation $d_{\text{F}}$, with $\text{F}$ for fidelity, to avoid ambiguities with the Bhattacharyya distance.}
\begin{eqnarray*}
d_{\text{F}}(A,B) & = & -\log \Tr |\sqrt{A}\sqrt{B}|\\
& = & - \log \Tr \sqrt{\sqrt{A}\,B \sqrt{A}}.
\end{eqnarray*}
It is known that  $d_{\text{F}}(A,B)$, is related to  $d_{\text{C}}(A,B)$ and $d_{\text{B}}(A,B)$ by the following inequalities
\begin{equation*}
d_{\text{F}}(A,B)\leq d_{\text{B}}(A,B)\leq d_{\text{C}}(A,B) \leq 2 d_{\text{F}}(A,B)\leq 2 d_{\text{B}}(A,B).
\end{equation*}
Both the conditions $d_{\text{C}}(A,B)=d_{\text{F}}(A,B)$ and $d_{\text{C}}(A,B)=2d_{\text{B}}(A,B)$ are possible for properly chosen density operators $A$ and $B$.
 Finally, the R\'enyi divergence of order $\alpha$ can be defined as
\begin{equation}
D_\alpha(A||B)=\frac{1}{\alpha-1}\log \Tr A^\alpha B^{1-\alpha}.
\label{eqdef: Quantum_D_rho}
\end{equation}

\section{A Preview: an ``umbrella'' bound}
\label{sec:Umbrella}

In this section, we present an upper bound to $E(R)$ for classical channels that can be interpreted as an extension of Lov\'asz' work in the direction of giving at least a crude upper bound to $E(R)$ for those rates that Lov\'asz' own bound proves to be strictly larger than the zero-error capacity. However, the intent is to obtain Lov\'asz' bound as a consequence of an upper bound on $E(R)$ and not vice versa.
The obtained bound on $E(R)$ is loose at high rates, but it has two important merits. First, it makes immediately clear how Lov\'asz' idea can be extended to find an upper bound to $E(R)$ that will give $C_0\leq \vartheta$ as a direct consequence.
Second, it reveals an important analogy between the Lov\'asz theta function and the cutoff rate. We will in fact introduce a function $\vartheta(\rho)$ that varies from the cutoff rate of the channel, when $\rho=1$, to the Lov\'asz theta function, when $\rho\to\infty$. The idea is to keep Lov\'asz' result in mind as a target but building it as the limit of a smoother construction. This construction is related to that of Gallager's expurgated lower bound to $E(R)$, but it is used precisely in the opposite direction. We believe that these analogies could shed new light on the understanding of the topic and deserve further study.

\subsection{Bhattacharyya distances and scalar products}
\label{sec:Umbrella-Bhattacharyya}
In deriving the desired bound, we will start our interpretation of classical-quantum channels as auxiliary mathematical tools for the study of classical channels. Contrarily to what may be considered the most traditional approach, however, in this section we do not interpret our channel's transition probabilities as the eigenvalues of positive semidefinite commuting operators. We consider instead the transition probabilities as the squared absolute values of the components of some wave functions, and it is thus more instructive to initially consider only pure-state channels. In this direction, we also need to recall briefly some important connections between the reliability function $E(R)$ and the Bhattacharyya distance between codewords. This connection is of great importance since the Bhattacharyya distance between distributions is related to a scalar product between unit norm vectors in a Hilbert space. It is this property that makes an adaptation of Lov\'asz' approach to study the function $E(R)$ possible.

For a generic input symbol $\myin$, consider the unit norm vector
\begin{equation}
\psi_\myin=\left(\sqrt {\myCh{\myin}{1}}, \sqrt {\myCh{\myin}{2}}, \ldots,\sqrt {\myCh{\myin}{|\myOut|}}\right)^\dagger
\label{def:psi_from_P}
\end{equation}
of the square roots of the conditional probabilities of the output symbols given input $\myin$. We call this the \emph{state vector} of input symbol $\myin$, in obvious analogy with the input signals of pure-state classical-quantum channels. We will also use the simplified notation $\psi_\myin=\sqrt{\myChv{\myin}}$.
Consider then the memoryless $n$-fold extension of our classical channel, that is, for an input sequence $\bm{\myin}=(\myin_1,\myin_2,\ldots,\myin_n)$, consider the square root of the conditional probability of a sequence $\bm{\myout}=(\myout_1,\myout_2,\ldots,\myout_n)$
\begin{equation*}
\sqrt{\myChn{\bm{\myin}}{\bm{\myout}}}=\prod_{i=1}^n \sqrt{\myCh{\myin_i}{\myout_i}}.
\end{equation*}
If all $\bm{y}$ sequences are listed in lexicographical order, we can express all the square roots of their conditional probabilities as the components of the vector $\bm{\psi}_{\bm{\myin}}=\sqrt{\myChnv{\bm{\myin}}}$, which satisfies
\begin{equation}
\bm{\psi}_{\bm{\myin}}=\psi_{\myin_1}\otimes\psi_{\myin_2}\otimes\cdots\psi_{\myin_n}
\label{eq:defPsi}
\end{equation}
where $\otimes$ is the Kronecker product. We call this vector the \emph{state vector} of the input sequence $\bm{\myin}$, again in analogy with classical-quantum channels. Let for simplicity $\bm{\psi}_m$ be the state vector of the codeword $\bm{\myin}_m$; then we can represent our code $\{\bm{x}_1, \bm{x}_2, \ldots, \bm{x}_M\}$ by means of its associated state vectors
$\{\bm{\psi}_1, \bm{\psi}_2, \ldots, \bm{\psi}_M\}$. Since all square roots are taken positive, note that our classical channel has a positive zero-error capacity if and only if there are at least two state vectors $\psi_\myin$, $\psi_\mysin$ such that $\braket{\psi_\myin}{\psi_\mysin}=0$. This implies that codes can be built such that $\braket{\bm{\psi}_m}{\bm{\psi}_{m'}}=0$ for some $m$, $m'$, that is, the two codewords $m$ and $m'$ cannot be confused at the output. 

However, the scalar product 
$\braket{\bm{\psi}_m}{\bm{\psi}_{m'}}$ plays a more general role, since it is related to the so called Bhattacharyya distance between the two codewords $m$ and $m'$. In particular, in a binary hypothesis testing between codewords $m$ and $m'$, an extension of the Chernoff Bound allows to assert that the minimum error probability $\Pe$ vanishes exponentially fast in the block length $n$ and that  \cite{shannon-gallager-berlekamp-1967-2}
\begin{equation*}
\log\frac{1}{\Pe}=d_{\text{C}}(\myChnv{\bm{\myin}_m},\myChnv{\bm{\myin}_{m'}})+o(n).
\end{equation*}
Using \eqref{eq:d_Bd_C}, we then see that 
\begin{equation*}
\log\frac{1}{\braket{\bm{\psi}_m}{\bm{\psi}_{m'}}}+o(n)\leq  \log\frac{1}{\Pe}\leq 2\log\frac{1}{\braket{\bm{\psi}_m}{\bm{\psi}_{m'}}}+o(n).
\end{equation*}
We add as a comment that equality holds on the left for the class of channels, introduced in \cite{shannon-gallager-berlekamp-1967-2}, called \emph{pairwise reversible} channels. These channels are such that for any pair of inputs $\myin$, $\mysin$ the quantity $\mu_{\myChv{\myin},\myChv{\mysin}}(s)$ is minimized by $s=1/2$, which  implies that $d_{\text{B}}(\myChv{\myin},\myChv{\mysin})=d_{\text{C}}(\myChv{\myin},\myChv{\mysin})$ for any pair of inputs $\myin$, $\mysin$.
For any channel, for a given code, the probability of error $\Pemax$ is lower bounded by the probability of error in each binary hypothesis test between two codewords. Hence, we find that 
\begin{equation}
\log\frac{1}{ \Pemax}\leq  \min_{m\neq m'} 2\log\frac{1}{ \braket{\bm{\psi}
_m}{\bm{\psi}_{m'}}} + o(n)
\label{eq"pemaxprodscal1}
\end{equation}
and, for pairwise reversible channels, 
\begin{equation}
\log\frac{1}{ \Pemax}\leq  \min_{m\neq m'} \log\frac{1}{ \braket{\bm{\psi}
_m}{\bm{\psi}_{m'}}} + o(n).
\label{eq"pemaxprodscal2}
\end{equation}
Hence, it is possible to upper bound $E(R)$ by lower bounding the quantity 
\begin{equation*}
\gamma=\max_{m\neq m'}\braket{\bm{\psi}
_m}{\bm{\psi}_{m'}}.
\end{equation*}

Lov\'asz' work aims at finding a value $\vartheta$ as small as possible that allows to conclude that, for a set of $M=e^{nR}>e^{n\vartheta}$ codewords, $\gamma$ cannot be zero, and thus at least two codewords are confusable. Here, instead, we want something more, that is, to find a lower bound on $\gamma$ for each code with rate $R>\vartheta$, so as to deduce an upper bound on $E(R)$ for all $R>\vartheta$.

\subsection{The umbrella bound}
\label{sec:Umbrella-umbrella}
Consider the scalar products between the channel state vectors $\braket{\psi_\myin}{\psi_\mysin}\geq 0$. For a fixed $\rho\geq 1$, consider then a set of $|\myIn|$ ``tilted'' state vectors, that is, unit norm vectors $\tilde{\psi}_1, \tilde{\psi}_2,\ldots, \tilde{\psi}_{|\myIn|}$ in any Hilbert space such that $| \braket{\tilde{\psi}_\myin}{\tilde{\psi}_\mysin}|\leq \braket{\psi_\myin}{\psi_\mysin}^{1/\rho}$. We call such a set of vectors $\{\tilde{\psi}_\myin\}$ an \emph{orthonormal representation of degree $\rho$} of our channel, and we call $\Gamma(\rho)$ the set of all possible such representations,
\begin{equation*}
\Gamma(\rho) = \left\{ \{\tilde{\psi}_\myin\} \, :\,  | \braket{\tilde{\psi}_\myin}{\tilde{\psi}_\mysin}|\leq \braket{\psi_\myin}{\psi_\mysin}^{1/\rho}\right\}.
\end{equation*}
Observe that $\Gamma(\rho)$ is non-empty since the original $\psi_\myin$ vectors satisfy the constraints.
The \emph{value} of an orthonormal representation is the quantity 
\begin{equation*}
V(\{\tilde{\psi}_\myin\})=\min_f\max_\myin\log \frac{1}{|\braket{\tilde{\psi}_\myin}{f}|^2}
\end{equation*}
where the minimum is over all unit norm vectors $f$. The optimal choice of the vector $f$ is called, following Lov\'asz, the \emph{handle} of the representation. We call it $f$ to point out that this vector plays essentially the same role as the auxiliary output distribution $\mathbf{f}$ used in the sphere packing bound of \cite{shannon-gallager-berlekamp-1967-1} (with their notation), a role that will be played by an auxiliary density operator $F$ later on in Section \ref{sec:SPB}.

Call now $\vartheta(\rho)$ the minimum value over all representations of degree $\rho$, 
\begin{align*}
\vartheta(\rho) & = \min_{\{\tilde{\psi}_\myin\} \in \Gamma(\rho)}V(\{\tilde{\psi}_\myin\})\\
& = \min_{\{\tilde{\psi}_\myin\} \in \Gamma(\rho) }\min_f\max_\myin\log \frac{1}{|\braket{\tilde{\psi}_\myin}{f}|^2}.
\end{align*}
The function $\vartheta(\rho)$ allows us to find an upper bound to $E(R)$ that we call the \emph{umbrella bound}. Later in Section \ref{sec:SP-umbrella}, we will interpret this bound from a different perspective, and we will introduce an evolution based on the sphere packing bound.
We have the following result.

\begin{theorem}
For any code of block-length $n$ with $M$ codewords and any $\rho\geq 1$, we have
\begin{equation*}
\max_{m}\sum_{m'\neq m}\braket{\bm{\psi}_m}{\bm{\psi}_{m'}}\geq \frac{\left( M e^{-n\vartheta(\rho)} -1\right)^\rho}{(M-1)^{\rho-1}}.
\end{equation*}
\end{theorem}
\begin{corollary}
For the reliability function of a general DMC we have the bound
\begin{equation}
E(R)\leq 2\rho\,\vartheta(\rho), \qquad R>\vartheta(\rho).
\label{eq:bound_theta2_1}
\end{equation}
If the channel is pairwise reversible, we can strengthen the bound to
\begin{equation}
E(R)\leq \rho\,\vartheta(\rho), \qquad R>\vartheta(\rho).
\label{eq:bound_theta2_2}
\end{equation}
\end{corollary}

\begin{IEEEproof}
Note that, for an optimal representation of degree $\rho$ with handle $f$, we have $|\braket{\tilde{\psi}_\myin}{f}|^2\geq e^{-\vartheta(\rho)}$, $\forall \myin$. Set now $\bm{f}=f^{\otimes n}$ and for an input sequence $\bm{\myin}=(\myin_1,\myin_2,\ldots,\myin_n)$ call, in analogy with \eqref{eq:defPsi},  $\tilde{\bm{\psi}}_{\bm{\myin}}=\tilde{\psi}_{\myin_1}\otimes\tilde{\psi}_{\myin_2}\otimes\cdots
\tilde{\psi}_{\myin_n}$. Observe that  we have 
\begin{align}
|\braket{\tilde{\bm{\psi}}_{\bm{\myin}}}{\bm{f}}|^2 & =  \prod_{i=1}^n|\braket{\tilde{\psi}_{\myin_i}}{f}|^2\\
& \geq  e^{-n\vartheta(\rho)} \label{eq:tilted_vs_F}.
\end{align}
This is the key step which is central to both Lov\'asz' approach and to the sphere packing bound: the construction of an auxiliary state which is ``close'' to all possible states associated to any sequence. In this case the states are close in terms of scalar product, while in the sphere packing bound they will be close in terms of the more general R\'enyi divergence. The basic idea, however, is not different. 

Let us first check how Lov\'asz' bound is obtained. Lov\'asz' approach is to bound the number $M$ of codewords with orthogonal state vectors, using the property that if $\tilde{\bm{\psi}}_1,\tilde{\bm{\psi}}_2,\ldots\tilde{\bm{\psi}}_M$ form an orthonormal set, then
\begin{align*}
1 & =  \|\bm{f}\|_2^2 \\
& \geq  \sum_m |\braket{\tilde{\bm{\psi}}_{m}}{\bm{f}}|^2\\
& \geq  Me^{-n\vartheta(\rho)}.
\end{align*}
Hence, if $M>e^{n\vartheta(\rho)}$, there are at least two non-orthogonal vectors in the set, say $|\braket{\tilde{\bm{\psi}}_m}{\tilde{\bm{\psi}}_{m'}}|^2>0$. But this implies that $|\braket{\bm{\psi}_m}{\bm{\psi}_{m'}}|^2\geq|\braket{\tilde{\bm{\psi}}_m}{\tilde{\bm{\psi}}_{m'}}|^{2\rho}>0$. Hence, if $R>\vartheta(\rho)$, no zero-error code can exist. We still have the free choice of $\rho$, and it is obvious that larger values of $\rho$ can only give better results. It is then preferable to simply work in the limit of $\rho\to \infty$ and thus build the representation $\{\tilde{\psi}_x\}$ under the only constraint that $| \braket{\tilde{\psi}_\myin}{\tilde{\psi}_\mysin}| = 0$ whenever $| \braket{{\psi}_\myin}{{\psi}_\mysin}|=0$. This gives precisely Lov\'asz' result.

Now, instead of bounding $R$ under the hypothesis of zero-error communication, we want to bound the probability of error for a given $R>\vartheta(\rho)$.
Considering the tilted state vectors of the code, we can rewrite equation \eqref{eq:tilted_vs_F} as
\begin{align*}
|\braket{\tilde{\bm{\psi}}_{m}}{\bm{f}}|^2  &  =  \bra{\bm{f}}\left(\ket{\tilde{\bm{\psi}}_m}\bra{\tilde{\bm{\psi}}_m}\right)\ket{\bm{f}}
\\ & \geq  e^{-n\vartheta(\rho)}.
\end{align*}
The second expression above has the benefit of easily allowing to average it  over different codewords. So, we can average this expression over all $m$ and, defining the matrix $\mathbf{\tilde{\Psi}}=\left(\tilde{\bm{\psi}}_1,\ldots,\tilde{\bm{\psi}}_M\right)/\sqrt{M}$, we get
\begin{equation*}
\bra{\bm{f}}\mathbf{\tilde{\Psi}}\mathbf{\tilde{\Psi}}^*\ket{\bm{f}}\geq e^{-n\vartheta(\rho)}.
\end{equation*}
Since $\bm{f}$ is a unit norm vector, this implies that the matrix $\mathbf{\tilde{\Psi}}\mathbf{\tilde{\Psi}}^*$ has at least one eigenvalue larger than or equal to $e^{-n\vartheta(\rho)}$. This in turn implies that also the matrix  $\mathbf{\tilde{\Psi}}^*\mathbf{\tilde{\Psi}}$ has itself an eigenvalue larger than or equal to $e^{-n\vartheta(\rho)}$, that is
\begin{equation*}
\lambda_{\max}\left(\mathbf{\tilde{\Psi}}^*\mathbf{\tilde{\Psi}}\right)\geq e^{-n\vartheta(\rho)}.
\end{equation*}
It is known that for a given matrix $A$ with elements $A(i,j)$, the following inequality holds
\begin{equation*}
\lambda_{\max}(A)\leq \max_{i}\sum_j|A(i,j)|.
\end{equation*}
Using this inequality with $A=\mathbf{\tilde{\Psi}}^*\mathbf{\tilde{\Psi}}$ we get
\begin{align*}
e^{-n\vartheta(\rho)} & \leq 
\max_m\sum_{m'} \frac{ |\braket{\tilde{\bm{\psi}}_m}{\tilde{\bm{\psi}}_{m'}}| } {M}\\& =  \frac{1}{M}\left(1+\max_m\sum_{m'\neq m}
|\braket{\tilde{\bm{\psi}}_m}{\tilde{\bm{\psi}}_{m'}}|\right).
\end{align*}
We then deduce
\begin{align*}
\frac{Me^{-n\vartheta(\rho)}-1}{M-1} & \leq \max_m\frac{1}{M-1}\sum_{m'\neq m}
|\braket{\tilde{\bm{\psi}}_m}{\tilde{\bm{\psi}}_{m'}}|\\
&\leq \max_m\frac{1}{M-1}\sum_{m'\neq m}
\braket{\bm{\psi}_m}{\bm{\psi}_{m'}}^{1/\rho}\\
& \leq \max_m\left(\frac{1}{M-1}\sum_{m'\neq m}
\braket{\bm{\psi}_m}{\bm{\psi}_{m'}}\right)^{1/\rho},
\end{align*}
where the last step is due to the Jensen inequality, since $\rho\geq 1$.
Extracting the sum from this inequality we obtain the inequality stated in the theorem.

To prove the corollary, simply note that 
\begin{align*}
\max_{m\neq m'}  \braket{\bm{\psi}_m}{\bm{\psi}_{m'}} & \geq \max_m \frac{1}{M-1}\sum_{m'\neq m}
|\braket{\bm{\psi}_m}{\bm{\psi}_{m'}}|\\
& \geq \left(\frac{Me^{-n\vartheta(\rho)}-1}{M-1}\right)^{\rho}\\
& \geq \left( e^{-n\vartheta(\rho)}-e^{-nR}\right)^{\rho}.
\end{align*}
The bound is trivial if $R\leq \vartheta(\rho)$. If $R>\vartheta(\rho)$, we deduce again Lov\'asz' result that there are two non-orthogonal codewords. But now we also have some further information; for $R>\vartheta(\rho)$, the second term in the parenthesis decreases exponentially faster than the first, which leads us to the conclusion that 
\begin{equation*}
\frac{1}{n}\min_{m\neq m'}\log \frac{1}{|\braket{{\bm{\psi}}_m}{{\bm{\psi}}_{m'}}|}\leq \rho \vartheta(\rho) + o(1).
\end{equation*}
The bounds in terms of $E(R)$ are then obtained by simply taking the limit $n\to\infty$ and using the bounds \eqref{eq"pemaxprodscal1} and \eqref{eq"pemaxprodscal2}.
\end{IEEEproof}

\begin{remark}
In passing from the theorem to the corollary, we have essentially substituted the \emph{maximum} Bhattacharyya distance between codewords for the largest \emph{average} distance from one codeword to the remaining ones. The reason for doing this is that we are unable to bound $E(R)$ efficiently in terms of the \emph{sum} of the distances, although intuition suggests that it should be possible to do it, in consideration of the behavior of $E(R)$ near the critical rate. This is related to the tightness of the union bound; it is our firm belief that this step is crucial and that improvements in this sense  could give important enhancements in the resulting bound. 
\end{remark}

A comment about the computation of this bound is in order. There is no essential difference between the evaluation of $\vartheta(\rho)$ and the evaluation of $\vartheta$. The optimal representation $\{\tilde{\psi}_\myin\}$, for any fixed $\rho$, can be obtained by solving a semidefinite optimization problem. If we consider the $(|\myIn|+1)\times (|\myIn|+1)$ Gram matrix
\begin{equation*}
G=[\tilde{\psi}_1,\ldots,\tilde{\psi}_{|\myIn|}, f]^T[\tilde{\psi}_1,\ldots,\tilde{\psi}_{|\myIn|}, f]
\end{equation*}
we note that finding the optimal representation amounts to solving the problem
\begin{equation*}
\begin{array}{lrcl}
& \max  & t & \\
\mbox{s.t.} & G(\myin,|\myIn|+1) & \geq  &t, \quad \forall \myin\leq |\myIn|\\
& G(\myin,\myin) & =  & 1, \quad \forall \myin\\
& G(\myin,\mysin) & \leq & \braket{\psi_\myin}{\psi_\mysin}^{1/\rho}, \quad \myin\neq \mysin\\
 &  G & \mbox{is} & \mbox{positive semidefinite.}
 \end{array}
\end{equation*}
The solution to this problem is $t^*=e^{-\vartheta(\rho)/2}$, and both the optimal representation vectors $\{\tilde{\psi}_\myin\}$ and the handle $f$ can be obtained by means of the spectral decomposition of the optimal $G$ found.

\subsection{Relation to known classical quantities}
\label{sec:Umbrella-classical}
We now study the behaviour of $\vartheta(\rho)$ for different values of $\rho$.
A first important comment is about the result obtained for $\rho=1$; the value $\vartheta(1)$ is simply the cutoff rate of the channel. Indeed, for $\rho=1$, we can without loss of generality use the obvious representation $\tilde{\psi}_\myin=\psi_\myin, \forall \myin$, since any different optimal representation will simply be a rotation of this (or an equivalent description in a space with a different dimension). In this case, all the components of all the vectors $\{\tilde{\psi}_\myin\}$ are non-negative, and this easily implies that the optimal $f$ can as well be chosen with non-negative components, since changing a supposedly negative component of $f$ to its absolute value can only improve the result.
Thus, $f$ can be written as the square root of a probability density $Q$ on $\myOut$ and we have 
\begin{align*}
\vartheta(1) & =  \min_f\max_\myin\log \frac{1}{|\braket{{\psi}_\myin}{f}|^2}\\
& =  \min_{Q} \max_\myin \left(-2\log \sum_\myout\sqrt{Q(\myout) \myCh{\myin}{\myout}}\right),
\end{align*}
where the minimum is now over all probability distribution $Q$ on the output alphabet $\myOut$.
As observed by Csisz\'ar \cite[Proposition 1, with $\alpha=1/2$]{csiszar-1995}, this expression equals the cutoff rate\footnote{We use the notation $R_1$ for the cutoff rate, instead of the more common $R_0$, for notational needs that will become clear in Section \ref{sec:Inform-radii}.}
 $R_1$ of the channel defined as
\begin{align*}
R_1 & = \max_{P}\left[- \log \sum_{\myin,\mysin}P(\myin)P(\mysin)\left(\sum_\myout \sqrt{\myCh{\myin}{\myout}\myCh{\mysin}{\myout}}\right)\right]\\
& =  \max_{P}\left[- \log \sum_{\myin,\mysin}P(\myin)P(\mysin) \braket{\psi_\myin}{\psi_\mysin}\right].
\end{align*}
The identity $\vartheta(1)=R_1$ will be discussed again later in light of the new interpretation that we will give of $\vartheta(\rho)$ after studying the sphere packing bound. We will see that it represents a nice connection between a classical channel and a pure-state classical-quantum channel possibly underlying the classical one.

Another important characteristic of the function $\vartheta(\rho)$ is observed in the limit $\rho\to\infty$. In the limit, the only constraint on representations is that $| \braket{\tilde{\psi}_\myin}{\tilde{\psi}_\mysin}| = 0$ whenever $| \braket{{\psi}_\myin}{{\psi}_\mysin}|=0$. Hence, when $\rho\to\infty$, the set of possible representations is precisely the same considered by Lov\'asz \cite{lovasz-1979}, and we thus have $\vartheta(\rho)\to\vartheta$ as $\rho\to \infty$. So, the value of $\vartheta(\rho)$ moves from the cutoff rate to the Lov\'asz theta function  when $\rho$ varies from $1$ to $\infty$.
This clearly implies that our bound to $E(R)$ is finite for all $R>\vartheta$ and thus Lov\'asz' bound
\begin{align*}
C_0 & \leq  \lim_{\rho\to\infty}\vartheta(\rho)\\
& =  \vartheta.
\end{align*}

In order to understand what happens for values of $\rho$ between 1 and $\infty$, it is instructive to consider first a class of channels introduced by Jelinek \cite{jelinek-1968} and later also studied by Blahut \cite{blahut-1977}. These are channels for which the $|\myIn|\times |\myIn|$ matrix $G_\rho$ with $(\myin,\mysin)$ element $G_\rho(\myin,\mysin)=\braket{\psi_\myin}{\psi_\mysin}^{1/\rho}$ is positive semidefinite for all $\rho\geq 1$. It was proved by Jelinek that, for these channels, the expurgated bound of Gallager \cite{gallager-1965} is invariant over $n$-fold extensions of the channel, that is, it has the same form when computed on a single channel use or on multiple channel uses (this is not true in general). Thus, if the conjecture made in \cite[pag. 77]{shannon-gallager-berlekamp-1967-1}, that the expurgated bound computed on the $n$-fold channel is tight asymptotically when $n\to \infty$, is true, then for these channels the reliability would be known exactly since it would equal the expurgated bound for the single use channel. It is also known that for these channels, the inputs can be partitioned in subsets such that all pairs of symbols from the same subset are confusable and no pair of symbols from different subsets are confusable. The zero error capacity in this case is simply the logarithm of the number of such subsets.
For these channels, since the matrix $G_\rho$ is positive semidefinite, there exists a set of vectors $\tilde{\psi}_1,\tilde{\psi}_2,\ldots,\tilde{\psi}_{|\myIn|}$ such that $\braket{\tilde{\psi}_\myin}{\tilde{\psi}_\mysin}=G_\rho(\myin,\mysin)$, that is, for all $\rho\geq 1 $, representations of degree $\rho$ exist that satisfy all the constraints with equality. 
In this case, the equivalence with the cutoff rate that we have seen for $\rho=1$ can be in a sense extended to other $\rho$ values. 
We will in fact see in Section \ref{sec:classical-pure} that we can write 
\begin{align}
\vartheta(\rho) & =  \min_f\max_\myin\log \frac{1}{|\braket{\tilde{\psi}_\myin}{f}|^2}\\
 & =  \max_{P}\left[- \log \sum_{\myin,\mysin}P(\myin)P(\mysin) \braket{\tilde{\psi}_\myin}{\tilde{\psi}_\mysin}\label{eq:alt_def_theta}\right]\\
  & =  \max_{P}\left[- \log \sum_{\myin,\mysin}P(\myin)P(\mysin) \braket{{\psi}_\myin}{{\psi}_\mysin}^{1/\rho}\label{eq:alt_def_theta2}\right]\\
& = \max_{P} \frac{\Ex(\rho,P)}{\rho},
\end{align}
where, in the last step, we have used \eqref{eq:class_expurgated3}.
Hence, under such circumstances, we find that $\vartheta(\rho)=\Ex(\rho)/\rho $, where $\Ex(\rho)$ is  the value of the coefficient  used in the expurgated bound as defined in equations \eqref{eq:class_expurgated2}-\eqref{eq:class_expurgated3}.
Note that, for each $\rho$, the expurgated bound is a straight line which intercepts the axis $R$ and $E$ at the points $\Ex(\rho)/\rho$ and $\Ex(\rho)$ respectively, which equal $\vartheta(\rho)$ and $\rho\,\vartheta(\rho)$. Hence, our bound is obtained by drawing the curve parameterized as $(\Ex(\rho)/\rho, 2\Ex(\rho))$ in the $(R,E)$ plane. This automatically implies that we obtain the bound
\begin{align*}
C_0 \leq \lim_{\rho\to\infty} \frac{\Ex(\rho)}{\rho},
\end{align*}
which gives the precise value of the zero-error capacity in this case (which is however trivial) and, if $C_0=0$, the bound
\begin{align*}
E(0) & \leq  \lim_{\rho\to\infty} 2\rho\,\vartheta(\rho)\\
& =  \lim_{\rho\to\infty} 2\Ex(\rho)\\
& =  2 \Eex(0).
\end{align*}
If the channel is pairwise reversible, this can then be improved to $E(R)\leq \Eex(0)$, which is obviously tight.

For general channels with a non-trivial zero-error capacity, like for example any channel whose confusability graph is a pentagon, what happens is that the matrix $G_\rho$ is in general positive semidefinite only for values of $\rho$ in a range $[1,\bar{\rho}]$, and then it becomes not positive semidefinite for  $\rho>\bar{\rho}$. This implies that for $\rho>\bar{\rho}$, representations that satisfy all the constraints with equality do not exist in general. In this case, the two expressions in equations \eqref{eq:alt_def_theta} and \eqref{eq:alt_def_theta2} are no longer equal and in general they could both differ from $\vartheta(\rho)$. If all the values $\braket{\tilde{\psi}_\myin}{\tilde{\psi}_{\mysin}}$ are nonnegative, however, we will prove by means of Theorem \ref{th:RinftyV} in Section \ref{sec:classical-pure} that the expression in \eqref{eq:alt_def_theta} equals $\vartheta(\rho)$. In this case, we see the interesting difference between $\vartheta(\rho)$ and $E_x(\rho)/\rho$. The two 	quantities follow respectively \eqref{eq:alt_def_theta} and \eqref{eq:alt_def_theta2}. When $\rho\to\infty$, $\vartheta(\rho)$ tends to $\vartheta$, an upper bound to $C_0$. The value $E_x(\rho)/\rho$ instead is known to converge to the independence number of the confusability graph of the channel \cite{korn-1968}, a lower bound to  $C_0$.
More generally, if $\braket{\tilde{\psi}_\myin}{\tilde{\psi}_{\mysin}}\geq 0$,  $\forall \myin,\mysin$, since $\vartheta(\rho)$ is given by equation \eqref{eq:alt_def_theta}, it is an upper bound to \eqref{eq:alt_def_theta2} and thus to $E_{\text{x}}(\rho)/\rho$. It can then also be proved that \eqref{eq:alt_def_theta} is multiplicative, in this case, over the $n$-fold tensor power of the representation $\{\tilde{\psi}_\myin\}$. This implies that, for all $n$, $\vartheta(\rho)$ is an upper bound to the (normalized) expurgated bound $E_{\text{x}}^{(n)}(\rho)/\rho$ computed for the $n$-fold memoryless extension of the channel. That is, $\vartheta(\rho)$ generalizes $\vartheta$ in the sense that, in the same way as
\begin{equation*}
\vartheta \geq \sup_{n}\lim_{\rho\to \infty} \frac{E_{\text{x}}^{(n)}(\rho)}{\rho}=C_0,
\end{equation*} 
also
\begin{equation*}
\vartheta(\rho) \geq \sup_{n}\frac{E_{\text{x}}^{(n)}(\rho)}{\rho}.
\end{equation*} 
The discussion of this point with generality requires some technicalities about the function $\vartheta(\rho)$ that would bring us too far and will be presented in a future work.

It is worth pointing out that, for some channels, the optimal representation may even stay fixed for $\rho$ larger than some given finite value $\rho_{\max}$, and $\vartheta(\rho)$ is thus constant for $\rho\geq \rho_{\max}$ (in this case, the bounds are useless for $\rho>\rho_{\max}$). This happens for example for the noisy typewriter channel with five inputs and crossover probability $1/2$. In this case $\bar{\rho}=\rho_{\max}\approx 1.44$; as shown in Fig. \ref{fig}, for $\rho<\rho_{\max}$ we have $\vartheta(\rho)=E_{\text{x}}(\rho)/\rho$ while, for $\rho\geq \rho_{\max}$, $\vartheta(\rho)=C_0=\log\sqrt{5}$.

\begin{figure}
\centering
	\includegraphics[width=0.7\linewidth]{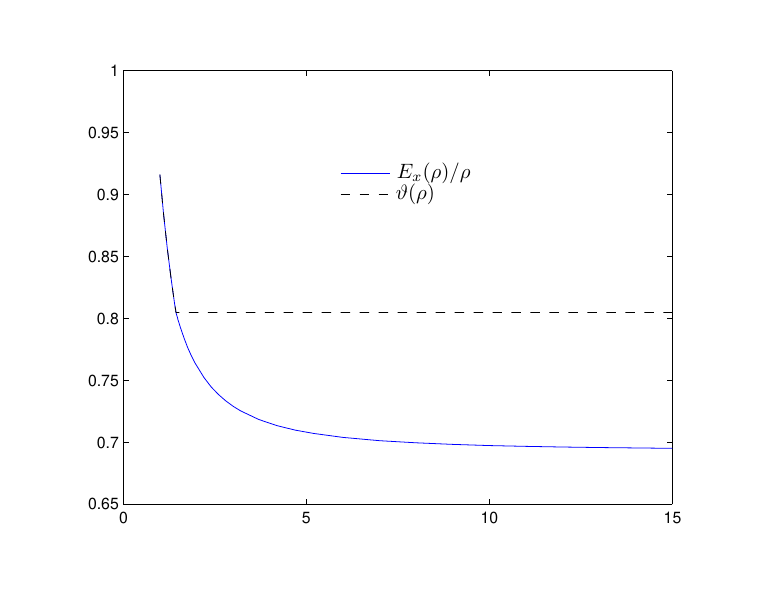}
	\caption{Plot of $\vartheta(\rho)$ and $E_x(\rho)/\rho$ for the noisy typewriter channel with five inputs and crossover probability $1/2$. \textbf{Note:} this plot is wrong in the published version of the paper.}
\label{fig}
\end{figure}

\subsection{Relation to classical-quantum channels}
\label{sec:Umbrella-quantum}

In deriving the umbrella bound in this section, we have mentioned classical-quantum channels but we have not explicitly used any of their properties. The derived bound could be interpreted as a simple variation of Lov\'asz' argument toward bounding $E(R)$. We decided in any case to use a notation that suggests an interpretation in terms of classical-quantum channels because, as we will see later in the paper, the bound derived here is a special case of a more general bound that can be derived by properly applying the sphere packing bound for classical-quantum channels.

In particular, while the construction of the representation $\{\tilde{\psi}_\myin\}$ appears in this section as a purely mathematical trick to bound $E(R)$ by means of a geometrical representation of the channel, it will be seen from the results of Section \ref{sec:SP-umbrella} that in the context of classical-quantum channels this procedure is a natural way to bound $E(R)$ by comparing the original channel with an auxiliary one. In the classical case, Lov\'asz' result came completely unexpected since it involves the unconventional idea of using vectors with negative components to play the same role of $\sqrt{\myChv{\myin}}$. When formulated in the classical-quantum setting, however, this approach becomes completely transparent and does not require pushing imagination out of the original domain. 
We may say that classical-quantum channels are to classical channels as complex numbers are to real numbers. In this analogy, Lov\'asz' theta function is like Cardano's solution for cubics.

\section{Quantum Binary Hypothesis Testing}
\label{sec:Quantum-H-test}
In this section, we consider the problem of binary hypothesis testing between quantum states. In particular, we will prove a quantum extension of the converse part of the Shannon-Gallager-Berlekamp generalized Chernoff bound \cite[Th. 5]{shannon-gallager-berlekamp-1967-1}. This is a fundamental tool in bounding the probability of error for codes over classical-quantum channels and it will thus play a central role in Sections \ref{sec:SPB} and \ref{sec:Zero-rate} for the proof of the sphere packing bound and the zero-rate bound.
 
Let $A$ and $B$ be two density operators in a Hilbert space $\mathcal{H}$. We are interested in the problem of discriminating between the hypotheses that a given system is in state $A$ or $B$. We suppose here that the two density operators have non-disjoint supports, for otherwise the problem is trivial.
Of particular importance in quantum statistics is the case where $n$ independent copies of the system are available, which means that we actually have to decide between the $n$-fold tensor powers $A^{\otimes n}$ and $B^{\otimes n}$ of $A$ and $B$. 
The decision has to be taken based on the result of a measurement that can be identified with a pair of positive operators $\{\mathds{1}-\Pi^{(n)}, \Pi^{(n)}\}$ associated with $A^{\otimes n}$ and $B^{\otimes n}$ respectively. The probability of error when the state is $A^{\otimes n}$ or $B^{\otimes n}$ are, respectively,
\begin{equation*}
\mathsf{P}_{\text{e}|A^{\otimes n}}=\Tr\left[\Pi^{(n)} A^{\otimes n}\right], \quad \mathsf{P}_{\text{e}|B^{\otimes n}}=\Tr\left[(\mathds{1}-\Pi^{(n)})B^{\otimes n}\right].
\end{equation*}
We are interested in the asymptotic behavior of the probability of error as $n$ goes to infinity.
The following result was recently derived in \cite{audenaert-2007}, \cite{nussbaum-szkola-2009} (see also \cite{audenaert-et-al-2008}).
\begin{theorem}[Quantum Chernoff Bound]
\label{th:QChernoff}
Let $A,B$ be density operators with Chernoff distance $d_{\text{C}}(A,B)$ and let $\eta_0$ and $\eta_1$ be positive real numbers. For any $n$, let
\begin{equation*}
\Pe^{(n)}=\inf_{{\Pi}^{(n)}}\left[\eta_0 \mathsf{P}_{\text{e}|A^{\otimes n}}+ \eta_1 \mathsf{P}_{\text{e}|B^{\otimes n}}\right],
\end{equation*}
where the infimum is over all measurements. Then
\begin{equation}
\lim_{n\to \infty}-\frac{1}{n}\log \Pe^{(n)}=d_{\text{C}}(A,B).
\label{eq:QChernoff}
\end{equation}
\end{theorem}
Note that the coefficients $\eta_0,\eta_1$ have no effect on the asymptotic exponential behavior of the error probability. With fixed $\eta_0,\eta_1$, the optimal projectors are such that the error probabilities $\mathsf{P}_{\text{e}|A^{\otimes n}}$ and $\mathsf{P}_{\text{e}|B^{\otimes n}}$ have the same exponential decay in $n$.

In some occasions, and in particular for the purpose of the present paper, it is important to characterize the performance of optimal tests when different exponential behaviour for the two error probabilities are needed. The following result has been recently obtained as a generalization of the previous theorem \cite{nagaoka-2006}, \cite{audenaert-et-al-2008}.
\begin{theorem}
\label{th:hoeffding}
Let $A,B$ be density operators with non-disjoint supports and let
 $\xi=\Tr (B\,\mbox{supp}(A))$. Let $e(r)$ be defined as
\begin{equation*}
e(r)=
\sup_{0\leq s < 1}  \frac{-s r - \mu_{A,B}(s)}{1-s},  \qquad  r \geq -\log \xi ,
\end{equation*}
and $e(r)=\infty$ if $r<-\log \xi$.
Let $\mathcal{P}$ be the set of all sequences of operators $\Pi^{(n)}$ such that
\begin{equation*}
\lim_{n\to \infty}-\frac{1}{n}\log \left(\Tr (\mathds{1}-\Pi^{(n)})B^{\otimes n}\right)\geq r.
\end{equation*}
Then
\begin{equation*}
\sup_{\{\Pi^{(n)}\}\in\mathcal{P}}\left[\lim_{n\to \infty}-\frac{1}{n}\log \left(\Tr \Pi^{(n)}A^{\otimes n}\right)\right] = e(r).
\end{equation*}
\end{theorem}

This generalization of the Chernoff bound, however, is not yet sufficient for the purpose of the present paper. In channel coding problems, in fact, what is usually of interest is the more general problem of distinguishing between two states $\bm{A}$ and $\bm{B}$ that are represented by tensor products of non-identical density operators such as
\begin{equation*}
\bm{A}=A_1\otimes A_2\otimes\cdots\otimes A_n \quad \mbox{and} \quad 
\bm{B}=B_1\otimes B_2\otimes\cdots\otimes B_n.
\end{equation*}
In this case, it is clear that the probability of error depends on the sequences $A_1,A_2,\cdots$ and $B_1,B_2,\cdots$ in such a way that an asymptotic result of the form of Theorems \ref{th:QChernoff} and \ref{th:hoeffding} is not to be hoped in general.
 For example, after the obvious redefinition of $\Pe^{(n)}$ in Theorem \ref{th:QChernoff}, the limit on the left hand side of \eqref{eq:QChernoff} may even not exist. For this reason, it is useful to establish a more general result than Theorems \ref{th:QChernoff} and \ref{th:hoeffding} which is stated directly in terms of two general operators, that in our case are to be interpreted as the operators $\bm{A}$ and $\bm{B}$ above. This is precisely what is done in \cite[Th. 5]{shannon-gallager-berlekamp-1967-1} for the classical case and we aim here at deriving at least the corresponding converse part of that result for the quantum case.

\begin{theorem}
\label{th:QSGBB}
Let $\bm{A}$ and $\bm{B}$ be density operators with non-disjoint supports, let $\Pi$ be a measurement operator for the binary hypothesis test between $\bm{A}$ and $\bm{B}$ and let the probabilities of error $\mathsf{P}_{\text{e}|\bm{A}},\mathsf{P}_{\text{e}|\bm{B}}$ be defined as
\begin{equation*}
\mathsf{P}_{\text{e}|\bm{A}}=\Tr\left[\Pi \bm{A}\right], \quad \mathsf{P}_{\text{e}|\bm{B}}=\Tr\left[(\mathds{1}-\Pi)\bm{B}\right].
\end{equation*}
 
 Let $\mu(s)=\mu_{\bm{A},\bm{B}}(s)$ as defined in equations \eqref{eq:mu_quantum_def_1}-\eqref{eq:mu_quantum_def_2}. Then, for any $0<s<1$, either
\begin{equation*}
\mathsf{P}_{\text{e}|\bm{A}}>\frac{1}{8}\exp\left[\mu(s)-s\mu'(s)-s\sqrt{2\mu''(s)}\right]
\end{equation*}
or
\begin{equation*}
\mathsf{P}_{\text{e}|\bm{B}}>\frac{1}{8}\exp\left[\mu(s)+(1-s)\mu'(s)-(1-s)\sqrt{2\mu''(s)}\right].
\end{equation*}
\end{theorem}

\begin{IEEEproof}
This theorem is essentially the combination of the main idea introduced in \cite{nussbaum-szkola-2009} for proving the converse part of the quantum Chernoff bound and of \cite[Th. 5]{shannon-gallager-berlekamp-1967-1}, the classical version of this same theorem. 
Since some intermediate steps of those proofs are needed, we provide the details here for the reader's convenience. 

Following \cite{audenaert-et-al-2008}, let the spectral decomposition of $\bm{A}$ and $\bm{B}$ be respectively 
\begin{equation*}
\bm{A}=\sum _i\alpha_i \ket{a_i}\bra{a_i} \quad \mbox{and}\quad \bm{B}=\sum_j \beta_j    \ket{b_j}\bra{b_j}.
\end{equation*}
where $\{a_i\}$ and $\{b_j\}$ are orthonormal bases.
First observe that, from the Quantum Neyman-Pearson Lemma (\cite{helstrom-1976}, \cite{holevo-1972}), it suffices to consider orthogonal projectors $\Pi$. So, we have
$\Pi=\Pi^2=\Pi\mathds{1}\Pi=\sum_j \Pi\ket{b_j}\bra{b_j}\Pi$. Symmetrically,  we have that $(\mathds{1}-\Pi)=\sum_i (\mathds{1}-\Pi)\ket{a_i}\bra{a_i}(\mathds{1}-\Pi)$. So we have 
\begin{align*}
\mathsf{P}_{\text{e}|\bm{A}} & =  \Tr \Pi \bm{A}\\
 & =  \sum_{i,j}\alpha_i |\bra{a_i} \Pi\ket{b_j}|^2;\\
\mathsf{P}_{\text{e}|\bm{B}} & =  \Tr (\mathds{1}-\Pi) \bm{B}\\
 & =  \sum_{i,j}\beta_j |\bra{a_i} \mathds{1}-\Pi\ket{b_j}|^2.
\end{align*}
Thus, for any positive $\eta_0, \eta_1$, we have 
\begin{align}
\eta_0 & \mathsf{P}_{\text{e}|\bm{A}} + \eta_1 \mathsf{P}_{\text{e}|\bm{B}}\nonumber \\ 
&=  \sum_{i,j}\left(\eta_0 \alpha_i |\bra{a_i} \Pi\ket{b_j}|^2 + \eta_1 \beta_j |\bra{a_i} \mathds{1}-\Pi\ket{b_j}|^2 \right)\nonumber\nonumber\\
 & \geq  \sum_{i,j}  \min(\eta_0\alpha_i,\eta_1\beta_j ) \left(|\bra{a_i} \Pi\ket{b_j}|^2 +  |\bra{a_i} \mathds{1}-\Pi\ket{b_j}|^2 \right)\nonumber\nonumber\\
 & \geq  \sum_{i,j}  \min(\eta_0\alpha_i,\eta_1\beta_j ) \frac{|\braket{a_i}{b_j}|^2}{2}\nonumber\\
 & =  \frac{1}{2}\sum_{i,j}  \min\left(\eta_0\alpha_i |\braket{a_i}{b_j}|^2,\eta_1\beta_j |\braket{a_i}{b_j}|^2 \right) 
\label{eq:minsum},
\end{align}
where the second-to-last inequality is motivated by the fact that for any two complex numbers $v,w$ we have $|v|^2+|w|^2\geq |v+w|^2/2$.

Now, following \cite{nussbaum-szkola-2009}, consider the two probability distributions defined by the \emph{Nussbaum-Szko{\l}a mapping}
\begin{equation}
Q_0(i,j)=\alpha_i |\braket{a_i}{b_j}|^2, \quad Q_1(i,j)=\beta_j |\braket{a_i}{b_j}|^2.
\label{eq:Pdef_first}
\end{equation}
These two probability distributions are both strictly positive for at least one pair of $(i,j)$ values, since we assumed $\bm{A},\,\bm{B}$ to have non-disjoint supports. Furthermore, they have the nice property that 
\begin{equation*}
\Tr \bm{A}^{1-s}\bm{B}^s=\sum_{i,j}Q_0(i,j)^{1-s}Q_1(i,j)^s,
\end{equation*}
so that
\begin{align*}
\mu_{ \bm{A},\bm{B}}(s) & =  \log \Tr \bm{A}^{1-s}\bm{B}^s\\
& =  \log  \sum_{i,j}Q_0(i,j)^{1-s}Q_1(i,j)^s\\
& =  \mu_{Q_0,Q_1}(s).
\end{align*}
Define\footnote{Note that $\tilde{Q}_0=Q_0$ and $\tilde{Q}_1=Q_1$ if $Q_0$ and $Q_1$ have the same support.} 
\begin{equation}
\tilde{Q}_s(i,j)=\frac{Q_0(i,j)^{1-s}Q_1(i,j)^s}{\sum_{i',j'}Q_0(i',j')^{1-s}Q_1(i',j')^s},
\label{eq:Qdef}
\end{equation}
and observe that
\begin{align}
\mu'(s) & =  E_{\tilde{Q}_s}\left[\log(Q_1/Q_0)\right],
\label{eq:mu'}\\
\mu''(s) & =  \mbox{Var}_{\tilde{Q}_s}\left[\log(Q_1/Q_0)\right],
\label{eq:mu''}
\end{align}
where the subscript $\tilde{Q}_s$ means that the expected values are with respect to the probability distribution $\tilde{Q}_s$. If one defines the set
\begin{equation}
\mathcal{Z}_s=\left\{(i,j): \left|\log\left(\frac{Q_1(i,j)}{Q_0(i,j)}\right)-\mu'(s)\right|\leq \sqrt{2\mu''(s)}\right\}
\label{eq:Ydef}
\end{equation}
then, by Chebyshev's inequality,
\begin{equation}
\sum_{(i,j)\in\mathcal{Z}_s}\tilde{Q}_s(i,j)>1/2.
\label{eq:Cheby_Q}
\end{equation}
It is easily checked, using the definitions \eqref{eq:Qdef} and \eqref{eq:Ydef}, that for each $(i,j)\in \mathcal{Z}_s$ the distribution $\tilde{Q}_s$ satisfies
\begin{equation}
\tilde{Q}_s(i,j)  \leq  Q_0(i,j)\left(\exp{[\mu(s)-s\mu'(s)-s\sqrt{2\mu''(s)}]}\right)^{-1},
\label{eq:qless1}
\end{equation}
\begin{equation}
\tilde{Q}_s(i,j)  \leq  Q_1(i,j)\Bigl(\exp [\mu(s)+(1-s)\mu'(s) \Bigr.  \Bigl. -(1-s)\sqrt{2\mu''(s)}]\Bigr)^{-1}.\label{eq:qless2}
\end{equation}
Hence, in $\mathcal{Z}_s$, $\tilde{Q}_s(i,j)$ is bounded by the minimum of the two expressions on the right hand side of \eqref{eq:qless1} and \eqref{eq:qless2}. 
If we call $\eta_0$ the coefficient of $Q_0(i,j)$ in \eqref{eq:qless1} and $\eta_1$ the coefficient of $Q_1(i,j)$ in \eqref{eq:qless2},  using \eqref{eq:Cheby_Q} we obtain 
\begin{eqnarray*}
\frac{1}{2} & <  & \sum_{(i,j)\in \mathcal{Z}_s} \tilde{Q}_s(i,j)\\
& \leq & \sum_{(i,j) \in \mathcal{Z}_s}\min\left(\eta_0 Q_0(i,j),\eta_1 Q_1(i,j)\right)\\
& \leq & \sum_{(i,j)}\min\left(\eta_0 Q_0(i,j),\eta_1 Q_1(i,j)\right).
\end{eqnarray*}

Now note that the last expression, by the definition of $Q_0$ and $Q_1$ in \eqref{eq:Pdef_first}, exactly equals the sum in \eqref{eq:minsum}. So, with the selected values of $\eta_0$ and $\eta_1$ we have $\eta_0 \mathsf{P}_{\text{e}|\bm{A}} + \eta_1 \mathsf{P}_{\text{e}|\bm{B}} >1/4$. Hence, either $\mathsf{P}_{\text{e}|\bm{A}}>\eta_0^{-1}/8$ or $\mathsf{P}_{\text{e}|\bm{B}} >\eta_1^{-1}/ 8$, concluding the proof.
\end{IEEEproof}

For the special case where $\bm{A}=A^{\otimes n}$ and $\bm{B}=B^{\otimes n}$, the bounds on $\mathsf{P}_{\text{e}|\bm{A}}$ and $\mathsf{P}_{\text{e}|\bm{B}}$ derived in Theorem \ref{th:QSGBB} can be simplified in light of the observation that 
\begin{equation*}
\mu_{\bm{A},\bm{B}}(s)=n\,\mu_{A,B}(s).
\end{equation*}
With some algebra, using the convexity of $\mu_{A,B}(s)$ and its relation with $D_s(B||A)$, it is then possible to show that Theorem \ref{th:QSGBB} implies the converse part of Theorem \ref{th:hoeffding}.

\section{Sphere Packing Bound for\\ Classical-Quantum Channels}
\label{sec:SPB}

In this section, the sphere packing bound for general classical-quantum channels is proved. We will follow closely the proof given in \cite[Sec. IV]{shannon-gallager-berlekamp-1967-1} for the classical case. It is the author's belief that the proof of the sphere packing bound used in \cite{shannon-gallager-berlekamp-1967-1} is not really widely known, especially within the quantum information theory community, because, as explained in the introduction, the simpler approach used in \cite{haroutunian-1968} has become much more popular\footnote{Viterbi and Omura \cite{viterbi-omura-book} define ``an intellectual tour-de-force'', even if characterized by ``flavor, style, elegance'', the proof of the sphere packing bound of \cite{shannon-gallager-berlekamp-1967-1}
and Gallager himself defines it as ``quite complicated'' \cite{gallager-2001} and ``tedious and subtle to derive'' \cite{gallager-book}. See Appendix \ref{sec:App_history} for some historical comments on the proof of the theorem in the classical case.}. A detailed analysis of that proof, however, is useful for the understanding of the analogy between the sphere packing bound and Lov\'asz' bound that will be discussed in Section \ref{sec:Inform-radii}.
Furthermore, some intermediate steps in the proof are clearly to be adjusted from the classical case to the quantum case, and this does not always come as a trivial task. Hence, it is both useful and necessary to go through the whole proof used in \cite{shannon-gallager-berlekamp-1967-1}. To avoid repetitions, we present the proof directly speaking in terms of classical-quantum channels, taking advantage of the weaker results that we are pursuing here\footnote{In \cite{shannon-gallager-berlekamp-1967-1}, bounds for fixed $M$ and $n$ are obtained. Here, we are only interested instead in determining $E(R)$ and we can then work in the asymptotic regime $n\to \infty$.}  with respect to \cite[Th. 5]{shannon-gallager-berlekamp-1967-1} to overcome some technical difficulties that arise in this more general context.

\begin{theorem}[Sphere Packing Bound]
\label{th:sphere-packing}
Let $S_1,\ldots,S_{|\myIn|}$ be the input signal states of a classical-quantum channel and let $E(R)$ be its reliability function. Then, for all positive rates $R$ and all positive $\varepsilon < R$,
\begin{equation*}
E(R)\leq \Esp(R-\varepsilon),
\end{equation*}
where $\Esp(R)$ is defined by the relations
\begin{align}
\Esp(R) & =  \sup_{\rho \geq 0} \left[ E_0(\rho) - \rho R\right]\label{eq:Esp},\\
E_0(\rho) & =  \max_{P}E_0(\rho,P)\label{eq:E0rho},\\
E_0(\rho,P) & = -\log\Tr\left( \sum_\myin P(\myin) S_\myin^{1/(1+\rho)} \right)^{1+\rho}.
\label{eq:E0rhoq}
\end{align}
\end{theorem}

\begin{remark}
The role of the arbitrarily small constant $\varepsilon$ is only important for the single value of the rate $R=R_\infty$ where the sphere packing bound goes to infinity.
\end{remark}

\begin{IEEEproof}
The key point is the idea first used (in a published work) by Fano \cite{fano-book} of bounding the probability of error for at least one codeword $\bm{\myin}_m$ by studying a binary hypothesis testing problem between $\bm{S}_{\bm{\myin}_m}$ and a dummy state $\bm{F}$. Roughly speaking, we will show that there exists one $m$ and a density operator $\bm{F}$ such that 
\begin{itemize}
\item[-] the probability under state $\bm{F}$ of the outcome associated to the decision for message $m$, call it $\mathsf{P}_{m|\bm{F}}=\Tr(\Pi_m\bm{F})$, is ``small'';
\item[-] the state $\bm{F}$ is only distinguishable from $\bm{S}_{\bm{\myin}_m}$  ``to a certain degree'' in a binary detection test. 
\end{itemize} 
Using Theorem \ref{th:QSGBB}, this will imply that the probability $\mathsf{P}_{m|m}$ cannot be too high. The whole proof is devoted to the construction of such a state $\bm{F}$, which has to be chosen properly depending on the  code. We are now ready to start the detailed proof.

We first simplify the problem using a very well known observation, that is, the fact that for the study of $E(R)$ we can only  consider the case of \emph{constant composition} codes. It is well known that every code with rate $R$ and block length $n$ contains a constant composition subcode of rate $R'=R-o(1)$, where $o(1)$ goes to zero when $n$ goes to infinity (see \cite{blahut-book,viterbi-omura-book,csiszar-korner-book}). This is due to the fact that the number of different compositions of codewords of length $n$ is only polynomial in $n$ while the code size is exponential. Hence, we will focus on this constant composition subcode and consider it as our initial code. Let thus our code have $M$ codewords with the same composition $P$, that is, $P$ is the distribution on $\myIn$ such that symbol $\myin$ occurs $nP(\myin)$ times in each codeword.

Let $\bm{F}$ be a density operator in $\mathcal{H}^{\otimes n}$. We will first apply Theorem \ref{th:QSGBB} using one of the states $\bm{S}_{\bm{\myin}_m}$ as state $\bm{A}$ and $\bm{F}$ as state $\bm{B}$. This will result in a trade-off between the rate of the code $R$ and the probability of error $\Pemax$, where both quantities will be parameterized in the parameter $s$, a higher rate being allowed if a larger $\Pemax$ is tolerated and vice versa. This trade-off depends of course on the composition $P$ and on the density operator $\bm{F}$. We will later pick $\bm{F}$ properly so as to obtain the best possible bound for a given $R$ valid for all compositions $P$.

For any $m=1,\ldots,M$ consider a binary hypothesis test between $\bm{S}_{\bm{\myin}_m}$ and $\bm{F}$. We assume that their supports are not disjoint (we will later show, after equation \eqref{def:fs} below, that such a choice of $\bm{F}$ is possible)  and define the quantity
\begin{align*}
\mu(s) & =   \mu_{\bm{S}_{\bm{\myin}_m},\bm{F}}(s)\\
& = \log \Tr \bm{S}_{\bm{\myin}_m}^{1-s}\bm{F}^s.
\end{align*}
Applying Theorem \ref{th:QSGBB} with $\bm{A}=\bm{S}_{\bm{\myin}_m}$, $\bm{B}=\bm{F}$ and $\Pi=\mathds{1}-\Pi_m$, we find that for each $s$ in $(0,1)$, either 
\begin{equation*}
\Tr\left[\left(\mathds{1}-\Pi_m\right)\bm{S}_{\bm{\myin}_m}\right]>\frac{1}{8}\exp\left[\mu(s)-s\mu'(s)-s\sqrt{2\mu''(s)}\right]
\end{equation*}
or
\begin{equation*}
\Tr\left[ \Pi_m\bm{F}\right]>\frac{1}{8}\exp\left[\mu(s)+(1-s)\mu'(s)-(1-s)\sqrt{2\mu''(s)}\right].
\end{equation*}
Note now that $\Tr\left[\left(\mathds{1}-\Pi_m\right)\bm{S}_{\bm{\myin}_m}\right]=\Pem\leq \Pemax$ for all $m$. Furthermore, since $\sum_{m=1}^M \Pi_m\leq \mathds{1}$, for at least one value of $m$ we have $\Tr\left[\Pi_m\bm{F}\right]\leq 1/M=e^{-nR}$. Choosing this particular $m$, we thus obtain from the above two equations that 
either
\begin{equation}
\Pemax>\frac{1}{8}\exp\left[\mu(s)-s\mu'(s)-s\sqrt{2\mu''(s)}\right]
\label{eq:cond1}
\end{equation}
or
\begin{equation}
R<-\frac{1}{n}\left[ \mu(s)+(1-s)\mu'(s)-(1-s)\sqrt{2\mu''(s)} - \log 8\right].
\label{eq:cond2}
\end{equation}

In these equations we begin to see the aimed trade-off between the rate and the probability of error. It is implicit here in the definition of $\mu(s)$ that both equations depend on $\bm{S}_{\bm{\myin}_m}$ and $\bm{F}$. Since $m$ has been fixed, we can drop its explicit indication and use simply $\bm{\myin}$ in place of $\bm{\myin}_m$ from this point on. We will now let $R_n(s,\bm{S}_{\bm{\myin}},\bm{F})$ denote the right hand side of \eqref{eq:cond2}, that is 
\begin{equation*}
R_n\left(s,\bm{S}_{\bm{\myin}},\bm{F}\right)=-\frac{1}{n}\Bigl( \mu(s)+(1-s)\mu'(s) \Bigr. \\ \Bigl. - (1-s)\sqrt{2\mu''(s)} - \log 8 \Bigr). 
\end{equation*}
This expression allows us to write $\mu'(s)$ in \eqref{eq:cond1} in terms of $R_n\left(s,\bm{S}_{\bm{\myin}},\bm{F}\right)$ so that, taking the logarithm in equation \eqref{eq:cond1}, our conditions can be rewritten as either
\begin{equation*}
R<R_n\left(s,\bm{S}_{\bm{\myin}},\bm{F}\right)
\end{equation*}
or
\begin{equation*}
\log \frac{1}{\Pemax} < -\frac{\mu(s)}{1-s}-\frac{s}{1-s} n R_n\left(s,\bm{S}_{\bm{\myin}},\bm{F}\right)\\
+2s\sqrt{2\mu''(s)} +\frac{\log 8}{1-s}.
\end{equation*}

At this point, we exploit the fact that we are considering a fixed composition code. Since we want our result to depend only on the composition $P$ and not on the particular sequence $\bm{\myin}$, we choose $\bm{F}$ so that the function $\mu(s)$ also only depends on the composition $P$. We thus choose $\bm{F}$ to be the $n$-fold tensor power of a state $F$ in $\mathcal{H}$, that is $\bm{F}=F^{\otimes n}$.
With this choice, in fact, we easily check that, if $\bm{\myin}=(\myin_1,\myin_2,\ldots,\myin_n)$,
\begin{align*}
\mu(s) & =  \log \Tr \bm{S}_{\bm{\myin}}^{1-s} \bm{F}^s\\
 & =  \log \prod_{i=1}^n \Tr S_{\myin_i}^{1-s} F^s \\
 & =  \log \prod_{\myin} \left(\Tr S_{\myin}^{1-s} F^s\right)^{n P(\myin)}\\
 & =  n \sum_\myin P(\myin) \log\left(\Tr S_{\myin}^{1-s} F^s\right)\\
 & =  n \sum_\myin P(\myin) \mu_{S_{\myin},F}(s).
\end{align*}
It is useful to recall that since we assumed the supports of $\bm{F}$ and  $\bm{S}_{\bm{\myin}}$ to be non-disjoint, the supports of $S_\myin$ and $F$ must be non-disjoint if $P(\myin)>0$, and thus all terms in the sum are well defined.
Note that we also have
\begin{align*}
\mu'(s)  & =  n\sum_\myin P(\myin) \mu_{S_\myin,F}'(s),\\
\mu''(s) & =  n\sum_\myin P(\myin) \mu_{S_\myin,F}''(s).
\end{align*}
With the same procedure used to obtain \eqref{eq:mu''} using the Nussbaum-Szko{\l}a mapping \eqref{eq:Pdef}, we see that for fixed $F$ and $s\in(0,1)$, $\mu_{S_\myin,F}''(s)$ is a variance of a finite random variable and it is thus non-negative and bounded by a constant for all $\myin$. This also implies that $\mu''(s)$ is bounded by a constant.

The essential point here is that the contribution of $\mu(s)$ and $\mu'(s)$ in our bounds will grow linearly in $n$, while the contribution of $\mu''(s)$ will only grow with $\sqrt{n}$. Hence, the terms involving $\mu''(s)$ in the above equations will not have any effect on the first order exponent of the bounds. A formalization of this fact, however, is tricky. In \cite{shannon-gallager-berlekamp-1967-1} the effect of $\mu''(s)$ in the classical case is dealt with by exploiting the fact that $\mu_{\myin,F}''(s)$ is a variance and proving that, uniformly over $s$ and $F$, $s\sqrt{\mu_{\myin,F}''(s)}\leq \log(e/\sqrt{P_{\text{min}}})$, where $P_{\text{min}}$ is the smallest non-zero transition probability of the channel. This allows the authors to obtain a bound valid for finite $n$. In our case, this procedure appears to be more complicated. If $\mu_{\myin,F}(s)$ is studied in the quantum domain of operators $S_\myin$ and $F$, then $\mu_\myin''(s,F)$ is not a variance, and thus a different approach must be studied; if $\mu_{\myin,F}(s)$ is studied by means of the Nussbaum-Szko{\l}a mapping, then in \eqref{eq:mu''} both $Q_0$ and $Q_1$ vary when $F$ varies, and thus there is no such $P_{\text{min}}$ to be used. For this reason, we need to take a different approach and we content ourselves with finding a bound on $E(R)$ using the asymptotic regime $n\to \infty$.

Simplifying again the notation in light of the previous observations, let us write $R_n(s,P,F)$ for $ R_n(s,\bm{S}_{\bm{\myin}},\bm{F})$. Using the obtained expression for $\mu(s)$, our conditions are 
 either
\begin{equation}
R<R_n(s,P,F)
\label{eq:condR<R(sqf)}
\end{equation}
or
\begin{equation*}
\frac{1}{n}\log \frac{1}{\Pemax} < -\frac{1}{1-s}\sum_\myin P(\myin) \mu_{S_\myin,F}(s)\\-\frac{s}{1-s} R_n(s,P,F)\\
+\frac{1}{n}\left(2s\sqrt{2\mu''(s)} +\frac{\log 8}{1-s}\right).
\label{eq:cond1/nPemax<..sumqkmuk}
\end{equation*}

Now, given a rate $R$, we want to bound $\Pemax$ for all codes with rate at least $R$. We can fix first the composition of the code, bound the probability of error, and then find the best possible composition. 
Since we can choose $s$ and $F$, for a given $R$ and $P$, we will choose them so that the first inequality is not satisfied, which will imply that the second one is, thus bounding $\Pemax$.

The point here is that we are free to choose $s$ and $F$, but we then need to optimize the composition $P$ in order to have a bound valid for all codes.
This direct approach, even in the classical case, turns out to be very complicated (see \cite[Sec. 9.3 and 9.4, pag. 188-303]{fano-book} for a detailed and however instructive analysis). The authors in \cite{shannon-gallager-berlekamp-1967-1} thus proceed in a more concise way by stating the resulting optimal $F$ and $P$ as a function of $s$ and then proving that this choice leads to the desired bound. Here, we will follow this approach showing that the same reasoning can be applied also to the case of quantum channels. It is worth noticing that the price for using the concise approach of \cite{shannon-gallager-berlekamp-1967-1} is that, contrarily to the approach in \cite{fano-book}, it does not allow us to derive tight bounds for constant composition codes with non optimal composition $P$.

It is important to point out that it is not possible to simply convert the quantum problem to the classical one using the Nussbaum-Szko{\l}a mapping \eqref{eq:Pdef_first} directly on the states $S_\myin$ and $F$ and then using the construction of \cite[eqs. (4.18)-(4.20)]{shannon-gallager-berlekamp-1967-1} on the obtained classical distributions. In fact, in \eqref{eq:Pdef_first}, even if one of the two states is kept fixed and only the other one varies, \emph{both} distributions vary. Thus, even if $F$ is kept fixed, the effect of varying $S_\myin$ for the different values of $\myin$ would not be compatible with the fact that in \cite[eq. (4.20)]{shannon-gallager-berlekamp-1967-1} a fixed $\mathbf{f}_s$ (in their notation) has to be used which cannot depend on $\myin$. 
Fortunately, it is instead possible to exactly replicate the steps used in \cite{shannon-gallager-berlekamp-1967-1}  by correctly reinterpreting the construction of $F$ and $P$ in the quantum setting.

For any $s$ in the interval $0<s<1$, define 
\begin{equation}
A(s,P) = \sum_\myin P(\myin) S_\myin^{1-s}.
\label{eq:alphas_gen_p}
\end{equation}
Let then $P_s$ be the distribution that minimizes the expression
\begin{equation}
\Tr A(s,P) ^{1/(1-s)},
\label{eq:tominimize}
\end{equation}
which surely admits a minimum in the simplex of probability distributions. Finally define
\begin{align}
A_s &= A(s,P_s)\\
&= \sum_\myin P_s(\myin) S_\myin^{1-s}.
\label{eq:alphas}
\end{align}

As observed by Holevo\footnote{The variable $s$ in \cite{holevo-2000} corresponds to our $s/(1-s)$, that we call $\rho$ here in accordance with the consolidated classical notation.} \cite[eq. (38)]{holevo-2000}, the distribution $P_s$ that achieves the minimum in \eqref{eq:tominimize} satisfies the conditions 
\begin{equation}
\Tr \left(S_\myin^{1-s}A_s ^{s/(1-s)}\right)\geq \Tr \left( A_s ^{1/(1-s)}\right), \qquad \forall \myin.
\label{eq:Lagr-cond}
\end{equation}

 Furthermore, equation \eqref{eq:Lagr-cond} is satisfied with equality for those $\myin$ with $P_s(\myin)>0$, as can be verified by multiplying it by $P_s(\myin)$ and summing over $\myin$.
Then, we define
\begin{equation}
F_s=\frac{A_s^{1/(1-s)}}{\Tr A_s^{1/(1-s)}}.
\label{def:fs}
\end{equation}
Since we can choose $s$ and $F$ freely, we will now tie the operator $F$ to the choice of $s$, using $F_s$  for $F$. 
We only have to keep in mind that $\mu'(s)$ and $\mu''(s)$ are computed by holding $F$ fixed. The distribution $P_s$ will instead be used later. Note further that we fulfill the requirement that $F$ and $S_\myin$ have non-disjoint supports, since the left hand side in \eqref{eq:Lagr-cond}  must be positive for all $\myin$.

As in \cite[eqs (4.21)-(4.22)]{shannon-gallager-berlekamp-1967-1}, we see that, using $F_s$ in place of $F$ in the definition of $\mu_{S_\myin,F}(s)$, we get
\begin{equation*}
\mu_{S_\myin,F_s}(s)=\log \Tr\left(S_\myin^{1-s}A_s^{s/(1-s)}\right)-s\log\Tr A_s^{1/(1-s)}
\end{equation*}
Using \eqref{eq:Lagr-cond} we then see that
\begin{align}
\mu_{S_\myin,F_s}(s) & \geq  (1-s)\log \Tr A_s^{1/(1-s)}\\ 
 & =  -(1-s)E_0\left(\frac{s}{1-s},P_s\right)\\
 & =  -(1-s)E_0\left(\frac{s}{1-s}\right),
 \label{eq:E0s/(1-s)}
\end{align}
with equality if $P_s(\myin)>0$. 
Here, we have used equation \eqref{eq:alphas}, the definitions \eqref{eq:E0rhoq} and \eqref{eq:E0rho}, and the the fact that $P_s$ minimizes \eqref{eq:tominimize}. Thus, with the choice of $F=F_s$, equations \eqref{eq:condR<R(sqf)} and \eqref{eq:cond1/nPemax<..sumqkmuk} can be rewritten as
(for each $s$) either
\begin{equation}
R<R_n(s,P,F_s)
\label{eq:lastR}
\end{equation}
or
\begin{equation}
\frac{1}{n}\log \frac{1}{\Pemax} < E_0\left(\frac{s}{1-s}\right) 
-\frac{s}{1-s} R_n(s,P,F_s)
\\+\frac{2s\sqrt{2}}{\sqrt{n}}\sqrt{\sum_\myin P(\myin) \mu_{S_\myin,F_s}''(s) } +\frac{\log 8}{(1-s)n},
\label{eq:lastPmax}
\end{equation}
where
\begin{equation}
R_n(s,P,F_s)=- \sum_\myin P(\myin)\left[ \mu_{S_\myin,F_s}(s)+(1-s) \mu_{S_\myin,F_s}'(s)\right]\\
+\frac{1}{\sqrt{n}}(1-s)\sqrt{2\sum_\myin P(\myin) \mu_{S_\myin,F_s}''(s)} + \frac{1}{n} \log 8.
\label{eq:defRnfinal}
\end{equation}

Now, for a fixed $R$, we can choose $s$ and then use the two conditions.
Dealing with these equations for a fixed code is more complicated in our case than in \cite{shannon-gallager-berlekamp-1967-1} due to the fact that we have not been able to bound uniformly the second derivatives $\mu_{S_\myin,F_s}''(s)$ for $s\in (0,1)$. Thus, we have to depart from \cite{shannon-gallager-berlekamp-1967-1}. Instead of considering a fixed code of block length $n$, consider sequences of codes of increasing block-length. From the definition of $E(R)$, there exists a sequence of codes of block-lengths $n_1,n_2,\ldots,n_k\ldots,$ with rates $R_{n_1},R_{n_2},\ldots,R_{n_k},\ldots$ such that $R=\lim_{k\to\infty} R_{n_k}$, and with probabilities of error $\Pemax^{(n_1)}, \Pemax^{(n_2)},\ldots,\Pemax^{(n_k)},\ldots$ such that
\begin{equation*}
E(R)=\lim_{k\to\infty} -\frac{1}{n_k}\log \Pemax^{(n_k)}.
\end{equation*}
Each code of the sequence will in general have a different composition\footnote{With some abuse of notation, we now use $P_n$ where $n$ is an index for the sequence and obviously does not have anything to do with the $s$ of $P_s$.} $P_n$. Since the compositions $P_n$ are in a compact set, there exists a subsequence of codes such that $P_n$ converges to, say, $\bar{P}$. Thus, we can directly assume this subsequence is our own sequence, remove the double subscript $n_k$, and safely assume that  $P_n\to\bar{P}$, $R_n\to R$ and $-\frac{1}{n}\log \Pemax^{(n)}\to E(R)$ as $n\to\infty$.

Observe that, since $\mu_{S_\myin,F_s}(s)$ is a non-positive convex function of $s$ for all $s\in (0,1)$, we have  \cite[Fig. 6]{shannon-gallager-berlekamp-1967-1}
\begin{align*}
\mu_{S_\myin,F_s}(s)+ (1-s) \mu_{S_\myin,F_s}'(s) & \leq \mu_{S_\myin,F_s}(1^-)\\ & \leq 0,
\end{align*}
which implies that $R_n(s,P_n,F_s)$ is a non-negative quantity.
Then, using a procedure similar to that used in \cite[pag. 100-102]{shannon-gallager-berlekamp-1967-1}, it is proved in Appendix \ref{appendix:cont_R} that $R_n(s,P_n,F_s)$ is a continuous function of $s$ in the interval $0<s<1$. 
Thus, for the rate $R_n$ of the code with block-length $n$, we can only have three possible cases:

\begin{enumerate}
\item $R_n>R_n(s,P_n,F_s) \quad \forall s \in (0,1)$;
\item $R_n<R_n(s,P_n,F_s) \quad \forall s \in (0,1)$;
\item $R_n=R_n(s,P_n,F_s)\quad$for some $s$ in $(0,1)$.
\end{enumerate}
For each $n$, one of the three possible cases above is satisfied and at least one\footnote{For some channels at certain rates, there could be more than one case which is satisfied infinitely often, depending on how $R_n$ is made to converge to $R$. This happens for example for those channels with $C_0=C$ whenever $R=C$. Note that this does not impact the correctness of the proof.} of the cases is then satisfied for infinitely many values of $n$.

Suppose thus that case 1) above is satisfied infinitely often. Then, for any \emph{fixed} $s\in (0,1)$ we have $R_n(s,P_n,F_s)< R_n$ infinitely often. 
We can focus on the subsequence of codes with this property and use it as our sequence, so that $R_n(s,P_n,F_s)< R_n$ for all $n$.
For these codes, since condition \eqref{eq:lastR} is not satisfied, then \eqref{eq:lastPmax} must be satisfied. Since $s\in (0,1)$ is fixed, we can make $n\to\infty$ so that the last two terms on the right hand side of \eqref{eq:lastPmax} vanish. In the limit, since $R_n(s,P_n,F_s)$ is non-negative, we get
\begin{align}
E(R) & \leq E_0\left(\frac{s}{1-s}\right) 
-\frac{s}{1-s} R_n(s,P_n,F_s)\\
& \leq E_0\left(\frac{s}{1-s}\right) 
\label{eq:E(R)to0}
\end{align}
for all $s\in(0,1)$. 
Letting then $s\to 0$ we find, using \eqref{eq:E0rhoq}, that $E(R)\leq E_0(0)=0$. Thus, surely $E(R)\leq \Esp(R)$ proving the theorem in this case.

Suppose now that case 2) above is satisfied infinitely often. Then, for any \emph{fixed} $s\in (0,1)$ we have $R_n<R_n(s,P_n,F_s)$ infinitely often. 
We can focus on the subsequence of codes with this property and use it as our sequence, so that $R_n<R_n(s,P_n,F_s)$ for all $n$.
As $n\to\infty$, since $s$ is fixed and $P_n\to\bar{P}$, we see from 
\eqref{eq:defRnfinal} that, 
\begin{equation}
R_n(s,P_n,F_s)=\\- \sum_\myin  \bar{P}(\myin) \left[ \mu_{S_\myin,F_s}(s)+ (1-s) \mu_{S_\myin,F_s}'(s)\right]+o(1).
\label{eq:Rnlimit}
\end{equation}
Since $R_n\to R$, the inequality $R_n<R_n(s,P_n,F_s)$ leads then, in the limit $n\to\infty$, to
\begin{equation*}
R  \leq  - \sum_\myin \bar{P}(\myin)\left[ \mu_{S_\myin,F_s}(s)+ (1-s) \mu_{S_\myin,F_s}'(s)\right].
\end{equation*}
Now, by the fact that $\mu_{S_\myin,F}(s)$ is convex and non-positive for all $F$, it is possible to observe that $\mu_{S_\myin,F_s}(s)-s\mu_{S_\myin,F_s}'(s)\leq 0$, which implies that $-\mu_{S_\myin,F_s}'(s)\leq -\mu_{S_\myin,F_s}(s)/s$. Thus, for all $s\in(0,1)$, 
\begin{align*}
R & \leq  \sum_\myin \bar{P}(\myin) \left(- \frac{1}{s}\mu_{S_\myin,F_s}(s)\right)\\
& \leq  \frac{1-s}{s}E_0\left(\frac{s}{1-s}\right)
\end{align*}
where, in the last step, we have used \eqref{eq:E0s/(1-s)}. Calling now $\rho=s/(1-s)$, we find that for all $\rho>0$
\begin{equation*}
R  \leq  \frac{E_0(\rho)}{\rho}.
\end{equation*}
Hence, for any $\varepsilon>0$, we find
\begin{align*}
\Esp\left(R - \varepsilon\right) & =  \sup_{\rho>0}\left(E_0\left(\rho\right) -\rho \left( R-\varepsilon\right) \right)
 \\   & \geq  \sup_{\rho>0} (\rho\, \varepsilon ).
\end{align*}
This means that $\Esp(R-\varepsilon)$ is unbounded for any $\varepsilon>0$, which obviously implies that $E(R)\leq \Esp(R-\varepsilon)$ for all positive $\varepsilon$, proving the theorem in this case.

Suppose finally  that case 3) is satisfied infinitely often. Thus, for infinitely many $n$, there is a $s_n\in(0,1)$ such that $R_n=R_n(s_n,P_n,F_{s_n})$. Since the values $s_n$ are in the interval $(0,1)$, there must exist an accumulating point for the $s_n$ in the closed interval $[0,1]$. We will first assume that such an accumulating point $\bar{s}$ exists satisfying $0<\bar{s}<1$. A subsequence of codes then exists with the $s_n$ tending to $\bar{s}$. Let this subsequence be our new sequence. For these codes, since condition \eqref{eq:lastR} is not satisfied for $s=s_n$, then \eqref{eq:lastPmax} must be satisfied with $s=s_n$.
We can first substitute $R_n(s_n,P_n,F_{s_n})$ with $R_n$ in \eqref{eq:lastPmax}. Letting then $n\to\infty$, we find that $R_n\to R$ and the last two terms on the right hand side of $\eqref{eq:lastPmax}$ vanish, since $\mu_{S_\myin,F_{s_n}}''(s)$ is bounded for $s_n$ sufficiently close to $\bar{s}\neq0,1$. Hence, we obtain in the limit
\begin{align*}
E(R) & \leq   E_0\left(\frac{\bar{s}}{1-\bar{s}}\right) 
-\frac{\bar{s}}{1-\bar{s}} R \\
& \leq   \sup_{\rho\geq0}\left(E_0\left(\rho\right) 
-\rho R\right)\\
& =  \Esp(R).
\end{align*} 

If the only accumulating point for the $s_n$ is $\bar{s}=1$ or $\bar{s}=0$, the above procedure cannot be applied since we cannot get rid of the last two terms in \eqref{eq:lastPmax} by letting $n\to\infty$, because we have not bounded $\mu_{S_\myin,F_{s}}''(s)$ uniformly over $s\in(0,1)$, and it may well be that $\mu_{S_\myin,F_{s_n}}''(s_n)$ is unbounded as $s_n$ approaches $0$ or $1$. These cases, however, can be handled with the same procedure used for cases 1) and 2) above. 
Assume that case 3) is satisfied infinitely often with the only accumulating points $\bar{s}=0$ or $\bar{s}=1$ for the $s_n$. Consider again the appropriate sub-sequence of codes as our sequence. Fix any $\delta>0$. For  $n$ large enough, there is no $s$ in the closed interval $ [\delta,1-\delta]$ which satisfies $R_n=R_n(s,P_n,F_s)$. Hence, for all $n$ sufficiently large, one of the two following conditions must be satisfied
\begin{enumerate}
\item [1')] $R_n>R_n(s,P_n,F_s) \quad \forall s \in [\delta,1-\delta]$;
\item [2')] $R_n<R_n(s,P_n,F_s) \quad \forall s \in [\delta,1-\delta]$.
\end{enumerate}
One of the two conditions must then be satisfied infinitely often. If condition 1') is satisfied infinitely often, we can use the same procedure used for condition 1), with a fixed $s\in [\delta,1-\delta]$, to obtain
\begin{equation*}
E(R)  \leq E_0\left(\frac{s}{1-s}\right), \quad s\in [\delta,1-\delta].
\end{equation*}
Since $\delta$ is arbitrary, we can let $\delta\to 0$ and then also let $s\to 0$, deducing again $E(R)\leq 0$.
If condition 2') is satisfied infinitely often, we can repeat the procedure used for case 2),  with a fixed $s\in [\delta,1-\delta]$, to obtain 
\begin{equation*}
R\leq \frac{1-s}{s}E_0\left(\frac{s}{1-s}\right), \quad s\in [\delta,1-\delta],
\end{equation*}
and then
\begin{align*}
\Esp\left(R - \varepsilon\right) & \geq \sup_{0<\rho \leq (1-\delta)/\delta } (\rho\, \varepsilon)\\
& =\frac{(1-\delta)\varepsilon}{\delta}.
\end{align*}
Letting then $\delta\to 0$ we prove again that $E_{sp}(R-\varepsilon)$ is unbounded for any $\varepsilon>0$. This concludes the proof of the theorem.
\end{IEEEproof}

In the high rate region, the obtained expression for the upper bound to the reliability function coincides with the random coding  expression which is respectively proved and conjectured to represent a lower bound to the reliability of pure-state and mixed-state  channels.
Thus, the sphere packing bound is proved to be tight in the high rate region for pure-state channels and we may as well conjecture that it is tight in the general case.

As in the classical case, however, the sphere packing bound is provably not tight in the low rate region. For channels with no zero-error capacity, it is possible to obtain tighter bounds by extending some classical bounds to the classical-quantum setting. These channels are considered in Section \ref{sec:Zero-rate}. For channels with a zero-error capacity, instead, contrarily to the classical case, the sphere packing bound has in the classical-quantum case some important properties that relates it to the Lov\'asz theta function. The next section gives a min-max characterization of the function $\Esp(R)$ which clarifies this relation. For this purpose, it is convenient now to use the R\'enyi divergence $D_{\alpha}(\cdot,\cdot)$ in place of $\mu(s)$.

\section{Information Radii, Zero-Error Capacity and the Lov\'asz Theta Function}
\label{sec:Inform-radii}

It is known \cite{csiszar-korner-book} that the capacity of a classical channel can be written as an information radius in terms of the Kullback-Leibler divergence according to the expression
\begin{equation}
C=\min_Q \max_x D(\myChv{\myin}||Q),
\label{eq:minmaxC}
\end{equation}
where the minimum is over all distributions $Q$ on the output alphabet $\myOut$.
As made clear by Csisz\'ar \cite{csiszar-1995}, a similar expression holds for what has sometimes been called \emph{generalized capacity} or \emph{generalized cutoff rate}\footnote{Since we also consider the zero-error capacity $C_0$ of channels in this paper,  we prefer to avoid any reference to generalized capacities in the sense of  \cite{csiszar-1995}. Instead, we prefer to adopt the notation of Savage \cite[eq. (15)]{savage-1966}. In light of Arikan's results \cite{arikan-1988}, this may also be more appropriate than the notation in \cite[eq. (5)]{jacobs-berlekamp-1967}.) 
} (see \cite{csiszar-1995} for a detailed discussion).
The function $E_{\text{sp}}(R)$ equals the upper envelope of all the lines $E_0(\rho)-\rho R$, and it is useful to define the quantity 
\begin{equation}
R_\rho=\frac{E_0(\rho)}{\rho},
\label{eqdef:R_rho}
\end{equation}
which is the value at which each of these lines meets the $R$ axis.
As explicitly pointed out in \cite[Prop. 1]{csiszar-1995}, equation \eqref{eq:minmaxC} can be generalized using the R\'enyi divergence defined in \eqref{eqdef:D_alpha} to show that
\begin{equation}R_\rho=\min_{Q}\max_x D_\alpha(\myChv{\myin}||Q), \quad \alpha=1/(1+\rho).
\label{eq:minmaxRrho}
\end{equation}
It is important to remark that, in the classical case, this property of the function $\Esp(R)$ was already observed in \cite[eq. (4.23)]{shannon-gallager-berlekamp-1967-1} even if stated in different terms, and it is essentially this property that is used in the proof of the sphere-packing bound. Many related min-max characterizations of $\Esp(R)$, which give a more complete picture, were given in \cite{blahut-1974}.

Using the known properties of the R\'enyi divergence (see \cite{csiszar-1995}), we find that when $\rho\to 0$ the above expression (with $\alpha\to 1$) gives the already mentioned expression for the capacity \eqref{eq:minmaxC}, while for $\rho\to\infty$ we obtain 
\begin{equation*}
R_\infty=\min_{Q}\max_x \left[-\log \sum_{y:W(y|x)>0}Q(y)\right],
\end{equation*}
which is the dual formulation of \eqref{eq:Rinfty_class}.

As already observed in \cite[eqs. (1)-(2)]		{korner-orlitsky-1998}, 
it is evident that there is an interesting similarity between the min-max expression for the ordinary capacity $C$ (and of $R_\rho$ in general) and the definition of the Lov{\'a}sz theta function. 
In this section, we show that this similarity is not a simple coincidence. By extending relation \eqref{eq:minmaxRrho} to general classical-quantum channels, we show that Lov{\'a}sz' bound to $C_0$ emerges naturally, in that context, as a consequence of the bound $C_0\leq R_\infty$. 

\begin{theorem}
\label{th:Rrho}
For a classical-quantum channel with states $\{S_x\}$, $x\in\mathcal{X}$, and $\rho>0$, the rate $R_\rho$ defined in \eqref{eqdef:R_rho} satisfies
\begin{equation*}
R_\rho=\min_{F}\max_{x} D_\alpha(S_x||F), \quad  \alpha=1/(1+\rho),
\end{equation*} 
where the minimum is over all density operators $F$.
\end{theorem}
\begin{IEEEproof}
Setting $\alpha=1/(1+\rho)$, we can write
\begin{equation*}
R_\rho =  \max_P \frac{1}{\alpha-1}\log\left[ \Tr\left( \sum_x P(x) S_x^{\alpha} \right)^{1/\alpha}\right]^{\alpha}
\end{equation*}
and, defining according to \eqref{eq:alphas_gen_p} $A(1-\alpha,P) =\sum_x P(x) S_x^{\alpha}$, we can write
\begin{equation}
R_\rho =  \max_P \frac{1}{\alpha-1}\log \| A(1-\alpha,P) \|_{1/\alpha},
\label{eq:rrhoins}
\end{equation}
where $\|\cdot\|_{r}$ is the Schatten $r$-norm. From the H\"older inequality we know that, for any  positive operators $A$ and $B$, we have
\begin{equation*}
\|A\|_{1/\alpha}\|B\|_{1/(1-\alpha)}\geq \Tr(AB),
\end{equation*}
with equality if and only if $B=\gamma A^{1-1/\alpha}$ for some scalar coefficient $\gamma$. Thus we can write
\begin{equation*}
\|A\|_{1/\alpha}=\max_{\|B\|_{1/(1-\alpha)}\leq 1} \Tr(AB),
\end{equation*}
where $B$ runs over positive operators in the unit ball in the $(1/(1-\alpha))$-norm.
Using this expression for the Schatten norm we obtain
\begin{align*}
R_\rho & =  \max_P  \frac{1}{\alpha-1}\log \max_{\|B\|_{1/(1-\alpha)}\leq 1} \Tr( A(1-\alpha,P) B)\\
& =\frac{1}{\alpha-1}\log  \min_P \max_{\|B\|_{1/(1-\alpha)}\leq 1} \Tr\left( \sum_x P(x) S_x^{\alpha} B \right).
\end{align*}
In the last expression, the minimum and the maximum are both taken over convex sets and the objective function is linear both in $P$ and $B$. Thus, we can interchange the order of maximization and minimization to get
\begin{align*}
R_\rho & = \frac{1}{\alpha-1}\log \max_{\|B\|_{1/(1-\alpha)}\leq 1}\min_P  \sum_x P(x) \Tr\left(S_x^{\alpha} B \right) \\
&= \frac{1}{\alpha-1}\log\max_{ \|B\|_{1/(1-\alpha)}\leq 1}\min_x  \Tr\left(S_x^{\alpha} B \right).
\end{align*}
Now, we note that the maximum over $B$ can always be achieved by a positive operator, since all the $S_x^\alpha$ are positive operators. Thus, we can change the dummy variable $B$ with $F=B^{1/(1-\alpha)}$, where $F$ is now a positive operator constrained to satisfy $\|F\|_1\leq 1$, that is, it is a density operator. Using $F$, we get
\begin{align*}
R_\rho & = \frac{1}{\alpha-1}\log \max_{F}\min_{x}  \Tr\left(S_x^{\alpha} F^{1-\alpha} \right) \\
& =  \min_{F}\max_{x} \frac{1}{\alpha-1} \log \Tr\left(S_x^{\alpha} F^{1-\alpha} \right)\\
& =  \min_{F}\max_{x} D_\alpha (S_x||F).
\end{align*}
where $F$ now runs over all density operators. 
\end{IEEEproof}

If all operators $S_x$ commute, which means that the channel is classical, then the optimal $F$ is diagonal in the same basis where the $S_x$ are, and we thus recover Csisz\'ar's expression for the classical case. Furthermore, for $\rho\to 0$ (that is, $\alpha\to 1$) we obtain the expression of the capacity as an information radius 
already established for classical-quantum channels \cite{hayashi-nagaoka-2003}. When $\rho=1$ (that is, $\alpha=1/2$) then, we obtain an alternative expression for the so called quantum cut-off rate \cite{ban-kurokawa-hirota-1998}. 

The most important case in our context, however, is the case when $\rho\to\infty$, that is, $\alpha\to 0$, for which we obtain
\begin{equation*}
R_\infty=\lim_{\rho\to \infty} R_\rho.
\end{equation*}
Let $S_\myin^0$ be the projector onto the support of $S_\myin$. 
We first point out that, by letting $\alpha\to 0$ in equation \eqref{eq:rrhoins}, we can write
\begin{align}
R_\infty & =\max_P \left[-\log \lambda_{\max} \left(\sum_\myin P(\myin) S_\myin^0\right)\right]\\
& = -\log \min_{P} \lambda_{\max} \left(\sum_\myin P(\myin) S_\myin^0\right).
\label{eq:rinflambdaminmixed}
\end{align}
We will not use this expression for now. We only point out that we will use this expression in the next section for the particular case of pure-state channels. Furthermore, the above expression shows that finding $R_\infty$ for a given channel is a so called \emph{eigenvalue problem}, a well known special case of linear matrix inequality problems, which can be efficiently solved by numerical methods \cite{boyd-book-1994}. If the states $S_\myin$ commute, the channel is classical and this problem is already known to reduce to a linear programming one.

Here, however, we proceed by using expressions that make more evident the relation with the Lov\'asz' theta function. Taking the limit $\rho\to\infty$ (that is, $\alpha\to 0$) in Theorem \ref{th:Rrho} and using definition \eqref{eqdef: Quantum_D_rho} we obtain
\begin{equation}
R_\infty=\min_{F}\max_{x} \log\frac{1}{\Tr\left(S_x^0 F \right)},
\label{eq:minmaxQRinfty}
\end{equation}
where the minimum is again over all density operators $F$. Note that the argument of the min-max in \eqref{eq:minmaxQRinfty} coincides with $D_{\min}(S_x|| F)$ according to the definition of $D_{\min}$ introduced in \cite{datta2009}.
The analogy with the Lov{\'a}sz theta function becomes evident if we consider the particular case of pure-state channels. If a channel has pure states  $S_x=\ket{\psi_x}\bra{\psi_x}$, the set $\{\psi_x\}$ is a valid representation of the confusability graph of the channel in Lov\'asz' sense. On the other hand, any representation in Lov\'asz' sense can be interpreted as a pure-state channel. Consider for a moment the search for the optimum $F$ in \eqref{eq:minmaxQRinfty} when restricted to rank-one operators, that is $F=\ket{f}\bra{f}$. We see that in this case we can write $\Tr(S_x^0F)=|\braket{\psi_x}{f}|^2$ and, so, we obtain precisely the value $V(\{\psi_x\})$. When searching over all possible $F$  we thus have
\begin{equation}
R_\infty\leq V(\{\psi_x\}).
\label{eq:RinftyvsV}
\end{equation}
Hence, we see that Lov{\'a}sz' bound $C_0\leq V(\{\psi_x\})$ can be deduced as a consequence of the inequality $C_0\leq R_\infty$. 

Now, as in the classical case, different classical-quantum channels share the same confusability graph and thus have the same zero-error capacity $C_0$. Hence, the best upper bound that we can obtain for $C_0$ is not in general obtained using the rate $R_\infty$ of the original channel but the rate $R_\infty$ of some auxiliary channel. It is then preferable to focus on confusability graphs. For a given graph $G$, we can then consider as a representation of the graph any set of states $\{S_x\}$ with $S_\myin S_\mysin=0$ if $\myin,\mysin$ are not connected in $G$. The role of the \emph{value} is then played by the rate $R_\infty$ associated to the states $\{S_\myin\}$, that we may denote as $R_\infty(\{S_\myin\})$. We can then define, in analogy with the theta function, the quantity
\begin{align*}
\vartheta_{\text{sp}} & = \min_{\{S_\myin\}} R_\infty(\{S_\myin\})\\
& =\min_{\{S_\myin\}}\min_{F}\max_{\myin} \log\frac{1}{\Tr\left(S_\myin^0 F \right)}\\
& =\min_{\{U_\myin\}}\min_{F}\max_{\myin} \log\frac{1}{\Tr\left(U_x F \right)},
\end{align*}
where $\{S_\myin\}$ runs over the sets of operators such that $S_\myin S_\mysin=0$ if symbols $\myin$ and $\mysin$ cannot be confused and $\{U_\myin\}$ runs over the sets of \emph{projectors} with that same property. 
Then we have the bound $C_0\leq \vartheta_{\text{sp}}$. We will later show that, in fact, $\vartheta_{\text{sp}}=\vartheta$, that is, when the rate $R_\infty$ is minimized over all channels compatible with the confusability graph, we obtain precisely the Lov\'asz theta function. In order to add some insight to this equivalence, however, we first give a self contained proof of the inequality $C_0\leq \vartheta_{\text{sp}}$ which does not involve the sphere packing bound and is obtained by ``generalizing'' Lov\'asz' approach.

\begin{theorem}
\label{th:mytheta}
For any graph $G$, we have
\begin{equation*}
C_0\leq \vartheta_{\text{sp}}\leq \vartheta.
\end{equation*}
\end{theorem}
\begin{IEEEproof}
The fact that $\vartheta_{\text{sp}}\leq \vartheta$ is obvious, since Lov\'asz' $\vartheta$ is obtained by restricting the minimization in the definition of $\vartheta_{\text{sp}}$ to pure states $S_\myin=\ket{\psi_\myin}\bra{\psi_\myin}$ and ``handle'' $F$. The inequality $C_0\leq \vartheta_{\text{sp}}$ derives from the sphere packing bound, but we prove it here adapting Lov\'asz' argument to the more general situation where general projector operators are used for the representation in place of the one-dimensional vectors, and a general density operator is used as the handle.

Let $\{U_\myin\}$ be a set of projectors for $G$ which achieves $\vartheta_{\text{sp}}$. To a sequence of symbols $\bm{\myin}=(\myin_1,\myin_2,\ldots,\myin_n)$, associate the operator (projector) $\bm{U}_{\bm{\myin}}=U_{\myin_1}\otimes U_{\myin_2}\cdots\otimes U_{\myin_n}$.
Consider then a set of $M$ non-confusable codewords of length $n$, $\bm{\myin}_1,\ldots,\bm{\myin}_M$, and their associated projectors $\bm{U}_{\bm{\myin}_1},\ldots,\bm{U}_{\bm{\myin}_M}$. Then, for $m\neq m'$, we have
\begin{align*}
\Tr(\bm{U}_{\bm{\myin}_m}\bm{U}_{\bm{\myin}_{m'}}) & =\prod_{i=1}^n\Tr(U_{\myin_{m,i}}U_{\myin_{m',i}})\\
& =0,
\end{align*}
since, for at least one value of $i$, $U_{\myin_{m,i}}U_{\myin_{m',i}}=0$ because codewords $\bm{\myin}_m,\bm{\myin}_{m'}$ are not confusable. Hence, since the states $\{\bm{U}_{\bm{\myin}_m}\}$ are orthogonal projectors, we have
\begin{equation*}
\sum_{m=1}^{M}\bm{U}_{\bm{\myin}_m}\leq \mathds{1},
\end{equation*}
where $\mathds{1}$ is the identity operator. Consider now the state $\bm{F}=
F^{\otimes n}$, where $F$ achieves $\vartheta_{\text{sp}}$ for the representation $\{U_x\}$. Note that, for each $m$, we have
\begin{align*}
\Tr (\bm{U}_{\bm{\myin}_m} \bm{F}) & = \prod_{i=1}^n \Tr (U_{\myin_{m,i}} F)\\
& \geq e^{-n \vartheta_{\text{sp}}}.
\end{align*}

Hence, we have 
\begin{align*}
1 & =   \Tr(\bm{F}) \\
& \geq  \sum_{m=1}^M \Tr (\bm{U}_{\bm{\myin}_m} \bm{F})\\
& \geq  Me^{-n \vartheta_{\text{sp}}}.
\end{align*}
Thus, we deduce that $M\leq e^{n\vartheta_{\text{sp}}}$. This implies that the rate of any zero-error code is not larger than $\vartheta_{\text{sp}}$, proving that $C_0\leq \vartheta_{\text{sp}}$.
\end{IEEEproof}

The following result is due to Schrijver \cite{schrijver2013}.
\begin{theorem}
For any graph $G$
\begin{equation*}
\vartheta_{\text{sp}}=\vartheta.
\end{equation*}
\label{th:schrijver}
\end{theorem}
\begin{IEEEproof}
Since we already know that $\vartheta_{\text{sp}}\leq\vartheta$, we only need to prove the inequality $\vartheta_{\text{sp}}\geq\vartheta$. The proof is based on an extension of \cite[Lemma 4]{lovasz-1979} and on \cite[Th. 5]{lovasz-1979}.

Let $\{U_x\}$ be an optimal representation of projectors for the graph $G$ with handle $F$. Let $\bar{G}$ be the complementary\footnote{Symbols $\myin,\mysin$ are connected in $\bar{G}$ if and only if they are not connected in $G$.} graph of $G$ and let $\{V_x\}$ be a representation of projectors for $\bar{G}$. For any pair of inputs $\myin\neq \mysin$ we have 
\begin{align*}
\Tr((U_\myin\otimes V_\myin)(U_\mysin\otimes V_\mysin)) & =\Tr(U_\myin U_\mysin)\Tr(V_\myin V_\mysin)
\\  & = 0,
\end{align*}
which means that the set
\begin{equation*}
\{U_\myin\otimes V_\myin\}_{\myin\in\myIn}
\end{equation*}
is a set of orthogonal projectors. Then, for any density operator $D$, we have
\begin{align*}
1 & = \Tr(F\otimes D)\\
& \geq \sum_\myin \Tr ((U_\myin\otimes V_\myin)(F\otimes D) ) \\
& = \sum_\myin \Tr(U_\myin F)\Tr(V_\myin D)\\
& \geq e^{-\vartheta_{\text{sp}}}\sum_\myin \Tr(V_\myin D).
\end{align*}

If we maximize the last term above over rank-one representations $V_\myin=\ket{v_\myin}\bra{v_\myin}$ of $\bar{G}$ and rank one operators  $D=\ket{d}\bra{d}$, we obtain
\begin{equation*}
\vartheta_{\text{sp}}\geq \max_{\{v_\myin\},d}\left[ \log \sum_\myin |\braket{v_\myin}{d}|^2\right],
\end{equation*}
where the maximum is over all rank-one representations $V_\myin=\ket{v_\myin}\bra{v_\myin}$ of $\bar{G}$, and unit norm vectors $d$. From \cite[Th. 5]{lovasz-1979}, however, the right hand side of the last equation is precisely  $\vartheta$ (with our logarithmic version of $\vartheta$).
\end{IEEEproof}

Theorem \ref{th:schrijver} conclusively shows that the sphere packing bound, when applied to classical-quantum channels, gives precisely Lov\'asz' bound to $C_0$. It also shows that  pure-state channels suffice for this purpose and that, for at least one optimal channel, the minimizing $F$ in \eqref{eq:minmaxQRinfty} can be taken to have rank one. That is, for some optimal pure-state channel, equality holds in \eqref{eq:RinftyvsV}.  It is worth pointing out that this is not true in general. For some pure-state channels, the optimal $F$ in \eqref{eq:minmaxQRinfty} has rank larger than one, and thus strict inequality holds in \eqref{eq:RinftyvsV}.

\section{Classical Channels and Pure-State Channels}
\label{sec:classical-pure}

It is useful to separately consider some properties of the sphere packing bound when computed for classical and for classical-quantum channels.
Classical channels can always be described as classical-quantum channels with commuting states $S_\myin$, and the sphere packing bound for these channels is precisely the same as the usual one \cite{shannon-gallager-berlekamp-1967-1}. So, there is no need to discuss this type of channels in general, and the aim of this section is instead to show that there is an interesting relation between a classical channel and a properly chosen pure-state channel. In order to make this relation clear, we first study the particular form of the sphere packing bound for pure-state channels. Note that, for these channels, the bound is known to be tight at high rates in light of the random coding bound \cite{burnashev-holevo-1998}. 

For a channel with pure states $S_\myin=\ket{{\psi}_\myin}\bra{{\psi}_\myin}$, we simply note that we have $S_\myin^{1/(1+\rho)}=S_\myin$ and, hence, the function $E_0(\rho,P)$ can be written in a simplified way. Let
\begin{equation*}
\bar{S}_P=\sum_\myin P(\myin) \ket{\psi_\myin}\bra{\psi_\myin		}
\end{equation*}
be the mixed state generated by the distribution $P$ over the input states. Then we can write
\begin{align*}
E_0(\rho, P) & =   -\log\Tr\left( \sum_\myin P(\myin) S_\myin \right)^{1+\rho}\\
& =    -\log\Tr\bar{S}_P^{1+\rho}\\
& = -\log \sum_i \lambda_i(\bar{S}_P)^{1+\rho},
\end{align*}
where the $\lambda_i(\bar{S}_P		)$'s are the eigenvalues of $\bar{S}_P$.
For these channels, the expressions for $R_\infty$ already introduced in 
the previous section reduce to
\begin{equation}
R_\infty=- \log \min_P\lambda_{\max}(\bar{S}_P)
\label{eq:rinftyminlambda}
\end{equation}
and
\begin{align}
R_\infty & = \min_{F}\max_{\myin} \log\frac{1}{\Tr\left(S_\myin F \right)}\\
& = \min_{F}\max_{\myin} \log\frac{1}{\bra{\psi_\myin}F\ket{\psi_\myin}}.
\label{eq:rinfty_ps_density}
\end{align}
We note that if the state vectors $\psi_\myin$ are constructed from the transition probabilities $\myCh{\myin}{\myout}$ of a classical channel as done in Section \ref{sec:Umbrella-Bhattacharyya} according to equation \eqref{def:psi_from_P}, then the expression in \eqref{eq:rinfty_ps_density} is very similar to what we called $\vartheta(1)$ in Section \ref{sec:Umbrella-classical}. The only difference is that we are now using a general density operator $F$ here, while we used a rank one operator $F=\ket{f}\bra{f}$ there. For a general set of vectors $\{\psi_\myin\}$, what happens is that if for the optimizing $P^*$ in \eqref{eq:rinftyminlambda} the eigenvalue $\lambda_{\max}(\bar{S}_{P^*})$ has multiplicity one, then the optimal density operator $F$ in  \eqref{eq:rinfty_ps_density} always has rank one. If the eigenvalue has multiplicity larger than one, instead, then the optimal $F$ can have rank larger than one.
We next show that for a pure-state channel with states $\{\psi_\myin\}$ constructed from a classical channel according to \eqref{def:psi_from_P}, there is always an optimal $F$ of rank one, and the value $R_\infty$ always equals the value $\vartheta(1)$ of the original classical channel. Since we already observed in Section \ref{sec:Umbrella-classical} that $\vartheta(1)$ is the cutoff rate of the classical channel, we deduce the interesting identity between the cutoff rate of a classical channel and the rate $R_\infty$ of a pure-state classical-quantum channel defined by means of \eqref{def:psi_from_P}. We actually prove a more general result, which will also be useful in proving a statement that we made in section \ref{sec:Umbrella-classical} about the so called non-negative channels. 

\begin{theorem}
\label{th:RinftyV}
If the vectors $\psi_\myin$ of a pure-state channel satisfy $\braket{\psi_\myin}{\psi_\mysin}\geq 0$, $\forall \myin,\mysin$, then we have 
\begin{equation*}
R_{\infty} = \max_{P} \left[-\log \sum_{\myin, \mysin} P(\myin)P(\mysin) \braket{\psi_\myin}{\psi_\mysin}\right].
\end{equation*}
Furthermore, $R_\infty$ equals the value $V(\{\psi_\myin\})$ in Lov\'asz' sense, that is, the optimal $F$ in equation \eqref{eq:rinfty_ps_density} can be chosen to have rank one.
\end{theorem}
\begin{IEEEproof}
We start with the expression
\begin{align*}
R_\infty & =- \log \min_P\lambda_{\max}(\bar{S}_P)\\
& = - \log \min_P\lambda_{\max}\left(\sum_\myin P(\myin) \ket{\psi_\myin}\bra{\psi_\myin}\right).
\end{align*}
Define a diagonal matrix $D_{P}$ with the distribution $P$ on its diagonal, and a matrix ${\Psi}$ with the vectors $\psi_\myin$ in its columns. Observe that 
\begin{align*}
\lambda_{\max}\left(\sum_\myin P(\myin) \ket{\psi_\myin}\bra{\psi_\myin}\right) & =  \lambda_{\max}({\Psi}D_{P}{\Psi}^\dagger)\\
& = \lambda_{\max}({\Psi}\sqrt{D_{P}}\sqrt{D_{P}}^\dagger{\Psi}^\dagger)\\
& = \lambda_{\max}(({\Psi}\sqrt{D_{P}})({\Psi}\sqrt{D_{P}})^\dagger)\\
& = \lambda_{\max}(({\Psi}\sqrt{D_{P}})^\dagger({\Psi}\sqrt{D_{P}})),
\end{align*}
so that
\begin{align*}
R_\infty  & =   -\log \min_P \lambda_{\max}(({\Psi}\sqrt{D_P})^\dagger({\Psi}\sqrt{D_P})) \\
 & =  -\log \min_P\max_{\|v\|=1}\bra{v}({\Psi}\sqrt{D_P})^\dagger({\Psi}\sqrt{D_P})\ket{v}\\
& =  -\log \min_P\max_{\|v\|=1}\bra{v}A_P\ket{v},
\end{align*}
where the maximum is over all unit norm vectors $v$ in $\R^{|\mathcal{X}|}$, and we set $A_P=({\Psi}\sqrt{D_P})^\dagger({\Psi}\sqrt{D_P})$. Since the matrix $A_P$ has non-negative entries $A_P(\myin,\mysin)=\sqrt{P(\myin)}\sqrt{P(\mysin)}\braket{\psi_\myin}{\psi_\mysin}\geq 0$, it is not difficult to see that the maximum can always be attained by a vector $v$ with non-negative components. We can then write $v=\sqrt{V}$, where $V$ is a distribution on $\mathcal{X}$. So,
\begin{align}
R_\infty  & = -\log \min_P\max_{V}\sum_{\myin, \mysin}\sqrt{V(\myin)}\sqrt{V(\mysin)}A_P(\myin,\mysin).
\label{eq:R0_bilinMinMax}
\end{align}
The sum in the last expression is a concave function of $P$ for fixed $V$ and vice versa. In fact, the generic term of the sum can be written as
\begin{equation*}
\sqrt{V(\myin)}\sqrt{V(\mysin)}\sqrt{P(\myin)}\sqrt{P(\mysin)}\braket{\psi_\myin}{\psi_\mysin}.
\end{equation*}
To prove concavity in $P$, for example, it suffices to note that $\braket{\psi_\myin}{\psi_\mysin}\geq 0$ and that, for two distributions $P_1$ and $P_2$, and $\alpha\in[0,1]$, we have
\begin{equation*}
\sqrt{\alpha P_1(\myin)+(1-\alpha)P_{2}(\myin)}\sqrt{\alpha P_1(\mysin)+(1-\alpha)P_{2}(\mysin)}\\
 \geq \alpha\sqrt{P_{1}(\myin)}\sqrt{P_{1}(\mysin)}+(1-\alpha)\sqrt{P_{2}(\myin)}\sqrt{P_{2}(\mysin)}
\end{equation*}
by the Cauchy-Schwartz inequality. Thus, it is not possible here to exchange the maximum and the minimum, as we did in the previous expressions for $R_\rho$, without further considerations. We proceed in a less conventional way by directly proving the optimality of the pair of distributions $P$ and $V$ both equal to
\begin{equation*}
P^*=\argmin_{P}\sum_{\myin, \mysin} P(\myin)P(\mysin) \braket{\psi_\myin}{\psi_\mysin}.
\end{equation*}
First note that the sum in the last expression is a convex function of $P$, since it is a quadratic form with nonnegative definite kernel matrix. This implies that $P^*$ can be determined by applying the usual Kuhn-Tucker conditions, which, after some calculations (see also \cite{jelinek-1968}), lead to
\begin{equation}
\sum_\mysin P^*(\mysin)\braket{\psi_\myin}{ \psi_\mysin}\geq \sum_{\myin,\mysin}P^*(\myin)P^*(\mysin)\braket{\psi_\myin}{ \psi_\mysin},\qquad \forall \myin
\label{eq:cond_optp*}
\end{equation}
with equality if $P^*(\myin)>0$.

Now, in order to prove that $V=P=P^*$ solves the min-max problem of eq. \eqref{eq:R0_bilinMinMax}, consider the conditions for optimality of $V$ given a fixed $P$. Since, as already observed, the function to maximize is concave, we can apply the usual Kuhn-Tucker conditions which lead, after some calculations, to the conclusion that $V$ is optimal for a given $P$ if and only if
\begin{equation}
\sum_{\mysin} \frac{1}{\sqrt{V(\myin)}}\sqrt{V(\mysin)}\sqrt{P(\myin)}\sqrt{P(\mysin)}\braket{\psi_\myin}{\psi_\mysin} \\ \leq \sum_{\myin_1,\myin_2}\sqrt{V(\myin_1)}\sqrt{V(\myin_2)}\sqrt{P(\myin_1)}\sqrt{P(\myin_2)}
\braket{\psi_{\myin_1}}{\psi_{\myin_2}}
\label{eq:cond_optq*}
\end{equation}
for all $\myin$, with equality if $V(\myin)>0$. Note that the inequality is surely strictly satisfied if $P(\myin)=0$, which implies that the condition is always satisfied for all $\myin$ such that $P(\myin)=0$, and also that, for such $\myin$, the optimal $V$ has $V(\myin)=0$. Also, $V(\myin)=0$ is optimal only if $P(\myin)=0$, for otherwise the associated condition is not met. Hence, for a given $P$,  the optimal $V$ satisfies $V(\myin)=0$ if and only if $P(\myin)=0$. Now we check if $V=P$ is optimal for the maximization by substituting $P$ for $V $ in the conditions of optimality. The condition for those $\myin$ with $V(\myin)>0$ (and thus $P(\myin)>0$) becomes 
\begin{align}
\sum_{\mysin} P(\mysin)\braket{\psi_\myin}{\psi_\mysin} & =  \sum_{\myin_1, \myin_2}P(\myin_1)P(\myin_2)\braket{\psi_{\myin_1}}{\psi_{\myin_2}},\label{eq:condpk>0} 
\end{align}
while the condition for those $\myin$ with $V(\myin)=0$ (and thus $P(\myin)=0$) is
\begin{align}
0 & \leq  \sum_{\myin_1, \myin_2}P(\myin_1)P(\myin_2)\braket{\psi_{\myin_1}}{\psi_{\myin_2}}.\label{eq:condpk=0} 
\end{align}
Comparing with eq. \eqref{eq:cond_optp*}, we thus see that $V=P$ is optimal if $P=P^*$, since all the conditions of eq. \eqref{eq:condpk>0} are satisfied with equality 
for all $\myin$ such that $P(\myin)>0$ (and thus $V(\myin)>0$), while the conditions in \eqref{eq:condpk=0} are always trivially satisfied.
Thus, $V=P$ is optimal for $P=P^*$. This automatically implies that $P=P^*$ is optimal, since we have
\begin{align*}
\max_{V}\sum_{\myin,\mysin}  \sqrt{V(\myin)}\sqrt{V(\mysin)}\sqrt{P(\myin)}\sqrt{P(\mysin)}\braket{\psi_\myin}{\psi_\mysin} 
& \geq  \left[\sum_{\myin,\mysin}  \sqrt{V(\myin)}\sqrt{V(\mysin)}\sqrt{P(\myin)}\sqrt{P(\mysin)}\braket{\psi_\myin}{\psi_\mysin}\right]_{V=P}\\
  & =  \sum_{\myin,\mysin}P(\myin)P(\mysin)\braket{\psi_\myin}{\psi_\mysin}\\
 & \geq  \sum_{\myin,\mysin}P^*(\myin) P^*(\mysin)\braket{\psi_\myin}{\psi_\mysin}.
\end{align*}
Thus, for the choice $P=P^*$ we have the optimal $V=P^*$ and thus 
\begin{align*}
R_\infty & =-\log \min_{P}\sum_{\myin,\mysin}P(\myin)P(\mysin) \braket{\psi_\myin}{\psi_\mysin}\\
&=\max_{P}\left[ -\log \sum_{\myin,\mysin}P(\myin)P(\mysin)\braket{\psi_\myin}{\psi_\mysin}\right],
\end{align*}
as was to be proven.

We only need now to prove that $R_\infty$ equals the value of the representation $\{\psi_\myin\}$ in Lov\'asz' sense, which means that the optimal $F$ in equation \eqref{eq:rinfty_ps_density} can be chosen to have rank one, or, in other words, that there exists a unit norm vector $f$ satisfying 
\begin{equation*}
|\braket{\psi_\myin}{f}|^2\geq e^{-R_\infty}.
\end{equation*}
In order to do this, consider the conditions expressed in equation \eqref{eq:cond_optp*}, which are satisfied for the optimal $P$ achieving $R_\infty$. Note that the right hand side of \eqref{eq:cond_optp*} is precisely the value $e^{-R_\infty}$ and it can be written as $\|\sum_\myin P^*(\myin)\psi_\myin\|^2$.
Hence, if we define 
\begin{equation*}
f=\frac{\sum_\myin P^*(\myin)\psi_\myin}{\|\sum_\myin P^*(\myin)\psi_\myin\|},
\end{equation*}
the conditions of equation \eqref{eq:cond_optp*} can be written as 
\begin{equation*}
\braket{\psi_\myin}{f}\geq \left\|\sum_\mysin P^*(\mysin)\psi_\mysin\right\|,\quad \forall \myin.
\end{equation*}
This implies that $f$ satisfies 
\begin{align*}
|\braket{\psi_\myin}{f}|^2 & \geq \left\|\sum_\mysin P^*(\mysin)\psi_\mysin\right\|^2\\
& =e^{-R_\infty}
\end{align*} as desired.
\end{IEEEproof}

\begin{remark}
The condition $\braket{\psi_\myin}{\psi_\mysin}\geq 0$, $\forall \myin,\mysin$, is by no means necessary. If a vector $\psi_\myin$ is substituted with $-\psi_\myin$, the density matrix $\bar{S}_P$ will not change, while some scalar products $\braket{\psi_\myin}{\psi_\mysin}$ will become negative. However, note that all the signs of the scalar products $\braket{\psi_{\myin_1}}{\psi_{\myin_2}}$ with $\myin_1=\myin$ or $\myin_2=\myin$ and $\myin_1\neq \myin_2$ will change.
\end{remark}

\begin{corollary}
The cutoff rate of a classical channel with transition probabilities $\myCh{\myin}{\myout}$ equals the rate $R_\infty$ of any classical-quantum pure-state channel with states $S_\myin=\ket{\psi_\myin}\bra{\psi_\myin}$ such that 
\begin{equation*}
\braket{\psi_\myin}{\psi_\mysin}=\sum_{\myout}\sqrt{\myCh{\myin}{\myout}\myCh{\mysin}{\myout}}.
\end{equation*}
In particular we have the following expression for the cutoff rate 
\begin{equation*}
- \log \min_P\lambda_{\max}\left(\sum_\myin P(\myin) \ket{\psi_\myin}\bra{\psi_\myin}\right).
\end{equation*}

\end{corollary}

\begin{remark}
Note that the original classical channel can be obtained from a pure-state channel defined by \eqref{def:psi_from_P} if a separable orthogonal measurement is used. Hence, any such pure-state channel can be interpreted as an underlying channel upon which the classical one is built. It is worth comparing this result with the known properties of the cutoff rate in the contexts of sequential decoding \cite{savage-1966}, \cite{arikan-1988} and list decoding \cite{forney-1968}. We have not yet studied this analogy, but we believe it deserves further consideration.
\end{remark}

We close this section with a bound on the possible values taken by $\Esp(R)$ for pure-state channels, that will also be useful in the next section. Since $\Esp(R)$ is finite and non-increasing for all $R>R_\infty$, it would be interesting to evaluate $\Esp(R_\infty^+)$, since this represents the largest finite value of the function $\Esp(R)$. Unfortunately, it is not easy in general to find this precise value. For the purpose of this paper, however, the following upper bound will be useful.
\begin{theorem}
\label{th:Esp(Rinfty)}
For any pure-state channel, we have 
\begin{equation*}
\Esp(R_\infty)\leq R_\infty.
\end{equation*}
\end{theorem}
\begin{IEEEproof}
Let $P^*$ be the optimal $P$ in \eqref{eq:rinftyminlambda}. Then we have
\begin{align*}
\Esp(R_\infty) & = \sup_{\rho\geq 0} {E}_0(\rho)-\rho {R}_\infty\\
& =  \sup_{\rho\geq 0} \left[ -\log \min_{P}\sum_i\lambda_i^{1+\rho}(\bar{{S}}_{P})+\rho \log \lambda_{\max}(\bar{{S}}_{P^*})\right].
\end{align*}
For each $\rho$ and each $P$, we have
\begin{align*}
\sum_i\lambda_i^{1+\rho}(\bar{{S}}_{P}) & \geq \lambda_{\max}^{1+\rho}(\bar{{S}}_{P})\\
& \geq \lambda_{\max}^{1+\rho}(\bar{{S}}_{P^*}),
\end{align*}
since $P^*$ minimizes $\lambda_{\max}\left(\bar{{S}}_P\right)$.
Hence, 
\begin{align*}
\Esp(R_\infty)& \leq \sup_{\rho\geq 0} -\log \lambda_{\max}^{1+\rho}(\bar{{S}}_{P^*})+\rho \log \lambda_{\max}(\bar{{S}}_{P^*})\\
& =  -\log \lambda_{\max}(\bar{{S}}_{P^*})\\
& = {R}_\infty.
\end{align*}
\end{IEEEproof}

\remark{The bound is tight, in the sense that $\Esp(R_\infty)=R_\infty$ is possible. For example, if the optimal $P=P^*$ in \eqref{eq:rinftyminlambda} is such that $\lambda_{\max}(\bar{{S}}_{P^*})$ has multiplicity one, and if the value $\Esp(R_\infty)$ is attained\footnote{This is not obvious in general. For general channels, the function $E_0(\rho)$ is not necessarily concave, and this implies that $\Esp(R_\infty)$ may in principle be obtained for a finite $\rho$. We conjecture, however, that $E_0(\rho)$ is concave for pure state channels.}  as $\rho\to\infty$, then by inspection of the proof we notice that $\Esp(R_\infty)=R_\infty$. Note also that this bound does not imply $E(R_\infty)\leq R_{\infty}$,
 due to the arbitrarily small but positive constant $\varepsilon$ which appears in the statement of Theorem \ref{th:sphere-packing} (see also footnote \ref{fn:E(C0)}).
}

\section{Sphere-Packed Umbrella Bound}
\label{sec:SP-umbrella}

In this section we consider again the umbrella bound of Section \ref{sec:Umbrella} and we extend it to a more general bound by means of the sphere packing bound. While the original idea was to bound the performance of a classical channel by means of auxiliary representations, that were in the end auxiliary pure-state classical-quantum channels, here we expand it to obtain an umbrella bound for a general classical-quantum channel by means of an auxiliary general classical-quantum channel. 

Given a classical-quantum channel $\mathcal{C}$ with density operators $S_\myin$, $\myin\in\myIn$, and given a fixed $\rho\geq 1$, consider an auxiliary classical-quantum channel  $\tilde{\mathcal{C}}$ with states $\tilde{S}_\myin$ such that 
\begin{equation*}
\Tr \sqrt{\tilde{S}_\myin}\sqrt{\tilde{S}_\mysin}\leq \left(\Tr|\sqrt{S_\myin}\sqrt{S_\mysin}|\right)^{2/\rho},
\end{equation*}
where $|A|=\sqrt{A^*A}$. We call such a channel an admissible auxiliary channel of degree $\rho$ and we call $\Gamma(\rho)$ the set of all such channels. For a fixed auxiliary channel  $\tilde{\mathcal{C}}$, let $\tilde{E}(R)$ be its reliability function and let $\tilde{E}_{\text{sp}}(R)$ be the associated sphere packing exponent.

To any sequence of $n$ input symbols $\bm{\myin}=(\myin_1,\myin_2,\ldots,\myin_n)$, we can 
associate a signal state $\bm{S}_{\bm{\myin}}=S_{\myin_1}\otimes S_{\myin_2}\cdots\otimes S_{\myin_n}$ for the original channel and a signal state $\tilde{\bm{S}}_{\bm{\myin}}=\tilde{S}_{\myin_1}\otimes \tilde{S}_{\myin_2}\cdots\otimes \tilde{S}_{\myin_n}$ for the auxiliary channel. Thus, to a set of $M$ codewords $\bm{\myin}_1,\ldots, \bm{\myin}_M$, with a simplified notation, we can associate  states $\bm{S}_1, \ldots, \bm{S}_M$ for the original channel and  states $\tilde{\bm{S}}_1, \ldots, \tilde{\bm{S}}_M$ for the auxiliary channel. We then bound the probability of error $\Pemax$ of the original channel $\mathcal{C}$ by bounding $\tilde{\mathsf{P}}_{\text{e},\text{max}}$ of the auxiliary channel $\tilde{\mathcal{C}}$.

Consider the auxiliary channel codeword states. It was proved by Holevo \cite{holevo-2000} that for such a set of states, there exists a measurement with probability of error for the $m$-th message bounded as
\begin{equation*}
\tilde{\mathsf{P}}_{\text{e}|m}\leq \sum_{m'\neq m}\Tr\sqrt{\tilde{\bm{S}}_m}\sqrt{\tilde{\bm{S}}_{m'}}.
\end{equation*}
This implies that 
\begin{align*}
\tilde{\mathsf{P}}_{\text{e},\text{max}}& \leq  (M-1)\max_{m\neq m'}\Tr\sqrt{\tilde{\bm{S}}_m}\sqrt{\tilde{\bm{S}}_{m'}}\\
& \leq  e^{nR}\max_{m\neq m'}\Tr\sqrt{\tilde{\bm{S}}_m}\sqrt{\tilde{\bm{S}}_{m'}}.
\end{align*}
Hence, asymptotically in the block length $n$, we find that 
\begin{align*}
\max_{m\neq m'}\Tr\sqrt{\tilde{\bm{S}}_m}\sqrt{\tilde{\bm{S}}_{m'}} & \geq  \tilde{\mathsf{P}}_{\text{e},\text{max}} e^{-nR}\\
&\geq  e^{-n (\tilde{E}_{\text{sp}}(R))+o(1))}e^{-nR}
\end{align*}
and hence
\begin{equation*}
\frac{1}{n}\min_{m\neq m'}\log \frac{1}{\Tr\sqrt{\tilde{\bm{S}}_m}\sqrt{\tilde{\bm{S}}_{m'}}}\leq  \tilde{E}_{\text{sp}}(R)+R+ o(1).
\end{equation*}
Considering the original states $\bm{S}_m$, we deduce that, 
\begin{equation}
\frac{1}{n}\min_{m\neq m'}\log \frac{1}{\Tr|\sqrt{\bm{S}_m}\sqrt{\bm{S}_{m'}}|}\leq \frac{\rho}{2} \left(\tilde{E}_{\text{sp}}(R)+R\right) + o(1).
\label{eq:fidelity_vs_Esp(R)}
\end{equation}

We now anticipate some notions that will be discussed in more detail in Section \ref{sec:Zero-rate} for expository convenience, since they play a central role also in other low rate bounds to $E(R)$. The left hand side of equation \eqref{eq:fidelity_vs_Esp(R)} is the minimum fidelity distance between codewords $\min_{m\neq m'}d_{\text{F}}(m,m')$ of equation \eqref{eq:codewords_d_F} below. Borrowing from Section \ref{sec:Zero-rate} the definition of the Chernoff distance between codewords \eqref{def:codewords_d_C}, the inequality $d_{\text{C}}(m,m')\leq 2d_{\text{F}}(m,m')$ from \eqref{eq:codewords_divergences}, and Theorem 
 \ref{th:quantum-codewords}, we  deduce
\begin{equation}
E(R)\leq\rho\left(\tilde{E}_{\text{sp}}(R)+R\right).
\label{eq:sph-pack-umbrella1}
\end{equation}
For particular types of channels, this last step can be tightened. For example, if the original channel is a pairwise reversible classical channel, then we can use the relation $d_{\text{C}}(m,m')= d_{\text{F}}(m,m')$, which holds in that case, and thus state
\begin{equation}
E(R)\leq\frac{\rho}{2}\left(\tilde{E}_{\text{sp}}(R)+R\right).
\label{eq:sph-pack-umbrella2}
\end{equation}
In general, however, it would be better to change the definition of $\Gamma(\rho)$ in order to take into account these types of particular cases. We have not yet investigated the topic in this direction and we thus only focus on the general case.

Clearly, for any rate $R$, the parameter $\rho$ and the auxiliary channel $\tilde{\mathcal{C}}$ can be chosen optimally. We thus have the following result.
\begin{theorem}
The reliability function $E(R)$ is upper bounded by the function $E_{\text{spu}}(R)$, where
\begin{equation}
E_{\text{spu}}(R) = \inf_{\rho\geq 1,  \tilde{\mathcal{C}}\in\Gamma(\rho) }\rho\left(\tilde{E}_{\text{sp}}(R)+R\right).
\label{bound:umbrella_2}
\end{equation}
\end{theorem}

A precise evaluation of this bound is not trivial. For a given $R$, one should find the optimal pair $(\rho, \tilde{C})$ and this is in general a complex task, which gives rise to interesting optimization problems. A complete treatment of this topic is still under investigation and will hopefully be detailed in a future work. It is important, however, to consider here the particular case of the bound used for classical channels by means of pure-state channels, in order to complete our interpretation of the umbrella bound given in Section \ref{sec:Umbrella} as a special case of this one. 

Suppose the states $S_\myin$ commute, which means that the original channel is a classical one, and assume that we restrict the set of possible admissible auxiliary channels to the pure-state ones. First observe that for commuting states $S_\myin$ we have $\Tr|\sqrt{S_\myin}\sqrt{S_\mysin}|=\Tr\sqrt{S_\myin}\sqrt{S_\mysin}$ while, for pure states $\tilde{S}_\myin=\ket{\tilde{\psi}_\myin}\bra{\tilde{\psi}_\myin}$, we have $\Tr \sqrt{\tilde{S}_\myin}\sqrt{\tilde{S}_\mysin}=|\braket{\tilde{\psi}_\myin}{\tilde{\psi}_\mysin}|^2$. Hence, for a classical channel, the restriction of $\Gamma(\rho)$ to pure state channels precisely corresponds to the set of admissible representations of degree $\rho$ in the sense of Section \ref{sec:Umbrella-umbrella}.
Then, for a fixed $\rho$ and for a fixed auxiliary channel, it is interesting to study the values of $R$ for which the bound is finite. This happens for rates $R$ larger than the rate $\tilde{R}_\infty$ where $\tilde{E}_{\text{sp}}(R)$ diverges and, from our previous analysis, we have
\begin{align*}
\tilde{R}_\infty &  = \min_{F}\max_{\myin} \log\frac{1}{\bra{\tilde{\psi}_\myin}F\ket{\tilde{\psi}_\myin}}\\
& =  - \log \min_P\lambda_{\max}\left(\bar{\tilde{S}}_P\right).
\end{align*}
We see that this is almost the same expression of the value of the representation $V(\{\tilde{\psi}_\myin\})$, and it is precisely the same value if $\braket{\tilde{\psi}_\myin}{\tilde{\psi}_\mysin}\geq 0$, $\forall \myin,\mysin$, as proved in Theorem \ref{th:RinftyV}.

It would be interesting to evaluate $\tilde{E}_{\text{sp}}(\tilde{R}_\infty)$ but, as mentioned in the previous section, this is not easy in general. By Theorem \ref{th:Esp(Rinfty)}, however, we know that $\tilde{E}(\tilde{R}_\infty)\leq \tilde{R}_\infty$.
This implies that for $R=\tilde{R}_\infty$, the right hand side of equation \eqref{eq:sph-pack-umbrella1} is not larger than $2\rho \tilde{R}_\infty$. Exploiting the fact that $E(R)$ is surely non-increasing, we then deduce the bound 
\begin{equation*}
E(R)\leq 2\rho \tilde{R}_\infty, \quad R>\tilde{R}_\infty,
\end{equation*}
which is to be compared with \eqref{eq:bound_theta2_1} with $\tilde{R}_\infty$ in place of $\vartheta(\rho)$. Considering the expression \eqref{eq:rinfty_ps_density} for $R_\infty$, and optimizing over the pure-state auxiliary channels, we thus see that the umbrella bound  derived in Section \ref{sec:Umbrella} is included as a particular case of the more general one derived in this section. In general, however, for a given representation, $\tilde{E}_{\text{sp}}(\tilde{R}_\infty)$ can be strictly smaller than $\tilde{R}_\infty$. Furthermore, the function $\Esp(R)$ has in many cases slope $-\infty$ in $R_\infty$ and thus the bound of equation \eqref{bound:umbrella_2}, for a given $R$, is in those cases not optimized for the value of $\rho$ which leads to $\tilde{R}_\infty=R$, but for a larger one, which leads to $\tilde{R}_\infty<R$, whenever possible. This implies that the bound derived here is in general strictly better than that derived in Section \ref{sec:Umbrella}.

We close this section by pointing out that, as a consequence of Theorem \ref{th:schrijver}, the smallest value of $R$ for which $E_{\text{spu}}(R)$ is finite is precisely the Lov\'asz theta function.

\section{Extension of  Low Rate Bounds to Classical-Quantum Channels}
\label{sec:Zero-rate}

As anticipated in the previous sections, the sphere packing bound is in general not tight at low rates.
For example, it is infinite over a non-empty range of positive small rates for all non trivial pure-state channels, even if there is no pair of orthogonal input states, which implies that the zero-error capacity of the channel is zero.
In this section, we deal precisely with channels with no zero-error capacity. For these channels, in the classical case some bounds that greatly improve the sphere packing bound were derived. The main objective of this section is to consider some possible extensions of these results to the classical-quantum case. 

A first interesting result that extends a low rate bound from the classical to the classical-quantum setting was already obtained in \cite{holevo-1998} for the case of pure-state channels. There, the authors proved the equivalent of the zero-rate upper bound to $E(R)$ derived in \cite{shannon-gallager-berlekamp-1967-2} for the case of pure-state channels with no zero-error capacity, thus proving that even in this case the expurgated bound is tight at zero rate. For general classical-quantum channels, a similar result was attempted in \cite{holevo-2000}, but the obtained upper bound to $E(R)$ at zero-rate does not coincide with the limiting value of the expurgated bound in this case.

In this section, we first present the extension of the zero-rate upper bound of \cite{berlekamp-thesis}, \cite{shannon-gallager-berlekamp-1967-2} to the case of mixed-state channels, which leads to the determination of the exact value of the reliability function at zero-rate. Then, we also discuss some other bounds. In particular, we consider the application of Blahut's bound \cite{blahut-1977} to the case of pure- and mixed-state channels.

A recurring theme in the study of the reliability of classical and classical-quantum channels is the fact that, at low rates, the probability of error is dominated by the worst pair of codewords in the code. At high rates, it is important to bound the probability of a message to be incorrectly decoded due to a bulk of competitors. The auxiliary state $\bm{F}$ used in the proof of the sphere packing bound serves precisely to this scope and represents this bulk of competitors. At low rates, instead,  there are essentially only few competitors (we may conjecture just one) which are responsible for almost all the probability of error. Thus, in the low rate region, it is important to bound the probability of error in a binary decision between any pair of codewords.  For this reason, we need to specialize Theorem 
\ref{th:QSGBB} to the case of binary hypothesis testing between two codewords, so as to obtain the quantum generalization of \cite[Th. 1]{shannon-gallager-berlekamp-1967-2}.

In the context of Theorem 
\ref{th:QSGBB}, thus, let $\bm{A}=\bm{S}_{\bm{x}_m}$ and $\bm{B}=\bm{S}_{\bm{x}_{m'}}$, call for simplicity $\mu(s)=\mu_{\bm{S}_{\bm{x}_m},\bm{S}_{\bm{x}_{m'}}}(s)$ and let $s^*$ minimize $\mu(s)$. Then we have $\mu'(s^*)=0$, and thus either 
\begin{equation*}
\mathsf{P}_{\text{e}| m}>\frac{1}{8}\exp\left[\mu(s^*)-s^*\sqrt{2\mu''(s^*)}\right]
\end{equation*}
or
\begin{equation*}
\mathsf{P}_{\text{e}|m'}>\frac{1}{8}\exp\left[\mu(s^*)-(1-s^*)\sqrt{2\mu''(s^*)}\right].
\end{equation*}
The key point is now to show that the second derivative term is unimportant for large $n$, so that the exponential behaviour is determined by $\mu(s^*)$.
If $P_{m,m'}(\myin,\mysin)$ is the joint composition between codewords $\bm{x}_m$ and $\bm{x}_{m'}$,we find 
\begin{equation*}
\mu(s)=n\sum_{\myin, \mysin} P_{m,m'}(\myin,\mysin) \mu_{\myin,\mysin}(s),
\end{equation*}
where, for ease of notation,
\begin{equation*}
\mu_{\myin,\mysin}(s)=\log \Tr S_\myin^{1-s} S_\mysin^s.
\end{equation*}
By definition of $s^*$ we can thus introduce the Chernoff distance between the two messages $m$ and $m'$ (this corresponds to the \emph{discrepancy} $D(m,m')$ in \cite{shannon-gallager-berlekamp-1967-2})
\begin{align}
d_{\text{C}}(m,m') & =  \frac{1}{n}d_{\text{C}}(\bm{S}_{\bm{x}_m},\bm{S}_{\bm{x}_{m'}})\\
& =  -\min_{0\leq s \leq 1}\sum_{\myin,\mysin}P_{m,m'}(\myin,\mysin) \mu_{\myin,\mysin}(s).
\label{def:codewords_d_C}
\end{align} 
We then have the follwing generalization of \cite[Th. 1]{shannon-gallager-berlekamp-1967-2}.

\begin{theorem}
\label{th:quantum-codewords}
If $\bm{\myin}_m$ and $\bm{\myin}_{m'}$ are two codewords in a code of block length $n$ for a classical-quantum channel with symbol states $S_\myin$, $\myin\in\myIn$, then 
\begin{equation}
\mathsf{P}_{\text{e}| m}>\frac{1}{8}\exp\left[ -n \left( d_{\text{C}}(m,m')+\sqrt{\frac{2}{n}} \log \lambda_{\text{min}}^{-1}\right)\right]
\label{eq:pem}
\end{equation}
or
\begin{equation}
\mathsf{P}_{\text{e}|	m'}>\frac{1}{8}\exp \left[- n\left(d_{\text{C}}(m,m')+\sqrt{\frac{2}{n}}\log\lambda_{\text{min}}^{-1}\right)\right],
\label{eq:pem'}
\end{equation}
where $\lambda_{\text{min}}$ is the smallest non-zero eigenvalue of all the states $S_\myin$, $\myin\in\myIn$.
\end{theorem}

\begin{IEEEproof}
The proof is essentially the same as in \cite[Th. 1]{shannon-gallager-berlekamp-1967-2}, with the only difference that we have to bound $\mu_{\myin,\mysin}''(s)$ for $\mu_{\myin,\mysin}(s)$ computed between density operators rather than probability distributions. Using the spectral decomposition of the density operators
\begin{align*}
S_\myin & =\sum_i\lambda_{\myin,i} \ket{\psi_{\myin,i}}\bra{\psi_{\myin,i}} \\S_\mysin & =\sum_j \lambda_{\mysin,j} \ket{\psi_{\mysin,j}}\bra{\psi_{\mysin,j}},
\end{align*}
where  $\{\psi_{\myin,i}\}$ and $\{\psi_{\mysin,j}\}$ are orthonormal bases, we can however use again the Nussbaum-Szko{\l}a mapping to define two probability distributions $Q_\myin(i,j)$ and $Q_\mysin(i,j)$ as
\begin{align}
Q_\myin(i,j)& =\lambda_{\myin,i} |\braket{\psi_{\myin,i}}{\psi_{\mysin,j}}|^2, \\ Q_\mysin(i,j)& =\lambda_{\mysin,j} |\braket{\psi_{\myin,i}}{\psi_{\mysin,j}}|^2,
\label{eq:Pdef}
\end{align}
so that
\begin{equation*}
\mu_{\myin,\mysin}(s)  =  \log  \sum_{i,j}Q_\myin(i,j)^{1-s}Q_\mysin(i,j)^s.
\end{equation*}
The proof in \cite[Th. 1]{shannon-gallager-berlekamp-1967-2} can then be applied using our distributions $Q_\myin(\cdot,\cdot)$ and $Q_\mysin(\cdot,\cdot)$ for $P(\cdot|i)$ and $P(\cdot|k)$ there, and noticing that in \cite[eq. (1.10)]{shannon-gallager-berlekamp-1967-2} we can use in our case the bound
\begin{align*}
\left|\log\frac{Q_\myin(i,j)}{Q_\mysin(i,j)}\right| & =   \left|\log\frac{\lambda_{\myin,i}}{\lambda_{\mysin,j}}\right|\\
& \leq  \log \lambda_{\text{min}}^{-1}.
\end{align*}
\end{IEEEproof}

By considering all possible pairs of codewords, it is clear that the optimal exponential behaviour of $\Pemax$ can be bounded in terms of the minimum discrepancy between codewords of an optimal code with $M$ codewords and block length $n$. 
Theorem \ref{th:quantum-codewords} implies that the results on the zero-rate reliability of \cite[Th. 3-4]{shannon-gallager-berlekamp-1967-2} apply straightforwardly to the classical-quantum case. These results are related to an upper bound to the reliability function in the low rate region derived by Blahut \cite{blahut-1977} and to upper bounds derived for the classical-quantum case in \cite{burnashev-holevo-1998} and \cite{holevo-2000}. A clarification of these relations is the objective of the next part of this section. 

In view of that, it is useful to remind here the relation between the Chernoff, the Bhattacharyya and the fidelity distances. If we define the Bhattacharyya and the fidelity distances between two messages $m$ and $m'$ as
\begin{align}
d_{\text{B}}(m,m') & =  \sum_{\myin,\mysin}P_{m,m'}(\myin,\mysin) d_{\text{B}}(S_\myin,S_\mysin)\\
d_{\text{F}}(m,m') & =  \sum_{\myin,\mysin}P_{m,m'}(\myin,\mysin) d_{\text{F}}(S_\myin,S_\mysin),
\label{eq:codewords_d_F}
\end{align}
then we see that the following inequalities hold
\begin{equation}
d_{\text{F}}(m,m')\leq d_{\text{B}}(m,m')\leq d_{\text{C}}(m,m')\\ \leq 2d_{\text{F}}(m,m')\leq 2d_{\text{B}}(m,m').
\label{eq:codewords_divergences}
\end{equation}

It is useful to investigate conditions under which $d_{\text{C}}(m,m')$ can be expressed exactly in terms of $d_{\text{B}}(m,m')$ or $d_{\text{F}}(m,m')$. One case is that of pairwise reversible channels. We observe that in the classical-quantum context, the condition for pairwise reversibility holds for example for all pure-state channels, for which $\mu_{\myin,\mysin}(s)$ is constant for all $\myin$ and $\mysin$. Another important case is the case where codewords $m$ and $m'$ have a symmetric joint composition, that is $P_{m,m'}(\myin,\mysin)=P_{m,m'}(\mysin,\myin)$ for all $\myin,\mysin$, for in that case the function $\mu(s)$ is symmetric around $s=1/2$, due to the fact that $\mu_{\myin,\mysin}(s)=\mu_{\mysin,\myin}(1-s)$.
In those cases, the Chernoff distance $d_{\text{C}}(m,m')$ can be replaced by the closed form expression of $d_{\text{B}}(m,m')$. In the general case, however, for a single pair of codewords, it is not possible to use $d_{\text{B}}(m,m')$ in place of $d_{\text{C}}(m,m')$ for lower bounding the probability of error, and it can be proved that in some cases $d_{\text{C}}(m,m')=2 d_{\text{B}}(m,m')$.

We are now in a position to consider the low rate upper bounds to the reliability function derived in \cite{shannon-gallager-berlekamp-1967-2} and in \cite{blahut-1977} discussing their applicability for classical and classical-quantum channels.
For classical channels, a low rate improvement of the sphere packing bound for channels with no zero error capacity was obtained in \cite{shannon-gallager-berlekamp-1967-2}. This bound is based on two important results:
\begin{enumerate}
\item A \emph{zero rate bound} \cite[Th 4]{shannon-gallager-berlekamp-1967-2}, first derived in \cite[Ch. 2]{berlekamp-thesis}, which asserts that, for a discrete memoryless channel with transition probabilities $\myCh{\myin}{\myout}$,
\begin{equation}
E(0^+)\leq \max_{P}-\sum_{\myin,\mysin} P(\myin)P(\mysin)\log \sum_\myout \sqrt{\myCh{\myin}{\myout}\myCh{\mysin}{\myout}}.
\label{eq:zero-rate}
\end{equation}
The right hand side of the above equation is also the value of the expurgated bound of Gallager as $R\to 0$. This implies that the bound is tight, and that the expression determines the reliability function at zero rate.

\item A \emph{straight line bound} \cite[Th. 6]{shannon-gallager-berlekamp-1967-2}, attributed to Shannon and Gallager in Berlekamp's thesis \cite[pag. 6]{berlekamp-thesis} which asserts that, given an upper bound $E_{\text{lr}}(R)$ to the reliability function which is tighter than $\Esp(R)$ at low rates, it is possible to combine $E_{\text{lr}}(R)$ ad $\Esp(R)$ to obtain an improved upper bound on $E(R)$ by drawing a straight line from any two points on the curves $E_{\text{lr}}(R)$ and $\Esp(R)$.
\end{enumerate}
Combining these two results, the authors obtain an upper bound to $E(R)$ which is strictly better than the sphere packing bound for low rates and is tight at rate $R=0$. 

Of the above two results, we only consider here the extension of the zero-rate bound, a possible generalization of the straight line bound being still under investigation at the moment. 
For quantum channels, a zero rate bound has been obtained by Burnashev and Holevo \cite{burnashev-holevo-1998} for the case of pure-state channels which essentially parallel the classical result. That is, they proved that in the case of pure-states $S_\myin=\ket{\psi_\myin}\bra{\psi_\myin}$, if there is no pair of orthogonal states, then 
\begin{equation}
E(0^+)\leq \max_{P}\left[-\sum_{\myin,\mysin}	 P(\myin)P(\mysin)\log |\braket{\psi_\myin}{\psi_\mysin}|^2\right].
\label{eq:Qzerorate-pure}
\end{equation}
As for the classical case, this bound coincides with a lower bound given by the expurgated bound as $R\to 0 $, thus providing the exact expression.

For general mixed-state channels, the reliability at zero rate was considered by Holevo \cite{holevo-2000} who obtained the bound
\begin{equation}
 \max_{P}\left[-\sum_{\myin,\mysin}P(\myin)P(\mysin)\log \Tr \sqrt{S_\myin}\sqrt{S_\mysin} \right]
 \leq E(0^+) \\  \leq 2\max_{P}\left[-\sum_{\myin,\mysin}P(\myin)P(\mysin) \log\Tr  |\sqrt{S_\myin}\sqrt{S_\mysin}|\right].
 \label{eq:HolevoUpLowBound}
\end{equation}
Note that, in light of the now clear parallel role of the R{\'e}nyi divergence in the classical and quantum case, the expression of the lower bound in \eqref{eq:HolevoUpLowBound} is the generalization of the right hand side of \eqref{eq:zero-rate}, while the upper bound is in the general case a larger quantity. Thus, we are inclined to believe that the correct expression for $E(0^+)$ should be the first one. 
This is actually the case, as stated in the following theorem.
\begin{theorem}
For a general classical-quantum channel with states $S_\myin$, $\myin\in\myIn$, no two of which are orthogonal, the reliability function at zero rate is given by the expression
\begin{equation}
E(0^+)= \max_{P}\left[-\sum_{\myin,\mysin}P(\myin)P(\mysin) \log \Tr \sqrt{S_\myin}\sqrt{S_\mysin}\right].
\label{eq:berlek-general}
\end{equation}
\end{theorem}

\begin{IEEEproof}
This theorem is the quantum equivalent of \cite[Th. 4]{shannon-gallager-berlekamp-1967-2}, and it is a direct consequence of Theorem \ref{th:quantum-codewords}. It can be noticed, in fact, that the proof of \cite[Th. 4]{shannon-gallager-berlekamp-1967-2} holds exactly unchanged in this new setting since it only depends on \cite[Th. 1]{shannon-gallager-berlekamp-1967-2} and on the definition and additivity property of the function $\mu(s)$. We do not go through the proof here since it is very long and it does not need any change.

\end{IEEEproof}

It is interesting to briefly discuss the Holevo upper bound for $E(0^+)$ for mixed-state channels, that is, the right hand side of \eqref{eq:HolevoUpLowBound}. First observe that, in the classical case, that is when all states commute, the expression reduces to 
\begin{equation}
E(0^+)\leq 2 \max_{P}\left[-\sum_{\myin,\mysin}P(\myin)P(\mysin) \log \sum_\myout \sqrt{\myCh{\myin}{\myout}\myCh{\mysin}{\myout}}\right].
\label{eq:2timesrate-zero}
\end{equation}
This bound (in the classical case) is much easier to prove than Berlekamp's bound \eqref{eq:zero-rate}. First note that Berlekamp's bound is relatively simple to prove for the case of pairwise reversible channels, exploiting the fact that $d_{\text{C}}(m,m')=d_{\text{B}}(m,m')$ in that case  (see \cite[Cor. 3.1]{shannon-gallager-berlekamp-1967-2}). Essentially the same proof allows to derive the bound \eqref{eq:Qzerorate-pure}, since $d_{\text{C}}(m,m')=d_{\text{B}}(m,m')$ also holds for pure-state channels as discussed in the previous section. For general classical channels, the same proof can be used by bounding $d_{\text{C}}(m,m')$ with $2d_{\text{B}}(m,m')$ as explained before, and this leads to the bound \eqref{eq:2timesrate-zero}. For general classical-quantum channels, finally, Holevo's bound on the right hand side of \eqref{eq:HolevoUpLowBound} is obtained with the same procedure by using the bound $d_{\text{C}}(m,m')\leq 2d_{\text{F}}(m,m')$. 

The proof of Berlekamp's bound \eqref{eq:zero-rate}  and \eqref{eq:berlek-general} for general classical and classical-quantum channels, respectively, is instead more complicated (see the proof of \cite[Th. 4]{shannon-gallager-berlekamp-1967-2}) and it relies heavily on the fact that the number of codewords $M$ can be made as large as desired. Roughly speaking, it is a combinatorial result on possible joint compositions of pairs of codewords extracted from arbitrarily large sets. The interested reader can check that the truly remarkable result in the proof of \cite[Th. 4]{shannon-gallager-berlekamp-1967-2} is a characterization of the joint compositions between codewords, which implies that the asymmetries of the functions $\mu_{\myin,\mysin}(s)$ can be somehow ``averaged'' due to the many possible pairs of codewords that can be compared (see \cite[eqs. (1.40)-(1.45), (1.54)]{shannon-gallager-berlekamp-1967-2} and observe that only the additivity of $\mu(s)$ and the fact that $\mu_{\myin,\mysin}'(1/2)=-\mu_{\mysin,\myin}'(1/2)$ are used).

We can now consider the low rate upper bound to the reliability function derived in \cite{blahut-1977}. Blahut considers the class of non-negative definite channels mentioned in Section \ref{sec:Umbrella-classical}.
For these channels, Blahut  \cite[Th. 8]{blahut-1977} shows that it is possible to relate the smallest  Bhattacharyya distance between codewords of a constant composition code for a given positive rate $R$ to a function $E_{\text{U}}(R)$ which he defines as
\begin{equation*}
E_{\text{U}}(R)=\max_{P} \min_{\hat{P}\in \mathcal{P}_R(P)}\\
\left[ -\sum_{\myin,\myin_1,\myin_2} P(\myin)\hat{P}_{\myin}(\myin_1) \hat{P}_{\myin}(\myin_2) \log 
\sum_\myout \sqrt{\myCh{\myin_1}{\myout}\myCh{\myin_2}{\myout}} 
\right],
\end{equation*}
where
\begin{equation*}
\mathcal{P}_R(P)=\left\{  \hat{P} | I(P;\hat{P})\leq R, \sum_\myin P(\myin) \hat{P}_{\myin}(\mysin)=P(\mysin) \right\}
\end{equation*}
and $I(P;\hat{P})$ is the mutual information between a variable with marginal $P$ and another variable with conditional distribution $\hat{P}_\myin$ given that the first variable is $\myin$. Blahut then derives an upper bound \cite[Sec. VI]{blahut-1977} on $E(R)$ by bounding the probability of error between codewords using the Bhattacharyya distance. 
For pairwise exchangeable channels, the bound states that $E(R)\leq E_{\text{U}}(R)$. For general non pairwise exchangeable channels, in this author's opinion\footnote{The application of the theorem to general channels is, in this author's opinion, not correct. This point would bring us too far, and it will thus
 be discussed in a separate note.}, the bound would need to be modified in the form $E(R)\leq 2 E_{\text{U}}(R)$.

Bringing this into the classical-quantum setting, redefining 
$E_{\text{U}}(R)$ using the expression $\Tr\sqrt{S_\myin}\sqrt{S_\mysin}$ in place of $\sum_\myout\sqrt{\myCh{\myin}{\myout}\myCh{\mysin}{\myout}}, $
 Blahut's bound $E(R)\leq E_{\text{U}}(R)$ can be applied for example to all non-negative definite pure-state channels and it can be applied to all non-negative channels with the correcting coefficient 2 in the form $E(R)\leq 2E_{\text{U}}(R)$. If the coefficient $2$ makes the bound much looser than the straight line bound in the classical case, in the quantum case it results in the best known upper bound to $E(R)$, since no straight line bound has been obtained yet.

Another important observation is that, for the particular case of binary symmetric channels, the Chernoff distance between messages $d_{\text{C}}(m,m')$ is proportional to the Hamming distance between the codewords. Thus, for a given $R$, any upper bound on the minimum Hamming distance between codewords is also an upper bound on the minimum Chernoff distance between messages, which easily implies a bound on $E(R)$. In particular, the reliability function $E(R)$ can be bounded using the JPL bound \cite{mceliece-et-al-1977}. The resulting bounds are tighter than $E_{\text{U}}(R)$ in the classical case as well as in the classical-quantum case.

Finally, it is worth pointing out that, as opposed to the zero-rate bound, deriving a quantum version of the straight line bound seems to be a more complicated task, if even possible. The straight line bound in the classical case, indeed, is proved by exploiting the fact that a decoding decision can always be implemented in two steps by splitting the output sequence in two blocks, applying a list decoding on the first block and a low rate decoding on the second one. In the case of quantum channels, this procedure does not apply directly since the optimal measurements are in general entangled and are not equivalent to separable measurements.

\section{conclusions and future work}
In this paper, we have considered the problem of lower bounding the probability of error in coding for discrete memoryless classical and classical-quantum channels. A sphere packing bound has been derived for the latter, and it was shown that this bound provides the natural framework for including Lov\'asz' work into the picture of probabilistic bounds to the reliability function.  An umbrella bound has been derived as a first example of use of the sphere packing bound applied to auxiliary channels for bounding the reliability of channels with a zero-error capacity. Additional side results have been obtained showing that interesting connections exist between classical channels and pure-state channels. We believe that there is much room for improvements, over the bounds derived in this paper, by means of known techniques already used with success in related works. There are (at least) three important questions that should be addressed by next works in this direction. The first is the possibility of finding a smooth connection between the sphere packing bound and $\vartheta$, as indicated in Fig. \ref{fig:umbrella}. The second is the possibility of including Haemers' bound to $C_0$ into the same picture, in order to obtain a bound to $C_0$ that is more general than both Lov\'asz' and Haemers' ones. The third important question to address is whether it is possible to extend the straight line bound to classical-quantum channels. This would give very good bounds to $E(R)$ in the low rate region for all channels without a zero-error capacity.

\section{Acknowledgements}
I would like to thank Alexander Schrijver for allowing me to include Theorem \ref{th:schrijver} and Alexander Holevo and the anonymous reviewers for useful comments which have greatly improved the quality of the paper.
I am indebted with Algo Car\`e for useful comments on Section \ref{sec:classical-pure}, and with Alexander Barg for his comments on Section \ref{sec:Umbrella}.
I thank Robert Gallager, Elwyin Berlekamp, Richard Blahut,  Andreas Winter and Imre Csisz\'ar for helping me to understand the historical context of the topic. 
Finally, I give special thanks to Pierangelo Migliorati for stimulating my interest on the cutoff rate and on the reliability function. I thank him, Riccardo Leonardi and all my colleagues at the University of Brescia for their support.
\appendices
\section{Continuity of $R_n(s,P,F_s)$ }
\label{appendix:cont_R}

For any $s\in(0,1)$, and probability distribution $P$ on $\myIn$, let
\begin{equation}
A(s,P) = \sum_\myin P(\myin) S_\myin^{1-s},
\label{def:alphasq}
\end{equation}
and call $\mathcal{A}(s)$ the convex set of all such operators when $P$ varies over the simplex of probability distributions. Recall that
\begin{align}
E_0\left(\frac{s}{1-s},P\right) & =   -\log \Tr\left(\sum_\myin P(\myin) S_\myin ^{1-s} \right)^{1/(1-s)}\label{Eminq}\\
& =  -\log \Tr A(s,P)^{1/(1-s)}\label{Eminalpha}.
\end{align}
Let then $P_s$ be a choice of $P$ that maximizes $E_0(s/(1-s),P)$, so that
\begin{equation*}
E_0\left( \frac{s} {1-s},P_{s}\right)=-\log \min_{A\in \mathcal{A}(s)} \Tr A^{1/(1-s)}.
\end{equation*}
Such a maximizing $P_s$ must exist, since $E_0(s/(1-s),P)$ is continuous on the compact set of probability distributions, but need not be unique.

Note now that the state $F_s$ as defined in \eqref{def:fs} is given by
\begin{equation*}
F_s=\frac{A(s,P_s)^{1/(1-s)}}{\Tr A(s,P_s)^{1/(1-s)}}.
\end{equation*}
Our aim is to prove that the function
\begin{equation*}
R_n(s,P,F_s)=- \sum_\myin P(\myin) \left[ \mu_{S_\myin,F_s}(s)+ (1-s) \mu_{S_\myin,F_s}'(s)\right]\\
+\frac{1}{\sqrt{n}}(1-s)\sqrt{2\sum_\myin P(\myin) \mu_{S_\myin,F_s}''(s)} + \frac{1}{n} \log 8.
\end{equation*}
is continuous in $s$. We do this by proving that the density operator $A(s,P_s)$ is a continuous function of $s$, from which the continuity of $F_s$ follows, implying the continuity of $\mu_{S_\myin,F_s}(s)$ and its first two derivatives, by means of the Nussbaum-Szko{\l}a mapping and the relations \eqref{eq:mu'} and \eqref{eq:mu''} when applied to states $S_\myin$ and $F_s$.

First observe that, for fixed $P$, both $A(s,P)$ and $E_0(s/(1-s),P)$ are continuous functions of $s$. 
Assume that $A(s,P_s)$ is not continuous at the point $s=\bar{s}$, where $0<\bar{s}<1$. By definition, there exists a $\delta_1>0$ and a sequence $s_n$ such that $s_n\to \bar{s}$ but $\|A(\bar{s},P_{\bar{s}})-A(s_n,P_{s_n})\|_1>\delta_1$.  By picking an appropriate subsequence if necessary, assume without loss of generality that $P_{s_n}\to \tilde{P}$ so that $A(s_n,P_{s_n})=A(s_n,\tilde{P})+\epsilon_1(n)=A(\bar{s},\tilde{P})+\epsilon_2(n)$ for some vanishing operators $\epsilon_1(n)$ and $\epsilon_2(n)$.
Note that $A(\bar{s},\tilde{P})\in \mathcal{A}(\bar{s})$ and, by construction, $\|A(\bar{s},P_{\bar{s}})-A(\bar{s},\tilde{P})\|_1\geq \delta_1$. Since $A(\bar{s},P_{\bar{s}})$ is the choice of $A$ that minimizes the function $\Tr A^{1/(1-\bar{s})}$ over the convex domain $\mathcal{A}(\bar{s})$, since this function is strictly convex in that domain, and since $A(\bar{s},\tilde{P})\in \mathcal{A}(\bar{s})$ is bounded away from $A(\bar{s},P_{\bar{s}})$ by a constant $\delta_1$, there exists a fixed positive $\delta_2>0$ such that
\begin{equation*}
\Tr A(\bar{s},\tilde{P})^{1/(1-\bar{s})}\geq \Tr A(\bar{s},P_{\bar{s}})^{1/(1-\bar{s})}+\delta_2.
\end{equation*}
But then, 
\begin{align*}
\Tr A(s_n,P_{s_n})^{1/(1-s_n)} & =  \Tr\left( A(\bar{s},\tilde{P})+\epsilon_2(n)\right)^{1/(1-s_n)}\\
& \geq  \Tr A(\bar{s},\tilde{P})^{1/(1-s_n)} - \varepsilon_1(n)\\
& \geq  \Tr A(\bar{s},\tilde{P})^{1/(1-\bar{s})} - \varepsilon_2(n)\\
& \geq   \Tr A(\bar{s},P_{\bar{s}})^{1/(1-\bar{s})}+\delta_2- \varepsilon_{2}(n)\\
& \geq   \Tr A(s_n,P_{\bar{s}})^{1/(1-s_n)}+\delta_2- \varepsilon_{3}(n),
\end{align*}
where $\varepsilon_1(n),\varepsilon_2(n),\varepsilon_{3}(n)$ are vanishing positive functions of $n$. This implies that $P_{s_n}$ is not optimal for $n$ large enough, contrarily to the assumed hypothesis.

\section{Historical note}
\label{sec:App_history}

\begin{quotation}
\emph{``Romanticism aside, however, the history
of science - like Orwell's Big Brother
state - usually writes and rewrites
history to remove inconvenient facts,
mistakes and idiosyncrasies, leaving only
a rationalized path to our present knowledge,
or what historians sometimes call
``whig'' history. In so doing, it not only
distorts the actual course of historical
events but also gives a misleadingly simplistic
picture of the richness of scientific activity ... ''}
\flushright -- Jeff Hughes
\end{quotation}

Since the results of this paper are essentially based on the original proof of the sphere packing bound as given in \cite{shannon-gallager-berlekamp-1967-1}, we believe that some historical comments on that proof may be of particular interest to the reader. It is important to point out, in fact, that even if \cite{shannon-gallager-berlekamp-1967-1} contains the first formal proof of this result, the bound itself had already been accepted before, at least among information theorists at the MIT.

The main idea behind the sphere packing bound was Shannon's \cite{gallager-personal-2012}. Elias proved the bound for the binary symmetric channel in 1955 \cite{elias-1955}, and the bound for general DMC was first stated by Fano \cite[Ch. 9]{fano-book} as an attempt to generalize Elias' ideas to the non-binary case. Fano's proof, however, was not competely rigorous, although it was correct with respect to the elaboration of the many complicated and ``subtle'' equations that allowed him to obtain the resulting expression for the first time. Fano's approach already contained the main idea of considering a binary hypothesis test between some appropriately chosen codewords and a dummy output distribution, and his procedure allowed him to solve the resulting minmax optimization problem with a direct approach which, in this author's opinion, could be defined ``tedious'' but not ``unenlightening'' \cite[pag. 91]{shannon-gallager-berlekamp-1967-1}, and which opened in any case the way that allowed to subsequently obtain the formal proof later on. 

It is not easy to precisely understand, from the published papers, when the formal proof was subsequently obtained and by whom. In fact, even if it was first published in the mentioned 1967 paper \cite{shannon-gallager-berlekamp-1967-1}, the result must have been somehow accepted before, at least at the MIT, since Berlekamp mentions the bound and the main properties of the function $\Esp(R)$ in an overview of the known results on the reliability function in his 1964 PhD thesis \cite[Ch. 1: \emph{Historical Backgound}]{berlekamp-thesis}, with references only to \cite{shannon-1959-gaussian}, \cite{fano-book}, \cite{gallager-1965} ``and others''. Gallager, as well, mentions the bound in his 1965 paper \cite{gallager-1965} attributing the ``statement'' of the bound to Fano (see Section I and eq. (44) therein). Note also that Fano and Gallager do \emph{not} call it ``sphere packing bound'' while Berlekamp does in his thesis.
It is worth pointing out that many results were not published by their authors at that time, as happened for example with the Elias bound for the minimum distance of binary codes, which is described in \cite{shannon-gallager-berlekamp-1967-2}. Analogously, in the introduction to \cite[Ch. 9]{fano-book}, Fano credits Shannon for previous derivation, in some unpublished notes, of parts of the results therein. It is known that Shannon was still very productive during the '60s \cite{berlekamp-newsletter-1998} and had many unpublished results on discrete memoryless channels \cite{berlekamp-personal-2012}. It is therefore not immediately clear which parts of the ideas used in the proof of the sphere packing bound were already known among MIT's information theorists. 
Furthermore even if the resulting expression was stated first by Fano, finding the rigorous proof required a reconsideration of Elias' original work for the binary case \cite{gallager-personal-2012}.

Finally, an important comment concerns the bound for constant composition codes with non-optimal composition. It is worth pointing out that, while Fano's version of the sphere packing bound includes the correct tight expression for the case of fixed composition codes with general non-optimal composition \cite{csiszar-korner-book}, the version given in \cite{shannon-gallager-berlekamp-1967-1} does not consider this case, and the bound is tight only for the optimal composition. The reader may also note that it is not even possible to simply remove the optimization over $P$ in that bound, using $E_0(\rho,P)$ in place of $E_0(\rho)$, since it can be proved that constant composition codes with  non optimal composition $P$ achieve an exponent strictly larger than $E_0(\rho,P)$ at those rates where the maximizing  $\rho$ in the definition of $\Esp(R)$ is less than one \cite{csiszar-korner-book}.

\end{document}